\documentclass[a4paper,twocolumn,11pt,accepted=2023-05-10]{quantumarticle}
\pdfoutput=1

\usepackage[utf8]{inputenc}

\usepackage{dcolumn}
\usepackage{bm}

\usepackage{float}
\makeatletter
\let\newfloat\newfloat@ltx
\makeatother

\usepackage[numbers,sort&compress]{natbib}
\usepackage{xcolor}
\usepackage{hyperref}
\usepackage{amsmath}
\usepackage{amssymb}
\usepackage{enumerate}
\usepackage{lipsum}
\usepackage{xpatch}
\usepackage{graphicx}
\usepackage{tabularx}
\usepackage{braket}
\usepackage{amsthm}
\usepackage{dirtytalk}
\usepackage{tikz}
\usepackage{commath}
\usepackage{qcircuit}
\usepackage{algorithm}
\usepackage{algpseudocode}
\usepackage[title]{appendix}
\usepackage{mathtools}
\usepackage{multirow}
\usepackage{subcaption}
\hypersetup{
    colorlinks=true,
    linkcolor=blue,
    filecolor=magenta,      
    urlcolor=cyan,
}

\newcounter{protocol}
\makeatletter
\newenvironment{protocol}[1][htb]{%
  \let\c@algorithm\c@protocol
  \renewcommand{\ALG@name}{Protocol}
  \begin{algorithm}[#1]%
  }{\end{algorithm}
}
\makeatother

\renewcommand{\P}{\mathcal{P}}
\newcommand{\R}{\mathcal{R}}
\newcommand{\X}{\mathcal{X}}
\newcommand{\Y}{\mathcal{Y}}
\newcommand{\A}{\mathcal{A}}

\newcommand{\Aad}{\mathcal{A}_{ad}}
\newcommand{\Ana}{\mathcal{A}_{weak}}
\newcommand{\E}{\mathcal{E}}
\newcommand{\h}{\mathcal{H}}
\newcommand{\inp}{\mathrm{in}}
\newcommand{\out}{\mathrm{out}}
\newcommand{\tr}{\text{Tr}}

\newcommand{\C}{\mathcal{C}}

\newcommand{\V}{\mathcal{V}}
\newcommand{\rin}{\rho^{\inp}}
\newcommand{\rout}{\rho^{\out}}
\newcommand{\rc}{\rho^{c}}
\newcommand{\rr}{\rho^{r}}
\newcommand{\sigr}{\sigma^{r}}
\newcommand{\Hilin}{\mathcal{H}^{d_{\inp}}}
\newcommand{\Hilout}{\mathcal{H}^{d_{\out}}}
\newcommand{\ver}{\texttt{Ver}}
\newcommand{\Gn}{\mathcal{G}^{PUF}}

\newcommand{\Gnre}{\mathcal{G}_{re}}
\newcommand{\Gea}{\mathcal{G}^{\E_{f_1}}}
\newcommand{\Geb}{\mathcal{G}^{\E_{f_2}}}

\newcommand{\Gel}{\mathcal{G}^{\E^L_f}}
\newcommand{\mbraket}[2]{\bra{#1}#2\rangle}
\newcommand{\din}{\mathcal{D}^{in}}
\newcommand{\dout}{\mathcal{D}^{out}}

\newcommand{\forge}{\text{forge}}
\newcommand{\success}{\text{success}}
\newcommand{\fail}{\text{fail}}
\newcommand{\dist}{\text{dist}}
\newcommand{\guess}{\text{guess}}
\newcommand{\classical}{\text{classical}}
\newcommand{\quantum}{\text{quantum}}
\newcommand{\extract}{\text{extract}}

\newtheorem{theorem}{Theorem}
\newtheorem{lemma}{Lemma}
\newtheorem{corollary}{Corollary}
\newtheorem{definition}{Definition}
\newtheorem{requirement}{Requirement}
\newtheorem{game}{Game}
\newtheorem{construction}{Construction}

\usetikzlibrary{matrix,shapes,arrows,positioning,chains,arrows.meta,calc,fit}

\newcommand{\yao}[1]{\textcolor{black}{#1}}

\begin{document}
%
\title{Quantum Lock: A Provable Quantum Communication Advantage}
%
%

\author{Kaushik Chakraborty}
\affiliation{School of Informatics, University of Edinburgh, Edinburgh, UK}
\author{Mina Doosti}
\affiliation{School of Informatics, University of Edinburgh, Edinburgh, UK}
\author{Yao Ma} 
\email{yao.ma@lip6.fr}
\orcid{0000-0003-1185-3431}
\thanks{This work has been presented at QCrypt 2022 (12th International Conference on Quantum Cryptography)}
\affiliation{Laboratoire d’Informatique de Paris 6 (LIP6), Sorbonne Université, Paris, France}
\author{Chirag Wadhwa}
\affiliation{Indian Institute of Technology Roorkee, India}
\author{Myrto Arapinis}
\affiliation{School of Informatics, University of Edinburgh, Edinburgh, UK}
\author{Elham Kashefi}
\affiliation{School of Informatics, University of Edinburgh, Edinburgh, UK}
\affiliation{Laboratoire d’Informatique de Paris 6 (LIP6), Sorbonne Université, Paris, France}

%


\begin{abstract}
Physical unclonable functions (PUFs) provide a unique fingerprint to a physical entity by exploiting the inherent physical randomness. In the review paper [Nature Electronics, 2020] on PUF technology, Gao \emph{et al.} discussed the vulnerability of most current-day PUFs to sophisticated machine learning-based attacks, highlighting the design of provably secure PUF as an important open problem. By encoding the outcome of the classical PUFs into qubits, we address this problem. Specifically, this paper proposes a generic design of provably secure PUFs, called \emph{hybrid locked PUFs} (HLPUFs), providing a practical solution for securing classical PUFs. An HLPUF uses a classical PUF (CPUF) and encodes the output into non-orthogonal quantum states (namely BB84 states, which are widely used for quantum key distribution) to hide the outcomes of the underlying CPUF from any adversary. Similar to the classical \emph{lockdown} technique [TMSCS, 2016], here we introduce a \emph{quantum lock}, to protect the HLPUFs from any general adversaries. The indistinguishability property of the non-orthogonal quantum states, together with the quantum lockdown technique, prevent the adversary from accessing the outcome of the CPUFs. We show that, for quantum polynomial-time adversaries, the ratio between the forging probabilities of the HLPUF, and the underlying CPUF is upper bounded by the distinguishing probabilities of those non-orthogonal states that decay exponentially in the number of output bits of the CPUF. Moreover, we show that by exploiting non-classical properties of quantum states, the HLPUF allows the server to reuse the challenge-response pairs for further client authentication. This result provides an efficient solution for running PUF-based client authentication for an extended period while maintaining a small-sized challenge-response pairs database on the server side. Later, we support our theoretical contributions by instantiating the HLPUFs design using accessible real-world CPUFs, called XOR-PUFs. We use the optimal classical machine-learning attacks to forge both the CPUFs and HLPUFs, and we certify the security gap in our numerical simulation for HLPUF construction, which is ready for implementation.

\end{abstract}
\maketitle  
%


\section{Introduction}

The recent advances in the development of quantum internet and both short-distance and long-distance quantum networks have enabled a broad range of applications from simple secure communication to advanced functionalities such as delegated quantum computation. Many of these applications are out of reach for classical networks \cite{broadbent_quantum_2016,wehner_quantum_2018,fitzsimons_private_2017,veriqloud_quantum_2019, pirandola_advances_2020,diamanti_demonstrating_2019,dynes_cambridge_2019,kozlowski_designing_2020,caleffi_quantum_2018,cacciapuoti_quantum_2020,unruh_everlasting_2013, BBD21Quantum}. Nevertheless, the search for other useful and implementable applications of quantum internet, and quantum communication networks in general, is a very active area of research. A common essential security feature for most such applications is the ability of secure authentication. In general, \emph{authentication} is captured by different definitions and security levels and plays a central role in establishing secure communications over untrusted channels \cite{alagic_quantum_2017, dulek_secure_2020, boneh_quantum-secure_2013}. In particular, \emph{entity authentication}, also known as \emph{device authentication} is a crucial, fundamental, and yet challenging and mostly unsolved task \cite{kang_controlled_2018,gollmann_what_1996}. This sets authentication as a good candidate for practical applications of quantum networks.

Among various approaches for authentication, hardware security provides a promising paradigm for solving such problems by exploiting the underlying properties of hardware and physical devices. In this context, Physically Unclonable Functions (PUF) are a full of potential technology that can establish trust in embedded systems without requiring any \emph{non-volatile memory} (NVM) \cite{GDD02,GCDD02,LLGS04}. A PUF derives unique volatile secret keys on the fly by exploiting the inherent random variations introduced by the manufacturing processes of the \emph{integrated circuits} (ICs). Any slight (yet unavoidable and uncontrollable) variation in the manufacturing process produces a different PUF, rendering the fabrication of an identical physical ‘clone’ of a PUF \cite{RS12} infeasible. Hence, PUFs provide copy-proof, cost-efficient unique hardware fingerprints. Usually, one can generate such fingerprints just by querying the PUF physically. In the literature, we refer to the query and response pairs as \emph{challenge-response pairs} (CRPs). Due to the uniqueness of these devices, different PUFs generate different CRPs.

The literature on classical PUFs (CPUFs) is rich, and there is a multitude of constructions available based on different hardware technologies \cite{GCDD02,GSST07,KL18}. We refer to \cite{Roel12} for a detailed review of the available constructions of classical PUFs. Although all of those constructions provide unique and inexpensive hardware fingerprints, they all suffer from providing sufficient randomness. As a result, most of the existing CPUF constructions are vulnerable against machine learning modelling-based attacks \cite{CHES:Becker15,becker_pitfalls_2014,Del19,RSS13,ruhrmair_modeling_2010}.  In these types of attacks, the attacker first collects a sufficient number of CRPs by adaptively querying the PUF and then uses that data to derive a numerical model using the tools from machine learning. Here, the goal of the model is to predict the response of the PUF to an arbitrary challenge. These attacks open multiple new research directions on designing machine learning-based attack-resilient PUFs \cite{NDJM19,SMCN17}. In the classical domain, there are a few proposals to prevent such sophisticated attacks. The \emph{lockdown technique} \cite{yu_lockdown_2016} is one such example. Informally speaking, it provides a two-way, i.e., server-client authentication. Here, the server first sends a part of the response along with a challenge to the client. The client first checks whether the sent partial response is consistent with the actual response from the PUF corresponding to the challenge that is sent by the server. The client replies with the rest of the response if the server passes this test. Though it prevents the adversary from querying the CPUF in an adaptive manner. However, all of such solutions are heuristic in nature, and none of them provides \emph{provable security} for CPUFs or PUF-based authentication protocols. On the other hand, in recent years, there has been a line of research suggesting to exploit quantum mechanical features of certain devices to design secure PUFs, known in the literature as quantum PUFs~\cite{ADDK19,KMK21,GGDF22Comparison}. Although these proposals provide provable security against quantum machine learning attacks, they are challenging to realise with current-day quantum technologies.

In this work, we introduce a new use-case of quantum communication with provable advantages in several aspects: A new PUF construction and a novel quantum entity authentication protocol that exploits the combination of hardware assumptions and quantum information to achieve secure authentication with provable exponential security advantage compared to its classical counterparts. We also formally prove that the protocol fulfills a specific desired property, namely, \emph{challenge reusability}, which is impossible unless using quantum communication, emphasizing the significance of quantum communication technology and quantum network for a new quantum security era. Moreover, we show that quantum communication makes our construction \emph{cheat-sensitive}, i.e., our PUF-based authentication protocol can detect the adversarial attempts (both passive and active) on intercepting the responses of the PUF. We aim to keep our construction implementable using present-day quantum communication technologies while exploiting the desirable security promises that are provided due to the quantum nature of the challenges and responses. Our PUF construction utilises classical PUFs, which are too weak to be useful in a standalone manner, but present the advantage of being widely accessible and easy to use, and enhances their security using commercially available tools from quantum communication. Here for the first time, we show that by encoding the output of classical PUFs into non-orthogonal qubits, one can enhance the security of PUFs against weak (non-adaptive) adversaries. As such, the first building block of our design is a construction we refer to as \emph{hybrid PUFs} (HPUFs), which encompasses a classical PUF and produces quantum responses for classical challenges. We prove that this construction provides security against the mentioned adversary. With this gadget at hand, we then introduce a construction that is secure against more powerful adaptive quantum adversaries (the general class of quantum polynomial-time (QPT) adversaries). To this end, we borrow the idea of the classical lockdown technique and, by redefining it in the quantum setting, we present our final construction, namely \emph{hybrid locked PUF} (HLPUF). We show that classical PUFs combined with quantum encoding and the new lockdown toolkit can considerably boost the security of classical PUFs without too much overhead. An important technological improvement compared to previous quantum-enhanced proposals where quantum memory was necessary is that for both HPUFs and HLPUFs, only a classical database of challenge-response pairs needs to be stored on the verifier's side. We formally prove adversarial bounds on the unforgeability of HLPUFs in comparison with the underlying classical PUFs, using rigorous proof techniques from quantum information theory. We also formally prove the security of our HLPUF-based device authentication protocol under realistic assumptions.

 

In addition to our theoretical contributions, to better demonstrate the applicability and strength of our results, we provide simulations for the design of HPUF constructions with underlying silicon CPUFs instantiated by the \emph{pypuf} python-based library~\cite{pypuf}. Furthermore, we simulate machine-learning-based modelling attacks on HLPUFs where an adversary acquires classical challenges and quantum-encoded responses from an HLPUF. Our simulation results assist in demonstrating our theoretical proofs by evidencing the security enhancement from CPUFs to HLPUFs. Another significance of our simulation results is that they certify the practicality, and security of our construction, even beyond the scope of the proven theorems, in a real-world scenario, as the CPUFs used in our simulations are commercially available and not only theoretical models. We also bring forward practical proposals to further improve the quality of such constructions.


Finally, through studying this construction, we will also address a long-standing open problem in the field of PUF-based authentication, which is the reusability of challenge-response pairs stored in the verifier database. One significant drawback of PUF-based authentication protocols is that the server/verifier cannot use the same challenge multiple times to authenticate a client/prover due to man-in-the-middle attacks. Therefore, the server exhausts all the challenges from the database after running several rounds of the authentication protocol. This limitation is unavoidable in any such classical protocols.
However, we show that due to the entropy uncertainty principle in quantum information theory, with our proposed construction, the server can reuse a challenge as long as they can successfully authenticate the client using that challenge in the previous rounds. Our result overcomes this open problem as we prove for the first time the challenge reusability of PUF-based applications. The entropy uncertainty principle also allows the honest server/client to detect any adversarial attempts on extracting information from the response of the HPUFs, providing the cheat sensitivity of our protocol. 

\section{Our Results}
We first present the construction of HLPUFs discussing our theoretical results \emph{w.r.t.} their security and showing a provable method of securing classical PUFs, using quantum communication. This result, as mentioned, provides a novel provable advantage that is only achievable using quantum communication. Then we introduce our HLPUF-based authentication protocol. In addition to discussing the security of the protocol, we also show a unique property of such protocols, namely challenge-reusability, which cannot be realised purely classically under similar assumptions. Lastly, we exhibit our theoretical results in practice through simulations, while stepping even closer to practice by using our construction to secure one of the most commercially available and cheap existing PUFs. 

\subsection{Construction of Hybrid Locked PUF (HLPUF)}
The core idea of our HLPUF construction is to hide the outcome of the classical PUF inside quantum states and prevent the adversary from implementing adaptive strategies or getting multiple copies of output quantum states using the quantum lock. The underlying component of our construction is a gadget which we name \emph{Hybrid PUF} (HPUF). HPUF is the part that protects the output interface of the classical PUF by encoding the classical outcomes in non-orthogonal states. Thus, an HPUF is a device with a classical bit-string as input and encoded quantum states as output.

\subsubsection{Construction of the HPUFs}
The construction uses a classical PUF $f:\{0,1\}^n \rightarrow \{0,1\}^{2m}$ with $2m$ outcome bits. From the $2m$ output bits, we construct $m$-pairs of bits. One example of such construction is to take the ($2j-1$)-th, and the $2j$-th (where $1\leq j \leq m$) output bits, and make a pair. Next, we define a two-to-one mapping of the tuple $(y_{2j-1}, y_{2j})$ (where $1\leq j \leq m$) of $f$'s outcome to a qubit $|\psi^{j}_{\out}\rangle \in \{|0\rangle, |1\rangle , |+\rangle, |-\rangle\}$. Here, $|+\rangle = \frac{1}{\sqrt{2}} (|0\rangle +|1\rangle)$, and $|-\rangle = \frac{1}{\sqrt{2}} (|0\rangle -|1\rangle)$. 
Therefore, the HPUF receives a classical query and produces a quantum state as a response. Figure \ref{fig:hpuf} illustrates the HPUF construction. For a more formal description of the construction, we refer to Construction~\ref{cons:hpuf_1} in the supplementary materials. 

\begin{figure}[h]
\includegraphics[width=0.49\textwidth]{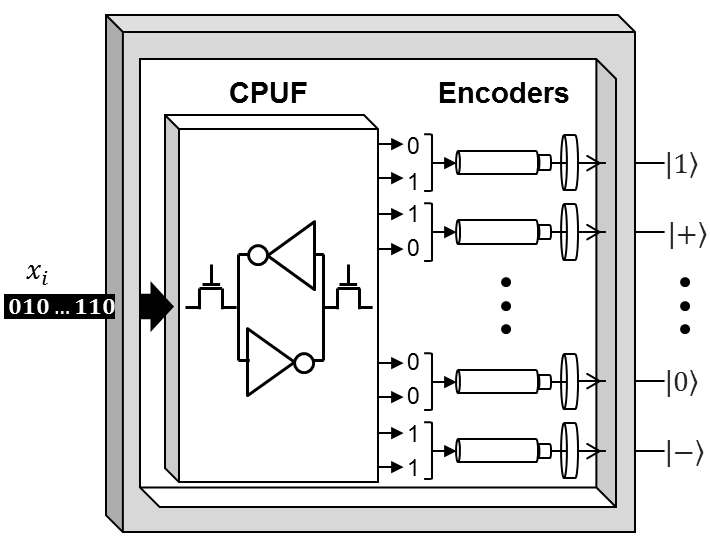}
\caption{HPUF Construction with Conjugate Coding}
\label{fig:hpuf}
\end{figure}

Intuitively, to forge the HPUF, the adversary needs to extract the classical outcome of each challenge from a series of quantum states produced by the HPUF. The task reduces to extracting information on all the two-bit outcomes of the classical PUF (say $(2j-1,2j)$-th bits) from each quantum state $|\psi^j_{\out}\rangle$. Thus the adversary needs to distinguish between four non-orthogonal states $ \{\ket{0}, \ket{1}, \ket{+}, \ket{-}\}$, which is possible with a probability at most $p_{\guess}$. Distinguishing an unknown non-orthogonal quantum state from a pre-determined set of the state is a well-known problem in quantum information which we exploit here in a more general way to introduce extra randomness on the adversary's extracted database of the underlying classical PUF.

Thus, an adversary trying to break HPUF is forced to run its forgery algorithm based on an imperfect training database. The adversarial model considered here assumes that the adversary gets access to a random set of these classical challenges and quantum responses, where there exist only one copy of each pair in the adversary's database. This model is usually referred to as \emph{weak adversary}. We later upgrade this adversary into a more powerful one, which is our target most powerful quantum adversary of interest, when introducing the locking mechanism of the construction.

Due to the probabilistic nature of this extraction process, the extra randomness, captured by probability $p_{\guess}$, enhances the security of the HPUF against weak quantum adversaries as they require considerably more challenge-response pairs to forge the HPUF. We refer to this specific forgery attack as \emph{measure-then-forge} strategy.
This attack is illustrated in Figure~\ref{fig:measure_forge_attack}. Our first result in Lemma \ref{lem:database_dest1} (see supplementary materials) shows that measure-then-forge is an optimal forging strategy for this problem. 

Given a set of $q$ random classical challenge and quantum response pairs, the adversary needs to extract \emph{enough} classical information to forge the HPUF with the most optimal forging algorithm. We assume that for a successful forgery, the adversary needs to extract the outcome of the CPUF from at least $(1-\varepsilon)q$ responses, where $0\leq \varepsilon \leq 1$. The value of $\varepsilon$ depends mainly on the quality of the CPUF and the noise tolerance of the machine-learning algorithm. The calculation of the $\varepsilon$ parameter is discussed in Section \ref{sec:disc}.

\onecolumngrid
\begin{figure}[ht!]
\begin{minipage}{0.99\textwidth}
\centering
\includegraphics[width=0.88\textwidth]{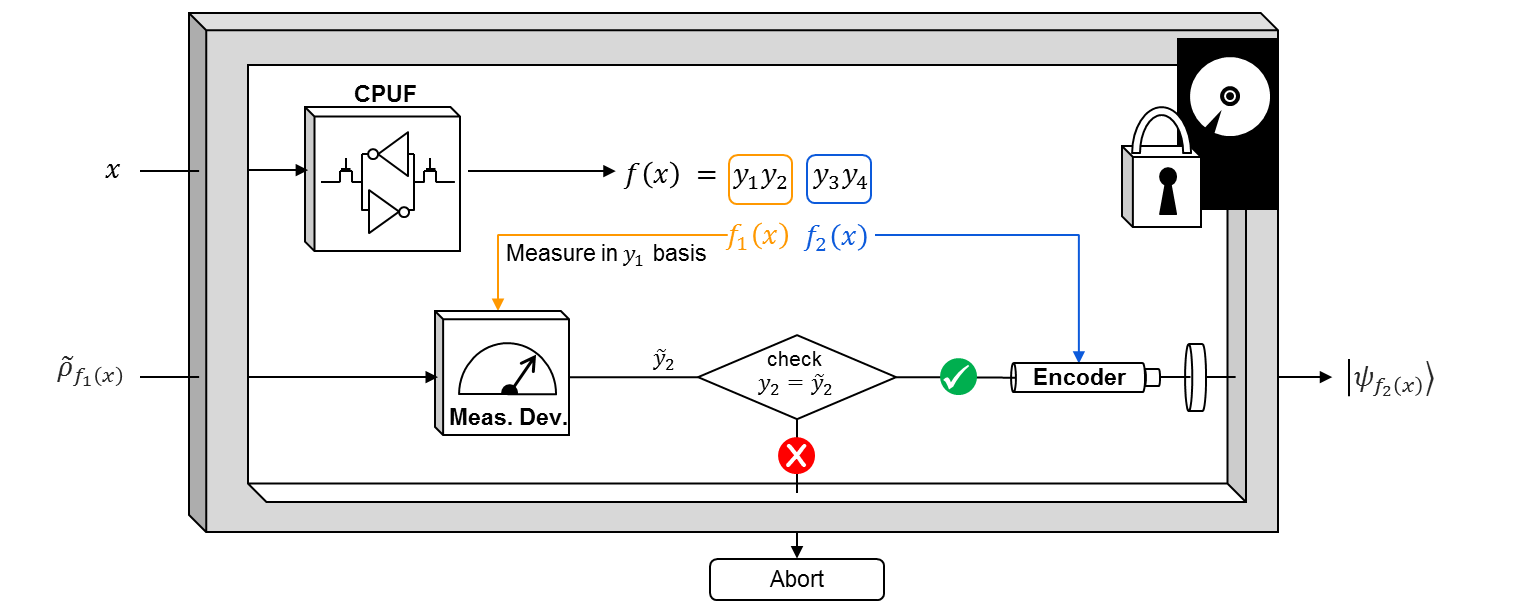}
\caption{HLPUF Construction. HLPUF uses an HPUF, a single-qubit quantum encoder device and a single-qubit measurement device, all inside a tamper-proof environment which prevents any quantum adversary from adaptively querying the HPUF.}
\label{fig:hlpuf}
\end{minipage}
\end{figure}
\twocolumngrid

To derive one of our central results, i.e. the quantum advantage brought by the HPUF construction, we prove an exponential gap between the success probabilities of optimal forgery attack on CPUF and HPUF. Let $P^{\classical}_{\forge}$ denote the probability of forging the CPUF using $q$ challenge-response pairs from the CPUF. We derive the following result, which is formally presented in Theorem \ref{thm:hpuf_prob_2} and Lemma \ref{lem:pguess}.

\noindent \textbf{Security Result 1}: The forging probability of HPUF, denoted as $p^{\quantum}_{\forge}$, is upper bounded by the following quantity.
\begin{equation}
    \label{eq:pril_pforge}
    p^{\quantum}_{\forge} \leq p_{\extract} \times P^{\classical}_{\forge}.
\end{equation}
where the $p_{\extract}$ probability itself is bounded as:
\begin{equation}
    \label{eq:pril_pext}
    p_{\extract} \leq \sum_{k=(1-\varepsilon)q}^{q}\binom{q}{k} (p_{\guess})^{2mk}(1-(p_{\guess})^{2m})^{q-k}, 
\end{equation}
Here, $p_{\guess}$ is the probability of guessing a single-bit outcome of the CPUF from a single-qubit outcome of the HPUF. As an important remark, we note that the classical forging probability,  $P^{\classical}_{\forge}$, is not a small value, given that the CPUF can be broken with a large enough number of queries. Therefore, the term $p_{\extract}$ is responsible for the exponential gap between the security of HPUF and CPUF and consequently highlights the role of quantum encoding in deriving this quantum advantage result. 

The $p_{\guess}$ probability itself is upper bounded (calculated in Lemma \ref{lem:guess_prob}) as follows as a function of a parameter $0.5 \leq p \leq 1$, which quantifies the randomness of the underlying CPUF. 
\begin{equation}
    \label{eq:pril_pguess}
    p_{guess} \leq p(1+\sqrt{2}p),
\end{equation}

If $p_{\extract}$ decays exponentially with the number of output bits of the HPUF, i.e., $m$ then $p_{\forge}^{\quantum}$ would be exponentially smaller than the success probability of CPUF forgery $p_{\forge}^{\classical}$. One can observe that for a smaller value of $\varepsilon$ (See Figure \ref{fig:pextract} in the supplementary materials), $p_{\extract}$ decays exponentially with $m$, showing an exponential separation in the security between the HPUF, and the CPUF. To conclude, we give concrete security bounds for HPUF based on its underlying insecure CPUF.


\subsubsection{Quantum Lock on the HPUFs}
Next, in order to prove the full quantum security of our construction, we need to uplift the previously considered weak adversary into any general \emph{adaptive} quantum adversary. An adaptive quantum adversary is free to build their database with any arbitrary query and in an adaptive manner, potentially depending on the previous queries\footnote{Note that here we don't allow superposition queries to the underlying CPUF inside the HLPUF. However, we allow the adversaries to run quantum algorithms on the challenge-response pair database.}. Particularly such adversaries can query HPUF multiple times with the same challenge $x$, obtaining several copies of $|\psi_{\out}\rangle$ and can easily extract the outcome $f(x)$ from multiple copies. Consequently, a probability $p_{\guess} \approx 1$ can be achieved in theory, and a strong adversary can forge the HPUF efficiently. Hence the construction of HPUFs on its own is not sufficient to achieve the most compelling desired notion of quantum security.


To complete our construction, we equip it with a mechanism called \emph{quantum lock}, which ensures security against general adaptive adversaries. The quantum lock is a mechanism that allows both parties to partially authenticate each other by having access to embedded small verification resources. As a result, it restricts the adversary from adaptively querying the device and reduces a powerful quantum polynomial time (QPT) adversary to a weak adversary.

We start by subdividing the output of the HPUF $\E_f : \{0,1\}^n \rightarrow (\h^{2})^{\otimes 2m}$ corresponding to a classical PUF $f:\{0,1\}^n \rightarrow \{0,1\}^{4m}$ into two different parts, where $\h^d$ denotes a $d$-dimensional Hilbert space of quantum states. The first part contains the first $m$ qubits, and the second half contains the last $m$ qubits of the outcome of the HPUF $\E_f$. Note that the first $m$ qubits of the HPUF's outcome come from the first $2m$ bits outcome of the underlying classical PUF $f$. For any challenge $x \in \{0,1\}^n$ we can write the outcome of the classical PUF as $f(x) = f_1(x)||f_2(x)$, where the mapping $f_1 :\{0,1\}^n \rightarrow \{0,1\}^{2m}$ denotes the first $2m$ bits of $f$ and $f_2 :\{0,1\}^n \rightarrow \{0,1\}^{2m}$ denotes the last $2m$ bits of $f$. Similarly, we can rewrite the HPUF $\E_f$ as a tensor product of two mappings $\E_{f_1}:\{0,1\}^n \rightarrow (\h^{2})^{\otimes m}$, and $\E_{f_2}:\{0,1\}^n \rightarrow (\h^{2})^{\otimes m}$, where for any challenge $x \in \{0,1\}^n$, $\E_{f_1}(x)$ denotes the first $m$ qubits of $\E_f(x)$, and $\E_{f_2}(x)$ denotes the last $m$ qubits of $\E_f(x)$.

The hybrid locked PUF, takes the classical input $x_i$ and a quantum state $\tilde \rho_1$ and produces the second half of the response of the hybrid PUF,  $\ket{\psi_{f_2(x_i)}}\bra{\psi_{f_2(x_i)}}$, as an output if $\tilde \rho_1$ is equal to the first half of the output of the hybrid PUF $\ket{\psi_{f_1(x_i)}}\bra{\psi_{f_1(x_i)}}$. Figure \ref{fig:hlpuf} illustrates the construction of HLPUF.

\onecolumngrid

\begin{figure}[ht!]
\begin{minipage}{0.99\textwidth}
\centering
\includegraphics[scale= 0.35,width=\textwidth]{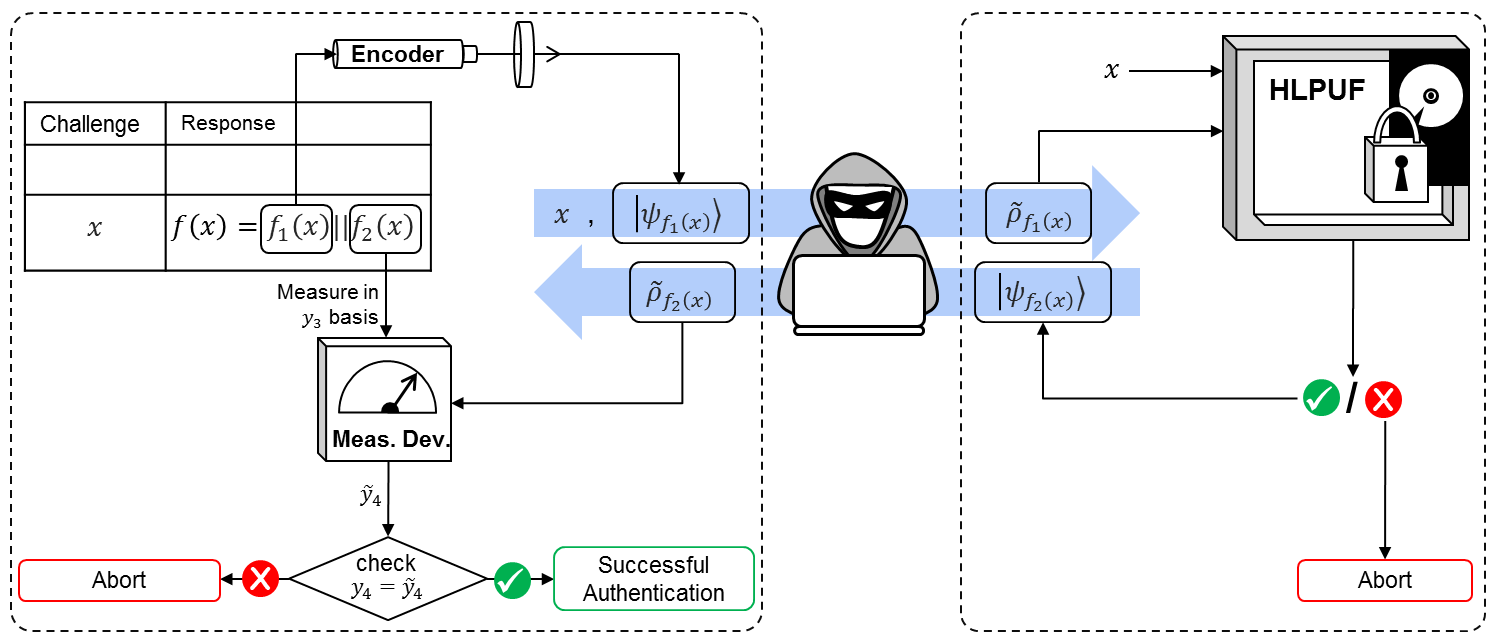}
\caption{HLPUF-based authentication protocol. In each authentication round, the verifier (server) uses a classical database and a quantum encoder to create the required form of challenge for HLPUF which consists of two parts: the classical challenge $x$, and the quantum state $\ket{\psi_{f_1(x)}}$, constructed based of the first half of the classical response, stored in the database. Then the verifier sends them through a public channel fully controlled by a quantum adversary, as illustrated in the figure. The prover (client) then inputs this two-part challenge into the HLPUF and either receives the state $\ket{\psi_{f_2(x)}}$ or gets a reject outcome and aborts the protocol, meaning the message did not come from the authentic verifier. The prover then sends back the quantum state through the same public quantum channel to the verifier, which will verify the client's response by measuring in $y_3$ according to the classical database. Recall that here, $f_2(x) = y_3 y_4$. Also, $\tilde{\rho}_{f_1(x)}$ and $\tilde{\rho}_{f_2(x)}$ denote the real quantum state received by the prover and verifier respectively, after the adversary's interaction with the original states.}
\label{fig:protocol}
\end{minipage}
\end{figure}

\twocolumngrid

Now we prove the promised security for this construction. Note that we assume that the adversary does not have any direct access to the outcome of the embedded classical PUF inside our construction. This assumption can be satisfied by putting the HLPUF inside a tamper-proof box. Thus under the assumption that the adaptive adversary has only access to the input/output ports of the HLPUF, we prove the security of our HLPUF construction, presented in the following informal theorem (The formal result and its proof can be found in Theorem \ref{thm:HLPUF_auth} in the supplementary materials).

\noindent \textbf{Security result 2 }\textit{(Informal)}: \textit{Suppose there is a HLPUF $\E^L_f$ which is made out of an HPUF $\E_f = \E_{f_1} \otimes \E_{f_2}$. If both $\E_{f_1}$ and $\E_{f_2}$ are secure against $q$-query weak adversaries then the HLPUF $\E^L_f$ is secure against any $q$-query adaptive adversaries.}

Intuitively, if an adversary tries to query the HLPUF with any arbitrary challenge $x$, then they need to produce a correct quantum state $|\psi_{f_1(x)}\rangle$, otherwise, the verification procedure inside the HLPUF fails, and the HLPUF replies with a garbage output $\perp$. The inability of the adversary to produce the outcome $|\psi_{f_1(x)}\rangle$ is itself insured via the unforgeability of the HPUF construction and the \emph{no-cloning} principle of the quantum states.


The only remaining option for the adaptive adversary would be to intercept the challenges sent by the server in the previous rounds and use them to query the HLPUF. Therefore practically, with the same challenge $x$ they can query the HLPUF only once. Given that the server chooses the challenges uniformly at random from its database, the adversary querying the HLPUF with those challenges will reduce their power to a weak adversary. As we showed the security of $\E_{f_1}$, and $\E_{f_2}$ against the $q$-query weak adversaries, with the proposed construction, the HLPUF remains secure against any $q$-query adaptive adversaries.

\subsection{HLPUF-based Authentication Protocol}

Putting our construction into practice, we propose an HLPUF-based authentication protocol. Figure \ref{fig:protocol} gives an illustration of the protocol and the formal description of the protocol is given in the supplementary materials. In a nutshell, the verifier (server) sends a challenge that consists of a classical part and a quantum state that will be verified on the prover's (client's) end when queried to the HLPUF device. If the verifier is successfully authenticated by the HLPUF, it produces the quantum response and sends it back to the verifier which can use it to authenticate the prover.

In Further, we formally prove the completeness and security of our protocol against adaptive quantum adversaries in the supplementary materials.

\noindent \textbf{Security Result 3}: The HLPUF-based authentication protocol shown in Figure~\ref{fig:protocol} is complete and secure (universally unforgeable) against any polynomial-time adaptive quantum adversary, given that an HLPUF is used according to the Construction~\ref{cons:hlpuf}, and all the assumptions for the construction are satisfied.

\subsection{Challenge Reusability and Cheat-Sensitivity}

In classical PUF-based authentication protocols, each challenge can be used only in a single authentication round due to man-in-the-middle attacks. The problem arises since the adversary can simply copy and record the challenges and responses and have a perfect copy of the challenger's database, which later can be used to falsely identify themselves. Therefore, the server needs to store an enormous database for running the authentication protocol for a long period. This is a fundamental limitation of classical PUFs~\cite{SD07,HYKD14}.

However, we show that HLPUFs provide an efficient and unique solution to this issue by exploiting the unclonability of the quantum states and the existence of uncertainty relations in quantum mechanics and quantum information. It allows the use of the same challenge several times for authentication without any security compromise. More precisely, each challenge-response pair can be reused under the circumstance of previous successful authentication rounds. This solution will resolve the important practical limitation of the challenger storing a big database or renewing the database of challenge responses frequently.

First, we clarify the condition under which the challenge can be reused. It is a straightforward observation that the challenges for which the verification test has failed should never be used again. A trivial attack, in this case, would be that the adversary intercepts the communication and stores the response state, and later when the same challenge has been queried again, will re-send the stored correct response state to pass the verification. As a result, all the challenges in the failed rounds should be discarded.

Nonetheless, one of our main results is to show that in the event of successful authentication, the challenges can be reused. Here, by successful authentication, we mean that the received response state passes the verification on the client and server side, and both are identified as honest parties. Even though the events of false identification of an adversary is still possible (for example, if the challenge is the same as one of the challenges that previously existed in the adversary's local database), our result, stated as follows, ensures that these events occur only with negligible probability. 

\noindent \textbf{Security result 4 }\textit{(Informal)}: If the HLPUF-based authentication protocol (Figure~\ref{fig:protocol}) doesn't abort for a specific challenge $x$, then the probability of the adversary successfully extracting the classical outcome of the PUF is upper bounded by $2^{-m}$. Therefore, the challenge $x$ can be reused.

This is an influential information-theoretic result that shows even in the presence of a powerful quantum adversary, if the challenge-response pair of HLPUF leads to successful authentication of the honest parties then the adversary has almost no information about the response $f(x)$ of the underlying CPUF $f$. We also show that using the same challenge for $k$ times, if the authentication is passed for all of them, the probability that the adversary successfully extracts the classical outcome of the PUF is upper bounded by $k2^{-m}$, which quantifies further this reusability feature. The results have been formally shown in Theorems \ref{th:uncertainty} and  \ref{thm:chall_reuse_mult} in the supplementary materials. This feature is uniquely been enabled due to quantum communication and the specific relation between the quantum states that we use for our encoding. Our results have been proven using a sophisticated toolkit in quantum information theory, namely, entropic uncertainty relations~\cite{deutsch1983uncertainty,coles2017entropic}, which have also been used for the full security proof of famous quantum protocols such as QKD.

Another relevant feature that our quantum communication-based solution provides is cheat sensitivity, meaning that due to the discussed quantum properties of our CRPs, a passive adversary trying to intercept and hijack the communication will be detected.  






\subsection{Our Theoretical Results in Practice: HLPUF's Resiliency to the Machine Learning-Based Attacks}
We validate and showcase the practicality of our theoretical results for HLPUF construction using numerical results and simulations. While introducing HPUF and our security results earlier, we gave a theoretical upper bound on the forging probability of HPUF. Our theoretical security analysis shows that exponential security can be achieved for this construction, relying on certain reasonable assumptions, including the existence of a classical PUF that is not broken with probability $1$, nonetheless is breakable with non-negligible probability given enough queries. Although such mid-level classical PUFs can be theoretically found, especially in optical-based constructions, we focus on putting our construction to into test using the cheapest and most widely available CPUFs. We choose silicon CPUFs such as arbiter PUFs for this purpose, which are known to be weak in security and breakable using machine-learning attacks. We compare the performance of these CPUFs with an HPUF that is constructed with the same underlying CPUF, performing measure-then-forge attacks using classical machine-learning algorithms (see Figure~\ref{fig:measure_forge_attack} for the illustration of the attack). The numerical simulation results assist in demonstrating our theoretical proofs by exhibiting an exponential advantage of success probability of HPUF forgery compared to its underlying CPUF with a limited $q$-query.



Here, we instantiate the underlying silicon CPUFs by a python-based library called \emph{pypuf} \cite{pypuf}. From the pypuf library, we consider the XOR Arbiter PUFs \cite{SD07} which are timing-based CMOS PUFs of the form $f:\{0,1\}^n \rightarrow \{0,1\}$. For constructing the HPUFs, we need an underlying CPUF with at least two bits outcome. Therefore, we use two such XOR arbiter PUFs (say $f_1$, and $f_2$) for instantiating an HPUF. For the forgery, we use the measure-then-forge strategy that we define in the HPUF section. As the best measurement strategy for the measure-then-forge attack, we use the upper bound we derived on the adversary's guessing probability of extracting a single-bit outcome of the classical PUF from the outcome of the HPUF (see Lemma \ref{lem:guess_prob} in the supplementary materials). After the measurement phase in the measure-then-forge strategy, the adversary ends up with a classical database. We use the classical \emph{logistic regression} (LR) algorithm for the forgery. Note that, for the $k$-XOR PUFs, the LR attacks show the best performance. Therefore, we use the same algorithm in our measure-then-forge strategy. For a more detailed description of the forgery attack, we refer to Section~\ref{ap:simulation} in the supplementary materials. Our numerical results can be categorised into two main contributions summarized as follows.

\subsubsection{Advantage over CPUFs}
First, our simulation results show a considerable advantage of our construction over CPUFs, even when constructed from the on-the-counter low-cost CPUFs. We summarize our numerical results on the advantage of HLPUF over the underlying CPUF in Figures~\ref{fig:simulation_k4} and \ref{fig:simulation_k5}. On each of the plots in these figures, the $X$-axis denotes the number of CRPs we use for the forgery, and the $Y$-axis denotes the accuracy of the forgery. The blue curves in each sub-figure represent the forging accuracy of the underlying CPUF. The red curves denote the forging accuracy of the HPUF against the general adaptive adversary, and the green curves denote the forging accuracy of the HLPUFs against general adaptive adversaries. From these plots, it is evident that without the quantum lock, the HPUF provides a very small advantage over the underlying CPUF. This implies that quantum communication alone is not sufficient in providing a higher security boost. However, the gap between the blue curve and the green curve in each of the plots of Figures~\ref{fig:simulation_k4} and \ref{fig:simulation_k5}, shows the importance of the quantum lock for providing a much higher security boost.

The simulation results show that if the adversary has enough challenge-response pairs from the HLPUF then eventually it can forge the HLPUF. However, if the adversary tries to forge the HLPUF, then it needs to measure to extract the classical information from the quantum state, i.e., the outcome of the HLPUF. This measurement can disturb the quantum state, and if the measurement is not successful then the authentication also fails. This is something different from the classical scenario, where the adversary can remain undetected and make the forgery. We refer to this property as the \emph{cheat-sensitivity} of the HLPUFs. Due to this property, we can safely use the HLPUFs in practice much more times than the prediction of Figures \ref{fig:simulation_k4} and \ref{fig:simulation_k5}.

\onecolumn

\begin{figure}[t]
    \begin{minipage}{0.99\textwidth}
    \centering
    \includegraphics[scale=0.6]{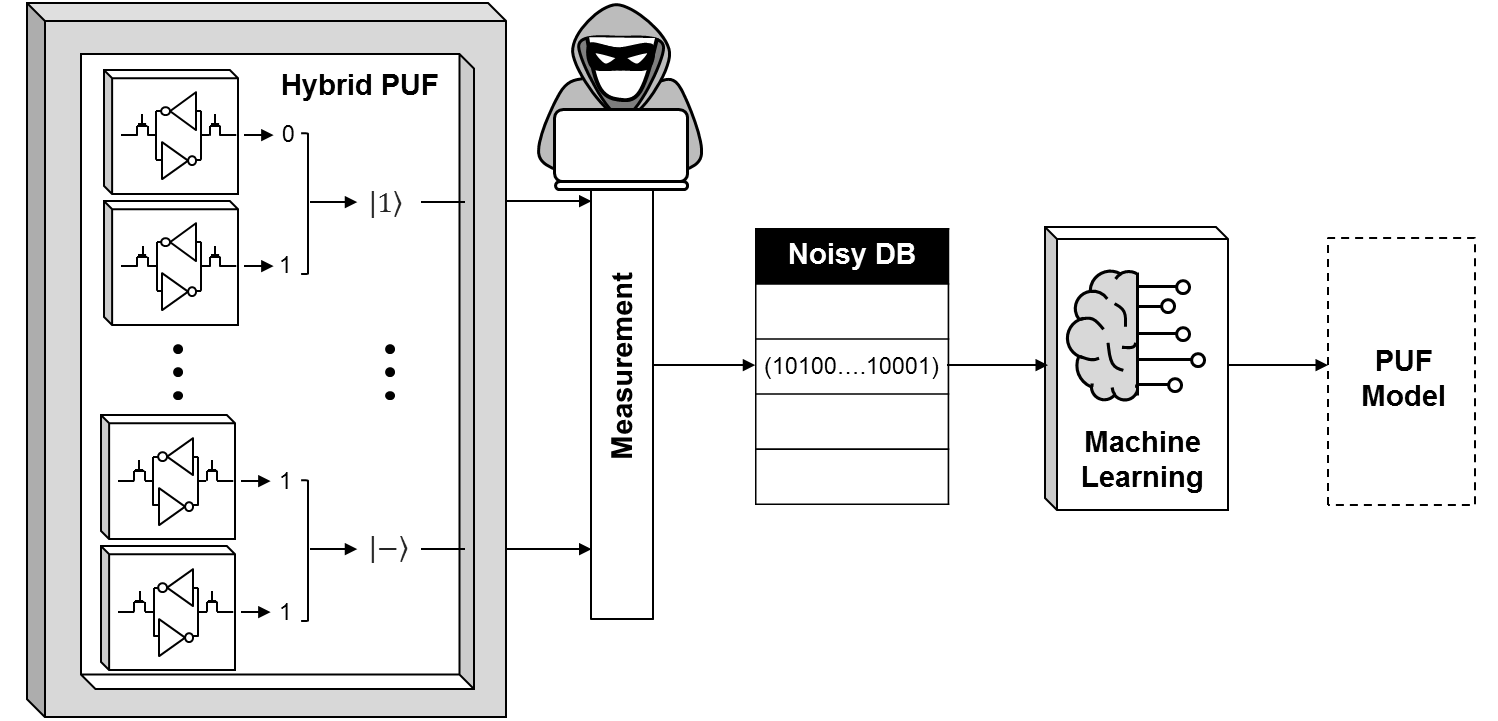}
    \caption{Illustration of the measure-then-forge attack. The quantum adversary receives a sequence of BB84 quantum states as the output of the HPUF and measures them with the optimal measurement strategy to obtain the underlying classical information of the responses of CPUF. Due to the quantum nature of the HPUF responses, even the best measurement strategy is still probabilistic, which leaves the adversary with a noisy version of the classical database. Then the adversary can run a machine-learning attack on the noisy database (in the optimal attack, this classical machine-learning algorithm is assumed to be optimal as well) to extract the mathematical model of the PUF.}\label{fig:measure_forge_attack}
    \end{minipage}
\end{figure}

\onecolumn
\begin{figure}[!h]
    \begin{minipage}{0.99\textwidth}
        \centering
        \subcaptionbox{\label{fig:n64k4}}{\includegraphics[width=0.5\linewidth]{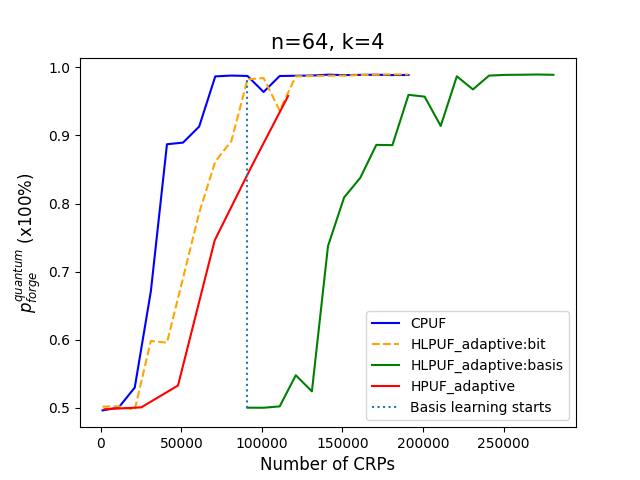}}\hspace*{\fill}
        \subcaptionbox{\label{fig:n128k4}}{\includegraphics[width=0.5\linewidth]{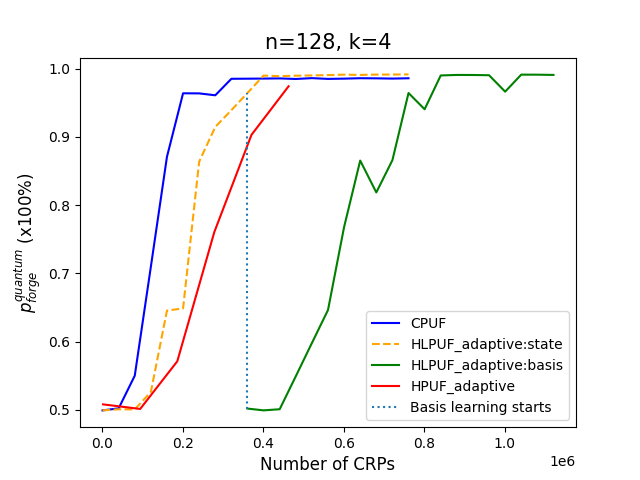}}
        \caption{Evolution of LR attack performance on CPUF(in blue), HPUF(BB84, in red for modelling a qubit), and HLPUF(BB84, in green for modelling a qubit) with different CRPs as the training set while the challenge size is 64 (\ref{fig:n64k4})/128 (\ref{fig:n128k4}) bits with k=4 XORPUFs}\label{fig:simulation_k4}
        \subcaptionbox{\label{fig:n64k5}}{\includegraphics[width=0.5\linewidth]{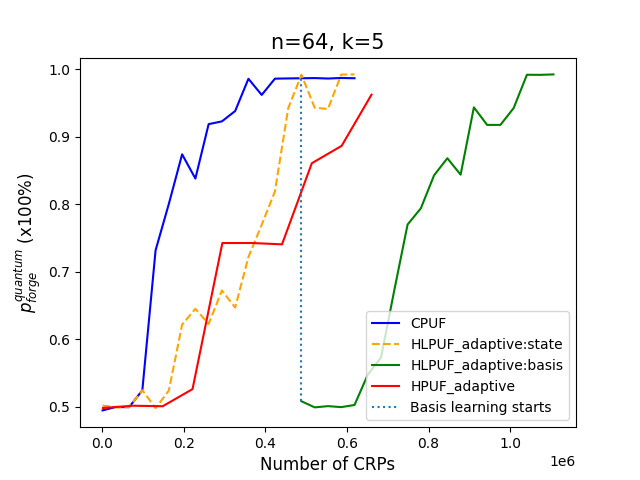}}\hspace*{\fill}
        \subcaptionbox{\label{fig:n128k5}}{\includegraphics[width=0.5\linewidth]{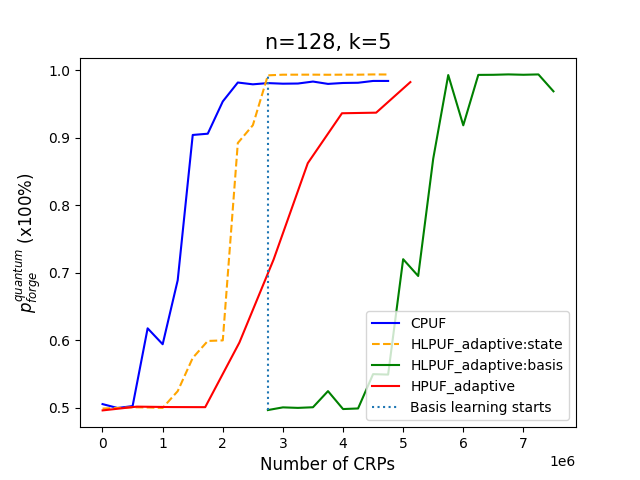}}
        \caption{Evolution of LR attack performance on CPUF(in blue), HPUF(BB84, in red for modelling a qubit), and HLPUF(BB84, in green for modelling a qubit) with different CRPs as the training set while the challenge size is 64 (\ref{fig:n64k5})/128 (\ref{fig:n128k5}) bits with k=5 XORPUFs}\label{fig:simulation_k5}
    \end{minipage}
\end{figure}

\begin{figure}[!h]
    \begin{minipage}{0.98\textwidth}
    \centering
      \subcaptionbox{\label{fig:cmp_cpuf}}{\includegraphics[width=0.5\linewidth]{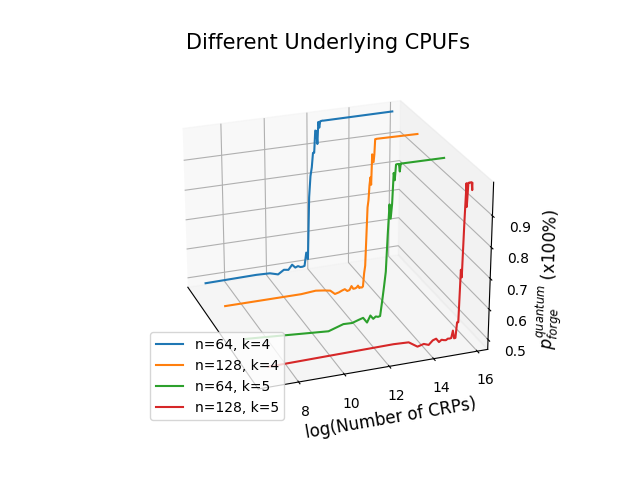}}\hspace*{\fill}
      \subcaptionbox{\label{fig:dim_cpuf}}{\includegraphics[width=0.5\linewidth]{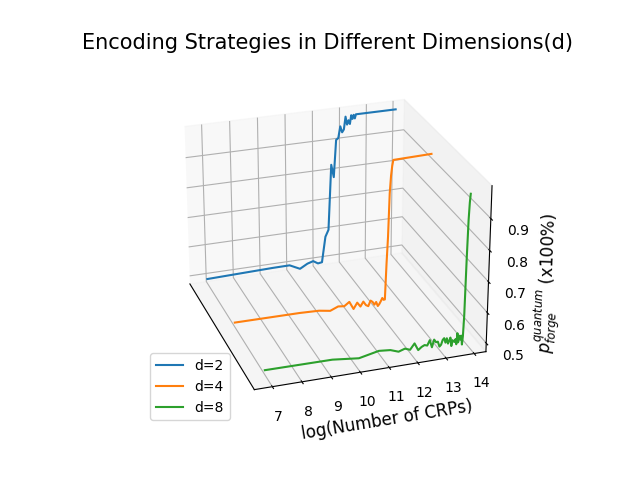}}
      \caption{Comparison of LR attack performance on HLPUFs with different underlying CPUFs (\ref{fig:cmp_cpuf}) and different encodings (\ref{fig:dim_cpuf}) strategies}\label{fig:diff_cpuf_enc}
    \end{minipage}
\end{figure}

\newpage
\twocolumn
\subsubsection{Practical solutions for boosting the security: Better CPUF or better quantum encoding}

In Figure \ref{fig:cmp_cpuf} we observe that if we increase the value of $k$ in the underlying $k$-XOR PUFs, then the adversary requires more challenge-response pairs for a successful forgery. This observation suggests that one possible way to enhance the security of the HLPUFs is to use more secure classical PUFs. Hence, we elaborate on the effect of different $k$-values on the HLPUF forgery. Moreover, the red plot in this figure also suggests that one can improve the security of HLPUFs significantly just by increasing the input size of the HLPUFs. 

We also explore another possible way to improve the security of the HLPUFs. The idea is to use a more sophisticated encoding than encoding two classical bits into a quantum state $\ket{\psi}$ such that $\ket{\psi} \in \{\ket{0}, \ket{1}, \ket{+}, \ket{-}\}$. Here we use the concept of Mutually Unbiased Bases (MUBs) \cite{BBRV02} of dimension $d=4$ or $d=8$ for the encoding. For the dimension $d=4$ ($d=8$), we encode four (six) classical bits to a two (three) qubits quantum state. We describe the encoding procedure in detail in the supplementary materials (see Section \ref{sec:MUB}). Intuitively, the higher dimensional encoding helps to reduce substantially the value of $p_{\guess}$ in the measure-then-forge strategy significantly. For example, in Section \ref{sec:MUB}, for the MUB encoding of dimension $8$, we calculate the value of $p_{\guess} \leq 0.62$.  

\begin{figure}[!ht]
    \centering
    \includegraphics[width=8cm]{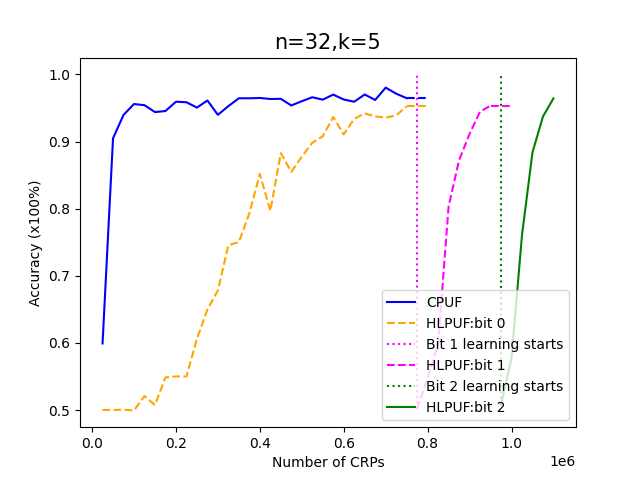}
    \caption{Evolution of LR attack performance of classical (in black) and hybrid (MUB in 8-dimension encoding, in red for modelling 3-qubit) constructions with different CRPs as the training set while the challenge size is 32 bits with k=5 XORPUFs}
    \label{fig:mub_k5} 
\end{figure}

In Figure \ref{fig:dim_cpuf}, we show the impact of this encoding on the forging probability. Specifically, we show an interesting simulation result in Figure \ref{fig:mub_k5}, where we only use $32$-bits input $5$-XOR PUF as an underlying CPUF. For such CPUFs, the total number of possible challenges is $2^{32} \approx 10^9$. In Figure \ref{fig:mub_k5}, we observe that the underlying CPUF can be forged using only $5000$ CRPs. On the other hand, for the forgery of the HLPUFs, the adversary requires almost $10^6$ queries. For the forgery of the HLPUF, the adversary needs to use almost all the CRPs. We can enhance the security of the HLPUFs by using higher-dimensional MUBs.

\section{Discussion}
\label{sec:disc}
In this paper, we proposed a new practical way to enhance the security of PUFs using quantum communication technology and showed a new use case for quantum communication, which benefits from both provability and practicality. We classify the adversaries into adaptive and weak adversaries based on their querying capabilities. This classification is not only useful in the proof reductions but also provides a step-by-step path towards a provably secure PUF against the strongest possible quantum adversaries. By harnessing the power of quantum information theory, here we propose a construction for a hybrid PUF with classical challenge and quantum response. The main idea is to encode the output of classical PUF into non-orthogonal quantum states. We show that for the forgery of the HPUF, any $q$-query weak adversary first needs to extract the classical string $f(x)$ from the outcome of the HPUF. The adversary tries to forge the CPUF using that extracted data. Due to the indistinguishability of the non-orthogonal quantum states, the adversary introduces extra randomness at the outcome of the CPUF, which in turn complicates the forging task for any QPT adversary. We have established the result under the assumption that for a $q$ random outcomes of the HPUF if the distance between the outcomes of CPUF and the extracted outcomes from the HPUF is above a threshold $\varepsilon$ then no QPT adversary can forge the HPUF. Under this assumption, we show that the probability of forging the HPUF is exponentially smaller than forging the CPUF. This is an exponential provable gap which is only achievable via quantum communication.
We also instantiated our HPUF design using real-world CPUF, called XOR-PUFs. In Figure \ref{fig:simulation_k4} and Figure \ref{fig:simulation_k5}, we show the gap in the number of queries the adversary needs to forge the HPUF compared to the underlying CPUF. As displayed in those figures, the probability of the HPUFs being fully broken is considerably small compared to their underlying CPUF. However, using an enormous number of samples, the adversary eventually forges the HPUF, certifying the assumption in our theoretical result. A more sophisticated encoding can enhance this gap. Later in Figure \ref{fig:mub_k5}, we show that the MUB of dimension $8$ encoding of the outcome of the CPUFs can enhance this gap substantially. 

In PUF-based authentication protocols, one important issue (both for classical and quantum PUFs) is that an adaptive adversary can query the PUF with arbitrary input challenges. It permits such an adversary to learn efficiently and emulate the input/output behaviour of the targeted PUF. We solve this problem with our quantum locking mechanism, leading to our HLPUF construction as discussed. In our proposed authentication protocol, we prove the security against adaptive adversaries. The advantage is twofold: On one hand, the probability of knowing information about a quantum state is upper-bounded compared to a classical PUF due to the quantum information theory. On the other hand, the implementation of hybrid PUFs is practical nowadays with the existing quantum communication technology. 

Another advantage of the hybrid locked construction is the reusability of the challenge-response pairs, which was impossible prior to this work for similar protocols. Therefore, with our solution, a server can perform secure client authentication for an extended period without exhausting its CRPs database. This result overcomes the fundamental drawbacks of the existing classical PUF-based authentication protocols while putting forward a novel and practical use case for our HLPUF construction as well as a unique feature enabled solely by quantum communication.


The no-cloning property of quantum states also prevents passive adversaries from intercepting and storing the qubits for forgery without getting detected by the server/client. Unlike the classical setting, quantum communication forces all adversaries to behave like active ones. In general, it is impossible for adversaries to extract information about the outcome of the underlying classical PUFs from the outcome of the HLPUFs without getting detected. This makes our HLPUF protocol cheat-sensitive, providing another advantage over CPUF-based authentication protocols.


The quantum communication part of our HLPUF construction relies on the conjugate coding, which is used in the quantum key distribution (QKD) protocols. QKD technology is one of the most mature quantum technologies. Long-distance QKD networks are already implemented and used in several countries like the USA, UK, China, EU, Japan, \cite{SFIK11,SLB11,PPM08,WCZ14,Court16} etc. Many commercially available QKD infrastructures provide almost $300$kb/s secret key rate over optical fibre links of length $120$km \cite{FBD17}. Moreover, the availability of the mature QKD on-chip technology \cite{SEG17,SSH20,BLL18} makes all the proposed constructions in this paper implementable using existing quantum technology. Our results show that picking off-the-shelf classical PUF technology and QKD technology can partially solve significant shortcomings of the device authentication problem in a quantum network. 

In this work, we show that our HLPUF construction makes the current-day insecure classical PUFs, secure with the help of quantum conjugate coding and lockdown techniques, and against present and future powerful quantum adversaries. However, all of our results are based on ideal implementations of the protocol. The next research direction will be to explore the performance of our HLPUF-based authentication protocol under channel noise and imperfect single-photon sources. Yet another intriguing research direction will be the design of robust variants of our protocol. Like some QKD protocols, our HLPUF becomes vulnerable to photon number splitting attacks if the source suffers from a multi-photon emission problem. Therefore, a further study of the feasibility and practicality of hybrid PUF constructions is an important future direction for bringing this technology from theory to practice.

Another interesting question arises in terms of the engineering design of the HLPUF, where a lockdown technique is exploited to prevent adaptive queries by network adversaries during usage. Explicitly, as a stand-alone construction, HLPUF construction implies a tamper-proof box where the underlying CPUF, as well as the quantum measurement and preparation apparatus, are under protection, except for the locked interface. A relevant question here is how a server can obtain a classical database of HLPUF given such tamper-proof environments. We argue that this is not an issue in the context of our proposed protocol and under the formal assumptions under which the protocol provides security guarantees. Firstly, we note that in the proposed protocols, the manufacturer, the server, and the client are all honest parties, and the construction of the HLPUF can be seen as a recipe for an honest manufacturer/server to construct such mechanisms given a CPUF which is potentially insecure, while followed by our adversarial model, the CPUF should not be queried directly at any point during the protocol. One can reasonably assume that the server first obtains the classical database of underlying CPUF prior to assembling HLPUF construction, then after assembling and sealing the box, transfers it to the client. We emphasise that such considerations will not affect the security guarantees of the protocol as they have been taken into account in our network adversarial model. 

Nonetheless, we also propose an alternative solution that can be implemented at the hardware engineering level to ensure our assumptions are being met while enabling the HLPUF to operate as a stand-alone hardware token, and not just within our given protocol. This can be achieved by integrating a \emph{programmable read-only memory} (PROM) based device inside HLPUF while assembling by the manufacturer. A PROM is a type of non-volatile classical memory chip that permits data to be written in only once after the device's manufacture \cite{harris2010digital, Authur79}. Once PROM is programmed, its content cannot be changed, which means the data is permanent. In practice, a small piece of PROM is needed, with at least 2 registers, to enable the HLPUF device to switch between \emph{setup} and \emph{handover} modes. The mode-switch procedure can be performed as follows: When the manufacturer produces an HLPUF device within a tamper-proof box, the registers of PROM are set to value $11$ as \emph{setup} mode, and it can be queried from outside. Once the mode has been set differently, it can never go back to $11$, which means that HLPUF has been used before in the setup mode. In setup mode, the server  can query the box with classical queries. On the first classical query, the register updates the mode to $01$ internally and will output classical responses, as long as it stays so. After the setup is done, the server can set the value of registers to $00$, in which case the encoding part of the device is activated and the HLPUF will output the quantumly encoded queries i.e., $\ket{\psi_{f(x)}}$. Of course, an adversary can do the same by querying HLPUF classically by setting registers from $11$ to $01$. However, this behaviour can be easily detected and when an honest party (server) receives the box, they will not use the HLPUF box, if it has ever been on a setup mode before. Furthermore, another engineering aspect to be taken is by harnessing device wear-out property to create limited access to the underlying CPUF \cite{DFKC17}. Finally, we note that the most efficient and practical design for such boxes although an interesting engineering problem, is not in the scope of this paper and is a completely distinct direction for future works.

\section{Acknowledgement}
\label{sec:ack}
This work is supported by grants from Région Ile-de-France, as well as Innovate UK funded project called AirQKD: product of a UK industry pipeline, Grant Number
106178



\bibliographystyle{abbrvnat}
\bibliography{biblio}

\newpage
\appendix
\newcommand{\hbAppendixPrefix}{S}
\renewcommand{\thefigure}{\hbAppendixPrefix\arabic{figure}}
\setcounter{figure}{0}
\begin{appendices}
\onecolumngrid
\section{Overview}
In the supplementary materials, we provide all the formal definitions and constructions, security proofs and other detailed technical results. The structure is as follows: First, in Appendix~\ref{ap:prelim} we introduce some of the basic notions and tools from quantum information and PUF literature that we will use later. In Appendix~\ref{ap:adv_model} we give a detailed description of an adaptive and weak quantum adversary, in the most general case of the unforgeability game where all the learning queries are density matrices. Then, we also give a more detailed version of the quantum unforgeability game, with adaptive and weak adversaries. In Appendix~\ref{ap:hpuf-formal} and Appendix~\ref{ap:hlpuf-const-security} we give the formal description of HPUF and HLPUF constructions respectively and then in Appendix~\ref{ap:hlpuf-security}, we present the main results of the paper formally. Then in Appendix~\ref{ap:reusability}, we discuss the challenge reusability result in further detail and in~\ref{ap:ch-reusability} we first give a brief introduction of the entropic uncertainty relations that have been used in the literature of quantum information for different purposes like security proof of QKD protocols. Then, we establish a formal version of Theorem~\ref{th:uncertainty}, in terms of the described uncertainty quantities, and finally, we give a full detailed proof of this theorem which we will use to establish the challenge reusability property for our HLPUF-based protocol.
In Appendix~\ref{ap:simulation} we discuss our simulation results more extensively, discussing also technical details about the effect of different quantum encoding. In Appendix~\ref{ap:limitation} we investigate the lockdown technique on quantum PUFs and we establish a general no-go result.
Finally, in Appendix~\ref{ap:security} we give the full and detailed security proofs for the theorem in Appendix~\ref{ap:hlpuf-security}, including the proof of Theorem~\ref{thm:cpuf_vs_hpuf}, Lemma~\ref{lem:guess_prob}, Lemma \ref{lem:database_dest1}, and Lemma \ref{lem:pguess}.

\section{Preliminaries}\label{ap:prelim}
In this section, we discuss some of the main concepts and definitions that we rely upon in the paper.

\subsection{Quantum information tools}
Quantum states are denoted as unit vectors in a Hilbert space $\h$. Any $d$-dimensional Hilbert space is equipped with a set of $d$ orthonormal bases. We say a quantum state is pure if it deterministically describes a vector in Hilbert space. On the other hand, a mixed quantum state is described as a probability distribution over different pure quantum states,
represented as a density matrix $\rho\in \h^d$. If a quantum state can be written as the tensor product of all its subsystems, we say that the state is \emph{separable}, otherwise, it is referred to as \emph{entangled state}.

If a quantum resource takes an input $\rho_{\inp}\in \h_{A}^{d_{\inp}}$ and produces an output $\rho_{\out}\in \h_{B}^{d_{\out}}$, we use a completely positive and trace preserving (CPTP) map $\E$ to describe the general quantum transformation $\E:\h_{A}^{d_{\inp}}\rightarrow \h_{B}^{d_{\out}}$.

The measurement of a quantum state is defined by a set of operators $\{M_{i}\}$ satisfying $\sum_i M_{i}^{\dagger}M_{i} = I$ with its conjugate transpose operator $M^{\dagger}$. The probability of getting measurement result $i$ on quantum state $\ket{\psi}$ is:
\begin{align*}
    P(i)=\bra{\psi}M_{i}^{\dagger}M_{i}\ket{\psi}=\bra{\psi}M_{i}\ket{\psi}.
\end{align*}
Furthermore, we define the set $\{E_i\}$ a \emph{POVM} (Positive Operator-Valued Measure) with positive operators $E_i=M_{i}^{\dagger}M_{i}$, where $\Sigma_i E_i=I$. 

An important property of the quantum states is the impossibility of creating perfect copies of general unknown quantum states, known as the \emph{no-cloning theorem} \cite{WZ1982}. This is an important limitation imposed by quantum mechanics which is particularly relevant for cryptography. A variation of the same feature makes it impossible to obtain the exact classical description of quantum states by having a single or very few copies, therefore, there exists a bound on how much classical information can be extracted from quantum states, known as Holevo bound~\cite{H1973}. Moreover, distinguishing between two unknown quantum states is also a probabilistic procedure known in the literature of quantum information as \emph{quantum state discrimination}. The distinguishability of the quantum states depends on their distance. There exist several distance measures for quantum states and quantum processes~\cite{nielsen_quantum_2010}, although, for the purpose of this paper, we introduce the fidelity, the trace distance and the diamond norm. The trace distance between two quantum states $\rho$ and $\sigma$ is defined as:
\begin{equation}
    \mathcal{D}_{tr}(\rho, \sigma) = \frac{1}{2} \|\rho - \sigma\|_1 = \frac{1}{2} Tr[\sqrt{(\rho - \sigma)^2}]
\end{equation}
The fidelity of mixed states $\rho$ and $\sigma$ is defined by the Uhlmann fidelity \cite{nielsen_quantum_2010}:
\begin{equation}
    F(\rho,\sigma)=[Tr(\sqrt{{\sqrt{\rho}}\sigma {\sqrt{\rho}}})]^{2}
\end{equation}
which will become $|\mbraket{\psi}{\phi}|^2$ the following expression for two pure quantum states $\ket{\psi}$ ($\rho=\ket{\psi}\bra{\psi}$) and $\ket{\phi}$ ($\sigma=\ket{\phi}\bra{\phi}$). The fidelity is bounded between 0 and 1, $0\leq F(\rho,\sigma)\leq 1$. $F(\rho,\sigma)=0$ when two states $\rho$ and $\sigma$ are orthogonal and $F(\rho,\sigma)=1$ when $\rho$ and $\sigma$ are identical. 

In this paper, we denote all the verification algorithms for checking equality of two quantum states by distance as a CPTP map $\texttt{Ver}: \h^{d} \otimes \h^d \rightarrow \{0,1\}$. For any two states $\rho_1 , \rho_2 \in \h^{d}$, this mapping is defined below.
 
\begin{equation}
    \label{eq:ver}
    \texttt{Ver}(\rho_1,\rho_2) := \begin{cases}
    &  1 ~~~~\text{if }\norm{\rho_1 - \rho_2}_1 \leq \epsilon, \\
    &  0 ~~~~\text{otherwise.}
    \end{cases}
\end{equation}

This general verification also includes measurements of quantum states as verification algorithms since it has been defined as a general CPTP map. Finally, we mention the notion of \emph{SWAP test}~\cite{buhrman_quantum_2001} as a quantum circuit for implementing the verification algorithm $\texttt{Ver(.)}$ above. The swap test's circuit uses the controlled version of a swap gate that swaps the order of two quantum states if the control qubit is $\ket{1}$. The circuit outputs $\ket{0}$ with probability $\frac{1}{2} + \frac{1}{2}F(\ket{\psi}, \ket{\phi})$ and it outputs $\ket{1}$ with probability $\frac{1}{2} - \frac{1}{2}F(\ket{\psi}, \ket{\phi})$. As can be seen, the success probability of this test depends on the fidelity of the states. This occurs because of the quantum nature of these states and measurements in quantum mechanics. 

\subsection{Models for PUF}
A Physical Unclonable Function is a secure hardware cryptographic device that is, by assumption, hard to clone or reproduce.
Here we give the mathematical model for the classical PUFs first, and then we also briefly mention the quantum analogue of them known as quantum PUF (QPUF) as defined in~\cite{ADDK19}. As classical PUFs are usually defined with probabilistic functions, due to their inherent physical randomness, we first define the notion of probabilistic functions as follows.

\begin{definition}[Probabilistic Function]
A \emph{probabilistic function} is a mapping $f: \R \times \X \rightarrow \Y$ with an input space $\X$, an \emph{random coin} space $\R$, and an output space $\Y$. 
\end{definition}

For a fixed input $x \in \X$, and a random coin (or key) 
$R \leftarrow \R$, we define the probability distribution of the output random variable $f(x) := f(R,x)$ over all $ y \in \Y$ as,

\begin{equation}
    p^f_x(y) := \Pr[f(x) = y|x] = \sum_{r:f(r,x) = y} \Pr[R = r].
\end{equation}

A classical PUF can be modelled as a probabilistic function $f:\R \times\X\rightarrow\Y$ where $\X$ is the input space, $\Y$ is the output space of $f$ and $\R$ is the identifier. The creation of a classical PUF is formally expressed by invoking a manufacturing process $f\leftarrow\mathcal{MP}_{C}(\lambda)$, where $\lambda$ is the security parameter. 

To model classical PUF $f$ in terms of security primitives, Armknecht et al. \cite{sako_towards_2016} define some requirements which are parameterized by some threshold $\delta_{i}$ and a negligible function $\epsilon(\lambda)\leq\lambda^{-c}$, where $c>0$ and $\lambda$ is large enough. Note that the requirements in our paper correspond to the requirements of intra and inter distances of PUF $f$.

\begin{definition}
The classical PUF $f:\R \times \X\rightarrow\Y$ with $(\mathcal{MP}_{C},\delta_{1},\delta_{2},\delta_{3},\epsilon,\lambda)$ satisfies the requirements defined below:
\begin{requirement}[$\delta_{1}$-Robustness]
Whenever a single classical PUF is repeatedly evaluated with a fixed input, the maximum distance between any two outputs $y_i\leftarrow f(x)$ and $y_j\leftarrow f(x)$ is at most $\delta_{1}$. That is for a created PUF $f$ and $x\in\X$, it holds that:
\begin{equation}
    \Pr\left[max({Dist(y_i,y_j)}_{i\neq j})\leq\delta_{1}\right]=1-\epsilon(\lambda).
\end{equation}
\end{requirement}
\begin{requirement}[$\delta_{2}$-Collision Resistance]
Whenever a single classical PUF is evaluated on different inputs, the minimum distance between any two outputs $y_i\leftarrow f(x_i)$ and $y_j\leftarrow f(x_j)$ is at least $\delta_{2}$. That is for a created PUF $f$ and $x_i,x_j\in\X$, it holds that:
\begin{equation}
    \Pr\left[min({Dist(y_i,y_j)}_{i\neq j})\geq\delta_{2}\right]=1-\epsilon(\lambda).
\end{equation}
\end{requirement}
\begin{requirement}[$\delta_{3}$-Uniqueness]
Whenever any two classical PUFs are evaluated on a single, fixed input, the minimum distance between any two outputs $y_i\leftarrow f_{i}(x)$ and $y_j\leftarrow f_{j}(x)$ is at least $\delta_{3}$.  That is for a created PUF $f$ and $x\in\X$, it holds that:
\begin{equation}
    \Pr\left[min({Dist(y_i,y_j)}_{i\neq j})\geq \delta_{3}\right]=1-\epsilon(\lambda)
\end{equation}
\end{requirement}
where $Dist(.,.)$ is a general notion of distance between the responses.
\end{definition}

We also introduce the notion of $\emph{randomness}$ for the classical PUF $f$. It says the maximal probability of $p^f_x(y)$ with an input $x_j\in \X$ on PUF $f_i$ where $i\in\R$. conditioned on the residual output space. A formal definition is as follows.
\begin{definition}[$p$-Randomness]\label{def:p-randomness}
We define the $p$-randomness of a classical PUF $f:\R \times \X\rightarrow\Y$ as \begin{equation}
    p := \max_{\substack{x \in \X \\ y \in \Y}} p^f_x(y).
\end{equation}
\end{definition}
For a correct valid modelling of PUF, $\delta_{1}<\delta_{2}$ and $\delta_{1}<\delta_{3}$ are necessary conditions to allow for a clear distinction between different input and different PUFs.

A quantum PUF, is again a hardware primitive that is unclonable by assumption which also utilises the properties of quantum mechanics. Similar to a classical PUF, a QPUF is assessed via challenge and response pairs (CPR). However, in contrast to a classical PUF where the CRPs are classical states, the QPUF CRPs are quantum states. Moreover, the evaluation algorithm of a QPUF is modelled by a general quantum transformation that is a CPTP map that produces an output in the form of a quantum state. A quantum transformation needs to have few requirements such as robustness, collision resistance and uniqueness to be considered a QPUF, similar to its classical counterpart. The focus of this paper is not on full quantum PUFs, and only for Section~\ref{ap:limitation}, where we discuss the feasibility of lockdown technique for general quantum PUFs, we use the QPUF as defined in~\cite{ADDK19}.

\section{Unforgeability against Adaptive and Weak Adversaries}\label{ap:adv_model}

\subsection{Models for adaptive and weak adversaries}

In this paper, we only consider the \emph{network adversarial model}, i.e., the adversary has only access to the communication channel. Moreover, we assume that the \textbf{manufacturer of the PUF is honest}. The network adversaries can get the challenge-response pairs just by intercepting the messages that are exchanged between the server and the clients. They can also pretend to be the server and make queries to the PUF on the client side with a challenge and get the response. 

Any network adversary that tries to predict the response of a PUF namely $\E: \din \rightarrow \dout$, can be modelled as an interactive algorithm. Here we consider Quantum Polynomial-Time (QPT) adversaries that have $q$-query classical access to the evaluation of the PUF, where $q$ is polynomial in the security parameter. An adaptive adversary can choose and issue any arbitrary query (up to $q$-query) which could also depend on the previous responses received from the PUF. On the other hand, a weak non-adaptive adversary, cannot choose the queries and instead receives $q$ CRPs of $\E$. In this case, the queries are being picked at random from a uniform distribution by an honest party and sent to the adversary. 

\subsection{Unforgeability with game-based security}
\emph{Unforgeability} is the main security property of PUFs. \emph{Unforgeability} means that given a subset of challenge-response pairs of the target PUF, the probability of correct estimation of a new challenge-response pair is negligible in terms of the security parameter. The unforgeability for Classical PUFs has been defined in~\cite{sako_towards_2016}, and for Quantum PUFs in~\cite{ADDK19} as a game-based definition. Moreover, a general game-based framework for quantum unforgeability has been defined in~\cite{DDKA21} for both quantum and classical primitives in an abstract way. Following the previous works, here in this paper, we present a game-based unforgeability definition for PUFs, emphasizing the adversary's capabilities in the learning phase, and capturing both adaptive and weak adversaries as defined in the previous section. We define the unforgeability of PUF as a formal game between two parties: a \emph{challenger} ($\C$) and an \emph{adversary} ($\A$). The game is divided with 4 phases: \emph{Setup}, \emph{Learning}, \emph{Challenge} and \emph{Guess}. A formal description is given as follows:

\begin{game}[Universal Unforgeability of PUF\footnote{We use the term \emph{Universal Unforgeability} as defined in~\cite{DDKA21}, to avoid confusion with a stronger security model. Nevertheless, in the PUF literature, this level of security is also called \emph{Selective Unforgeability} as also was used in~\cite{ADDK19}.}]
\label{game:uni-unf-abs} 
Let $\mathcal{MP}$ be the manufacturing process, $\texttt{Ver(.)}$ be a verification algorithm for checking the responses, and $\lambda$ the security parameter. We define the following game $\Gn(\A,\lambda)$ running between an adversary $\A$ and a challenger $\C$:
\begin{itemize}
    \item \textbf{Setup phase.}
        \begin{itemize}
            \item $\C$ selects a manufacturing process $\mathcal{MP}$ and security parameter $\lambda$. Then $\C$ creates a PUF by $\E\leftarrow\mathcal{MP}(\lambda)$, which is described by a CPTP map. The challenge and response domain $\din$ and $\dout$ are shared between $\C$ and $\A$.
        \end{itemize}
    \item \textbf{Learning phase.}
        \begin{itemize}
            \item If the adversary is adaptive, $\A = \Aad$:
            \begin{itemize}
                \item $\Aad$ selects any desired challenge $c_i \in \din$, and issues to $\C$ (up to $q$ queries).
                \item $\C$ queries the PUF with each challenge $c_i$ and sends the response $r_i = \E(c_i) \in \dout$ back to $\Aad$.
            \end{itemize}
            \item If the adversary is weak (non-adaptive), $\A = \Ana$:
            \begin{itemize}
                \item $\C$ selects a challenge $c_i \in \din$ uniformly at random from $\din$ and independent of $i$.
                \item $\C$ queries the PUF with $c_i$ and produces the response $r_i = \E(c_i)$.
                \item $\C$ issues to $\Ana$ the set of random challenges and their respective responses $\{(c_i,r_i)\}^q_{i=1}$.
            \end{itemize}
        \end{itemize} 

    \item \textbf{Challenge phase.} 
    \begin{itemize}
        \item $\C$ chooses a challenge $\tilde{c}$ uniformly at random from challenge domain $\din$.
        \item $\C$ issues $\tilde{c}$ to $\A$.
    \end{itemize}
    
    \item \textbf{Guess phase.}
    \begin{itemize}
        \item For the challenge $\tilde{c}$, $\A$ produces his forgery $\sigr\leftarrow \A(1^{\lambda},\tilde{c},\{(c_i,r_i)\}^q_i)$ and sends to $\C$.
        \item $\C$ runs a verification algorithm $b \leftarrow \texttt{Ver}(\sigr, \tilde{r})$, where $\tilde{r} = \E(\tilde{c})$ is the correct output and $b \in \{0,1\}$, to check the fidelity or equality of the responses.
        \item $\C$ outputs $b$. $\A$ wins if $b=1$.
    \end{itemize}
    
\end{itemize}
\end{game}

The above game is the abstract version of the unforgeability game that can be used for different classical or quantum PUFs and with different challenge types. For instance, the learning phase challenges $c_i$ can be classical bit-strings or quantum states and in that case, the domain $\din$ will be a  Hilbert. Here we mostly focus on the notion of classical and Hybrid PUFs. As a result, we do not need the full generalization to the quantum setting. Nevertheless, for the sake of completeness, we also give a full quantum version of this game-based definition in Appendix~\ref{ap:qunf}.

Note that the adversary could not choose arbitrarily the challenges in the challenge phase in this game. So it is so-called \emph{universal unforgeability}. Relatively, there are different notions of unforgeability e.g, \emph{unconditional unforgeability} and \emph{existential unforgeability} \cite{ADDK19}. Unconditional unforgeability models the PUF against an unbounded adversary with unlimited queries during the learning phase, which is the strongest notion of unforgeability. The difference between existential unforgeability and universal unforgeability is that the adversary could choose the challenges during the challenge phase with existential unforgeability instead of choosing the challenges by the challenger. Even though the universal unforgeability is the weaker one compared with the rest of the two, it is sufficient for most PUF-based applications. 

Finally, we define game-based security in terms of universal unforgeability in this setting:
\begin{definition}[Universal Unforgeability against Adaptive Adversary]\label{def:uni-unf-adapt}
A PUF with manufacturing process $\mathcal{MP}$ and verification algorithm $\texttt{Ver(.)}$ provides $(\epsilon,\lambda)$-universal unforgeability against adaptive adversary if the success probability of any adaptive QPT adversary $\Aad$ in winning the game $\Gn(\Aad, \lambda)$ is at most $\epsilon(\lambda)$.
\begin{equation}
    Pr[1\leftarrow \Gn(\Aad, \lambda)] \leq \epsilon(\lambda)    
\end{equation}
\end{definition}

\begin{definition}[Universal Unforgeability against Weak Adversary]\label{def:uni-unf-weak}
A PUF with manufacturing process $\mathcal{MP}$ and verification algorithm $\texttt{Ver(.)}$ provides $(\epsilon,\lambda)$-universal unforgeability against weak (non-adaptive) adversary if the success probability of any weak QPT adversary $\Ana$ in winning the game $\Gn(\Ana, \lambda)$ is at most $\epsilon(\lambda)$.
\begin{equation}
    Pr[1\leftarrow \Gn(\Ana, \lambda)] \leq \epsilon(\lambda)    
\end{equation}
\end{definition}

\subsection{Unforgeability game for general quantum PUF against adaptive and weak adversary}\label{ap:qunf}
In this appendix, we introduce the full quantum unforgeability game against adaptive and weak (non-adaptive) adversaries. Any adversary that tries to predict the response of a PUF $\E:\h^{d_{\inp}} \rightarrow \h^{d_{\out}}$, can be modelled as an interactive algorithm. Here we consider Quantum Polynomial-Time (QPT) adversaries that have $q$-query access to the evaluation of the PUF, namely $\E$ where $q$ is polynomial in the security parameter. An adaptive adversary can choose and issue any arbitrary query which could also depend on the previous responses received from the PUF. On the other hand, a weak non-adaptive adversary, cannot choose the queries and will instead receive $q$ input/output pairs states of $\E$. In the case that all the queries are quantum, the post-learning phase database of a weak adversary can be easily modelled by the definition. However, an adaptive quantum adversary is likely to consume the quantum state of the response to be able to pick the next query adaptively. Hence modelling the post-query database of an adaptive quantum adversary is more challenging. In what follows we give a $q$-query mathematical model for adaptive and weak adversaries.

\begin{definition}[Adaptive and Weak Adversary]\label{def:adap-weak-adv}
Let $q$ be a positive integer, and $\E : \mathcal{H}^{d_{\inp}} \rightarrow \mathcal{H}^{d_{\out}}$ be a PUF.  
We model a probabilistic adversary as a CPTP map 
$\A : \R \times (\h^{d_{\inp}})^{\otimes q} \otimes (\h^{d_{\out}})^{\otimes q} \rightarrow (\h^{d_{\inp}})$. Such an adversary is called an \textbf{adaptive} adversary $\A_{ad}$ if for all random coin $r \in \R$ and for any $\bigotimes_{i=1}^q \rho^{\inp}_i \in (\h^{d_{\inp}})^{\otimes q}$ and for $\bigotimes_{i=1}^q  \rho^{\out}_i \in (\h^{d_{\out}})^{\otimes q}$ (where $\rho^{\out}_i :=\E(\rho^{\inp}_i)$), the mapping $\bigotimes_{i=1}^q (\rho^{\inp}_i\otimes \rho^{\out}_i) \rightarrow \A^{r}_{ad}(\bigotimes_{i=1}^q (\rho^{\inp}_i\otimes \rho^{\out}_i))$ is dependent on the  $\rho^{\inp}_1\otimes \rho^{\out}_1, \ldots , \rho^{\inp}_q\otimes\rho^{\out}_q$; 
For a \textbf{weak} adversary $\Ana$ the mapping $\bigotimes_{i=1}^q (\rho^{\inp}_i\otimes \rho^{\out}_i) \rightarrow \A^{r}_{ad}(\bigotimes_{i=1}^q (\rho^{\inp}_i\otimes \rho^{\out}_i))$ is independent of $\bigotimes_{i=1}^q (\rho^{\inp}_i\otimes \rho^{\out}_i)$. Moreover, the adversary has no choice over the query, i.e., all the queries $\otimes_{i=1}^q\rho^{\inp}_i$ are chosen following a distribution $\R$, and a third party chooses the distribution. 
\end{definition}

Intuitively, an adaptive adversary $\A : \R \times (\h^{d_{\inp}})^{\otimes q} \otimes (\h^{d_{\out}})^{\otimes q} \rightarrow (\h^{d_{\inp}})$ captures the strategy to choose the query input $\rho^{\inp}_{q+1} \in \h^{d_{\inp}}$ to the PUF $\E$. The adversary can use these query response pairs to predict the output of the PUF. 
We call the pair $(\bigotimes_{i=1} ^q \rho^{\inp}_i,\bigotimes_{i=1} ^q \rho^{\out}_i)$ that is generated after the $q$-round of interaction between an adversary $\A$ and a PUF $\E$, as a transcript. Note, that the transcripts depend on the choice of the random coins of $\A$.\\

Similar to Game~\ref{game:uni-unf-abs}, We define the unforgeability of PUF as a formal game between two parties: a \emph{challenger} ($\C$) and an \emph{adversary} ($\A$). The difference here is that our adversaries are defined according to Definition~\ref{def:adap-weak-adv}. A formal description is given as follows:

\begin{game}[Universal Unforgeability of PUF] Let $\mathcal{MP}$ be the manufacturing process, $\texttt{Ver(.)}$ be a verification algorithm for checking the responses, and $\lambda$ the security parameter. We define the following game $\Gn(\A,\lambda)$ running between an adversary $\A$ and a challenger $\C$:
\begin{itemize}
    \item \textbf{Setup phase.}
        \begin{itemize}
            \item $\C$ selects a manufacturing process $\mathcal{MP}$ and security parameter $\lambda$. Then $\C$ creates a PUF by $\E\leftarrow\mathcal{MP}(\lambda)$, which is described by a CPTP map. The challenge and response domain $\Hilin$ and $\Hilout$ are shared between $\C$ and $\A$.
        \end{itemize}
    \item \textbf{Learning phase.}
        \begin{itemize}
            \item If the adversary is adaptive, $\A = \Aad$:
            \begin{itemize}
                \item $\Aad$ selects and prepares an initial state $\rin_{0} \in \Hilin$, while having full access to the preparation algorithm.
                \item $\Aad$ issues to $\C$ the initial challenge state $\rin_{0}\otimes\rho_{anc}$ where $\rho_{anc}$ is an initially blank state.
                \item $\C$ queries the PUF with $\rin_{0}$ and sends the response $(\E\otimes\mathcal{I})\rin_{0}\otimes\rho_{anc}$ back to $\Aad$
                \item for the next challenges ($i \neq 0$), the adaptive adversary $\Aad$ produces a new challenge for next query as $\rin_{i}=\A^{r_{i}}_i((\E\otimes\mathcal{I})\rin_{i-1})$ and issues to $\C$.
                \item $\C$ queries the PUF with $\rin_{i}$ and sends the response to $\A$. Recursively, $\A$ obtains the CPRs with challenge $\rin_{i}=\A_{i}^{r_{i}}(\E\otimes I)\A_{i-1}^{r_{i-1}}(\E\otimes I)\dots\A_{1}^{r^1}(\E\otimes I)(\rin_{0})$ and corresponding response $\rho_{i}^{\out}=(\E\otimes I)(\rin_{i}\otimes\rho_{anc})$
            \end{itemize}
            \item If the adversary is (weak) non-adaptive, $\A = \Ana$:
            \begin{itemize}
                \item $\C$ selects a challenge $\rin_{i}$ uniformly at random from $\Hilin$ and independent of $i$, while being able to prepare arbitrary copies of each challenge.
                \item $\C$ queries the PUF with $\rin_{i}$ and produces the response $\E(\rin_{0})$.
                \item $\C$ issues to $\Ana$ the set of random challenges $\bigotimes_{i=1}^{q}\rin_{i}$ and their respective responses $\bigotimes_{i=1}^{q} \rout_{i}$.
            \end{itemize}
        \end{itemize} 

    \item \textbf{Challenge phase.} 
    \begin{itemize}
        \item $\C$ chooses a challenge $\rc$ uniformly at random from challenge domain $\Hilin$. $\C$ can produce multiple copies of the challenge, and the respective response locally.
        \item $\C$ issues $\rc$ to $\A$.
    \end{itemize}
    
    \item \textbf{Guess phase.}
    \begin{itemize}
        \item For the challenge $\rc$, $\A$ produces his forgery $\sigr\leftarrow \A(1^{\lambda},\rc,\{(\rho_{i}^{\inp},\rho_{i}^{\out})\})$ and sends to $\C$.
        \item $\C$ runs a verification algorithm $b \leftarrow \texttt{Ver}(\sigr, \rr, \rho_{\C})$, to check the fidelity of the responses. Where $\rr = \E(\rc)$ is the correct output, $\rho_{\C}$ is the local register of the challenger that can include extra copies of correct output if necessary for the verification, and $b \in \{0,1\}$.
        \item $\C$ outputs $b$. $\A$ wins if $b=1$.
    \end{itemize}
    
\end{itemize}
\end{game}

Finally, the security definitions can be defined based on this game, similar to definitions~\ref{def:uni-unf-adapt} and ~\ref{def:uni-unf-weak}.

\section{Formal construction of HPUF}\label{ap:hpuf-formal}

We have illustrated our HPUF construction in the main text. Here in Construction \ref{cons:hpuf_1}, we give the formal description of our HPUF design which is based on conjugate coding \cite{W1983}.
For our construction, we start with a classical PUF (CPUF) that has a certain amount of randomness (also denoted as min-entropy). To increase the min-entropy further, we encode the output of the CPUF into non-orthogonal quantum states and send the qubits through the communication channel. We refer to the entire system, i.e., CPUF together with a quantum encoding as hybrid PUF (HPUF).  

\begin{construction}[Hybrid PUF]
\label{cons:hpuf_1}
Suppose $f:\{0,1\}^n \rightarrow \{0,1\}^{4m}$ be a classical PUF, that maps an $n$-bit string $x_i \in \{0,1\}^n$ to an $4m$-bit string output $y_i \in \{0,1\}^{4m}$. We denote the $j$-th bit of $y_i$ as $y_{i,j} \in \{0,1\}$. From the $4m$-bit string, we prepare the set of $2m$-tuples $\{(y_{i,(2j-1)}, y_{i,2j})\}_{1\leq j \leq 2m}$. The hybrid PUF encodes each of the tuples $(y_{i,(2j-1)}, y_{i,2j})$ into a single qubit $|\psi^{i,j}\rangle$ (also known as \emph{BB84 states}). The exact expression of the encoding is defined in the following way,

\begin{equation}
    |\psi^{i,j}_{\out}\rangle\langle\psi^{i,j}_{\out}| := 
    \begin{cases}
        |0\rangle\langle 0|~&(y_{i,(2j-1)}, y_{i,2j})= (0,0)\\
        |1\rangle\langle 1|~&(y_{i,(2j-1)}, y_{i,2j})= (1,0)\\
        |+\rangle\langle +|~&(y_{i,(2j-1)}, y_{i,2j})= (0,1)\\
        |-\rangle\langle -|~&(y_{i,(2j-1)}, y_{i,2j})= (1,1)
    \end{cases}
\end{equation}

For any $x_i \in \{0,1\}^n$, the mapping of the HPUF $\E_{f} : \{0,1\}^{n} \rightarrow (\h^{2})^{\otimes 2m}$ is defined as follows.
\begin{equation}
    \label{mqp:hpuf1}
    x_i \rightarrow |\psi^i_{\out}\rangle\langle\psi^i_{\out}|~~(\text{or}~ |\psi_{f(x_i)}\rangle\langle \psi_{f(x_i)}|)
\end{equation}
where $|\psi^i_{\out}\rangle\langle\psi^i_{\out}|=\bigotimes_{j=1}^{2m} |\psi^{i,j}_{\out}\rangle\langle\psi^{i,j}_{\out}|$.
\end{construction}

\section{Hybrid Locked PUF}
\label{ap:hlpuf-const-security}
In this section, we give the first construction for lockdown mechanics in the quantum setting. We use our proposed HPUF construction to increase the security of the classical PUFs against quantum adversaries and then we combine it with our quantum locking mechanism and construct a Hybrid Locked PUF (HLPUF) that resits powerful quantum adaptive adversaries. We then give a PUF-based authentication based on HLPUF and analyse its security. 

\subsection{Lockdown technique for Hybrid PUF}
In construction~\ref{cons:hlpuf} we show how to apply the lockdown technique on a hybrid PUF. We refer to such HPUFs with the lockdown technique as the hybrid locked PUFs (HLPUFs).  We formalise the construction as follows: 

\begin{construction}[HLPUF]
\label{cons:hlpuf}
Suppose we have a hybrid PUF $\E_f$ where $f:\{0,1\}^n \rightarrow \{0,1\}^{4m}$ is a CPUF. The mapping of the HLPUF $\E^{L}_{f}:\din\times \h^{d_{\out_{1}}}\rightarrow \h^{d_{\out_{2}}}\otimes\h^{\perp}$ corresponding to a hybrid PUF $\E$ is defined as follows:
\begin{equation}
    \label{eq:qlpuf_2}
    (x_i,\tilde \rho_1) \rightarrow \begin{cases}
    & \hspace{-0.15in}\ket{\psi_{f_2(x_i)}}\bra{\psi_{f_2(x_i)}}~\text{if } \texttt{Ver}(\ket{\psi_{f_1(x_i)}}\bra{\psi_{f_1(x_i)}},\tilde \rho_1) = 1\\
    & \perp \hspace{2cm}\text{otherwise.}
    \end{cases}
\end{equation}
where $\texttt{Ver}(.,.)$ is verification algorithm that checks the equality of the first half of the response based on the classical response $y^1_i$. To be precise, $\texttt{Ver}(.,.)$ is specified by measuring each qubit of the incoming quantum state with corresponding basis according to $\{y_{i,2j}\}_{1\leq j\leq 2m}$ of response $y_i$ and check the equality $\texttt{Equal}(y_{i,2j}, \tilde y_{i,2j})_{1\leq j\leq 2m}$ in our construction.
\end{construction}

\begin{figure}[!h]
    \centering
    \begin{tikzpicture}[
    node distance = 3mm, every node/.style = {rectangle, rounded corners, align=center, scale=0.95}]
    \node (puf) [draw, inner xsep=1mm, inner ysep=3mm] {HPUF $\E_f$};
    \node [coordinate, right=1.8cm of puf] (ADL){};
    \node (first) [below=1cm of ADL, draw, inner xsep=1mm, inner ysep=3mm] {$\texttt{Ver}(\ket{\psi_{f_1(x_i)}}\bra{\psi_{f_1(x_i)}},\tilde \rho_1)$};
    \node (second) [right=3.5cm of puf,draw, inner xsep=1mm, inner ysep=3mm] {Output};
    \node[name=outer1, dashed, fit=(puf) (first) (second), draw, inner xsep=3.8mm, inner ysep=3mm] {};
    \node (input_1)[left=0.4cm of puf] {$x_i$};
    \node (output_1) [right=0.5cm of puf, inner ysep=1mm] {$\ket{\psi_{f(x_i)}}\bra{\psi_{f(x_i)}}$};
    \node (input_2)[left=1.9cm of first] {$\tilde \rho_1$};
    \node (output_2) [right=0.4cm of second] {$\ket{\psi_{f_2(x_i)}}\bra{\psi_{f_2(x_i)}}$};
    \node (output_3) [below=0.1mm of output_2] {$/\perp$};
    \draw[-latex'] (input_1) -- (puf); 
    \draw[-latex'] (puf) -- (output_1); 
    \draw[-latex'] (output_1) -- (second); 
    \draw[-latex'] (output_1.south -| first.north) -- (first); 
    \draw[-latex'] (input_2) -- (first);
    \draw[-latex'] (second) -- (output_2);

    \draw[-latex'] (first) -| node[pos=0.7,right,font=\footnotesize] {b} (second);
    \end{tikzpicture}
    \caption{Hybrid Locked PUF (HLPUF) $\E^L_f$. The verification algorithm $\texttt{Ver}(.,.)$ is specified by measurement as described in Construction \ref{cons:hlpuf}. Here, $\ket{\psi_{f(x_i)}}\bra{\psi_{f(x_i)}}=\ket{\psi_{f_1(x_i)}}\bra{\psi_{f_1(x_i)}}\otimes\ket{\psi_{f_2(x_i)}}\bra{\psi_{f_2(x_i)}}$
    \label{fig:qlpuf}}
\end{figure}


\subsection{Security Analysis}\label{ap:hlpuf-security}
In this section, we give a comprehensive security analysis of the previously proposed constructions. First, we show that using hybrid construction will exponentially improve the security of classical PUFs. More precisely, it will exponentially decrease the success probability of a quantum adversary in the universal unforgeability game, compared to a classical PUF with the same number of learning queries. Further, we show how much quantum communication can improve the security of a weaker classical PUF and as a result propose an efficient and secure construction that can be built using existing classical PUFs. Finally, we analyse the completeness and security of the hybrid PUF-based device authentication protocol and show that under the assumption that the inherent classical PUF resists the weak quantum adversary, the HLPUF-based protocol will be secure against an adaptive adversary.

\subsubsection{Assumptions on the CPUFs}
\label{sec:assump}
For the security analysis of our constructions, we consider the following assumptions of the CPUFs $f:\{0,1\}^n \rightarrow \{0,1\}^{4m}$.

\begin{enumerate}
    \item For any input $x \in \{0,1\}^n$ the probability distributions of the $4m$ output bits $f(x)_1, \ldots , f(x)_{4m}$ are independent and identically distributed (i.i.d).
    \item The output distributions $\{p^f_x(y)\}_{y \in \{0,1\}^{4m}}$ for all the inputs $x$ are independent and identically distributed (i.i.d).
\end{enumerate}

\subsubsection{Security of the HPUFs against weak adversaries}

Intuitively the security of our HPUF comes from the indistinguishability property of the non-orthogonal quantum states. In Theorem \ref{thm:cpuf_vs_hpuf}, we first show that the HPUFs are at least as secure as the underlying CPUFs. Here we only give the proof sketch, later in Appendix \ref{apn:thm1} we give the detailed proof. 

\begin{theorem}
\label{thm:cpuf_vs_hpuf}
Let $f:\{0,1\}^n \rightarrow \{0,1\}^{2m}$ be a classical PUF. If there is no QPT weak adversary who can win the universal unforgeability game for CPUF with more than a negligible probability in the security parameter, then the HPUF constructed from $f$ according to construction~\ref{cons:hlpuf}, is also universally unforgeable. 
\end{theorem}
\begin{proof}[Proof Sketch.]
  Here we prove the theorem using a contrapositive argument, i.e., we show that if any QPT weak adversary can forge the HPUF, then it can also forge the underlying CPUF efficiently. If any QPT weak adversary can forge the HPUF, i.e., win the universal unforgeability game with a non-negligible probability, then for a random challenge $x^* \in_R \{0,1\}^n$ it can produce the correct output state $|\psi_{f(x^*)}\rangle$. Note that, the adversary can produce multiple copies of the output state $|\psi_{f(x^*)}\rangle$ by fixing all the internal parameters of the attack algorithm to the same values. The forged quantum state $|\psi_{f(x^*)}\rangle$ is a product state of $m$ qubit states, where each qubit belongs to the set $\{|0\rangle, |1\rangle, |+\rangle,|-\rangle\}$. If the adversary has multiple copies of each qubit, then it can perform full state tomography just by measuring them in the $\{|0\rangle,|1\rangle\}$-basis, and $\{|+\rangle,|-\rangle\}$-basis. Thus, it can learn $f(x^*)$ from $|\psi_{f(x^*)}\rangle$ with probability arbitrarily close to one. Therefore, it can forge the CPUF with a non-negligible probability. This concludes the proof sketch. The full proof is given in Appendix~\ref{apn:thm1}. 
\end{proof}

The above theorem is an intuitive result that shows HPUF is stronger or at least as strong as the underlying CPUF. Although we want to prove a more powerful and explicit statement regarding HPUFs by quantifying how much the hybrid construction will boost security. In fact, we want to show that one can construct a secure unforgeable HPUF against a quantum adversary even if the underlying CPUF is breakable (with a certain probability) against the classical forger. To this end, we compare the success probability of a QPT adversary in breaking the HPUF in the universal unforgeability game, with the success probability of the adversary who breaks the CPUF with a certain non-negligible probability in a fixed query setting. This will allow us to show that some of the weak and considerably broken CPUFs can still be used to construct an asymptotically secure HPUF against stronger quantum adversaries since the quantum encoding drastically decreases the success probability. 

In Lemma \ref{lem:guess_prob}, first we give an upper bound on the adversary's guessing probability of the response $f(x_i)$ corresponding to a challenge $x_i$ and a single copy of the quantum response state $|\psi_{f(x_i)}\rangle$. The complete proof can be found in Section \ref{appn:lem1}.

\begin{lemma}
\label{lem:guess_prob}
Suppose $f:\{0,1\}^n \rightarrow \{0,1\}^{4m}$ be a CPUF with the following property,
\begin{equation}
    \forall~~x_i\in \{0,1\}^n, \forall~~1\leq j \leq 4m,~~p^{f}_{x_i}(y_{i,j} = 0) =\frac{1}{2}+\delta_r, 
\end{equation}
with a biased distribution $p=\frac{1}{2}+\delta_r$ where $0 \leq \delta_r \leq \frac{1}{2}$, and $\E_f$ be a HPUF corresponding to $f$ that we construct using Construction \ref{cons:hpuf_1}. Let a quantum adversary $\A^{i,j}_{guess}$ extract the value $y_{i,(2j-1)}$ out of $(y_{i,(2j-1)},y_{i,2j})$ from quantum state $|\psi^{i,j}_{\out}\rangle \langle \psi^{i,j}_{\out}|$ corresponding to a random challenge $x_i$. If all the output bits of the CPUF are independent and identically distributed, then for any quantum adversary $\A^{i,j}_{guess}$, and  $\forall~x_i\in \{0,1\}^n$,
\begin{align}
    \nonumber
    p_{\guess} &:= \Pr[\A^{i,j}_{guess}(x_i, |\psi^{i,j}_{\out}\rangle \langle \psi^{i,j}_{\out}|)= y_{i,(2j-1)}] \\
    \nonumber &\leq p(1+\sqrt{p^2+(1-p)^2})\\
    &\leq p(1+\sqrt{2}p)
\end{align}
\end{lemma}

Lemma \ref{lem:database_dest1} shows that the adversary needs to extract the classical information $f(x)$ that is encoded in the quantum state $|\psi_{f(x)}\rangle$ for the forgery of the HPUFs. Here we only state the lemma, and for the complete proof we refer to Appendix \ref{appn:lem2}.

\begin{lemma}
\label{lem:database_dest1}
Suppose $|D_q\rangle = \bigotimes_{i=1}^q \left(|x_i\rangle_C \otimes |\psi_{f(x_i)}\rangle_R \right)$ denotes the adversary's database of $q$ random CRPs that are generated from a HPUF $\E_f:\{0,1\}^n \rightarrow (\h^2)^{\otimes m}$. If $E(D_q)$ denotes the measurement strategy for forging the HPUF with probability $p_{\text{forge}}$ using the database $D_q$, then using the following measure-then-forge strategy that can forge the HPUF with the same probability $p_{\forge}$.
\begin{itemize}
    \item Adversary extracts the classical encoding $\{f(x_i)\}_{1 \leq i \leq q}$ from $|D_q\rangle$. Let $\{\tilde f(x_i)\}_{1 \leq i \leq q}$ denotes the extracted classical string.
    \item The QPT adversary applies a forging strategy using the extracted data set $\{\tilde f(x_i)\}_{1 \leq i \leq q}$.
\end{itemize}
\end{lemma}

Lemma \ref{lem:database_dest1} suggests that the optimal adversary first needs to extract the classical information from the database state $|D_q\rangle$, and then perform the modelling attack to guess $|\psi_{f(x^*)}\rangle$. In general, if the extracted classical information $\{\tilde f(x_i)\}_{1 \leq i \leq q}$ from the database state $|D_q\rangle$ is very far from the original encoded string $\{f(x_i)\}_{1 \leq i \leq q}$ then it would be difficult for the adversary to forge the HPUF, based on that noisy data set. Here, we define the distance between $\tilde D^x_q = \{\tilde f(x_i)\}_{1 \leq i \leq q}$, and $D^x_q = \{f(x_i)\}_{1 \leq i \leq q}$ as follows.
\begin{equation}
    \dist(\tilde D^x_q, D^x_q) := \frac{\sum_{i=1}^q \text{Mis-match}(\tilde f(x_i),f(x_i))}{q},
\end{equation}
where we define $\text{Mis-match}(\tilde f(x_i),f(x_i))$ as follows.

\begin{equation}
    \text{Mis-match}(\tilde f(x_i),f(x_i)) := 
    \begin{cases}
    & 1 ~~~~~\text{If }(\tilde f(x_i) \neq f(x_i))\\
    & 0 ~~~~~\text{Otherwise.}
    \end{cases}
\end{equation}
It is reasonable to assume that no forging strategy can forge the HPUF with a non-negligible probability that runs on the noisy database set $\tilde D^x_q$ such that $\dist(\tilde D^x_q, D^x_q) > \varepsilon$, where $0\leq\varepsilon \leq 1$ is a parameter that quantifies the error threshold. In the next lemma, we give an upper bound on extracting $\tilde D^x_q$ from $|D_q\rangle$ such that $\dist(\tilde D^x_q, D^x_q) \leq \varepsilon$. Intuitively, a robust HPUF is with low $\varepsilon$ such that an adversary can not forge it with a noisy data set that is very far away from the original data set. Otherwise, the $\varepsilon$ should be high with a bad HPUF.

\begin{lemma}
\label{lem:pguess}
Suppose $|D_q\rangle = \bigotimes_{i=1}^q \left(|x_i\rangle_C \otimes |\psi_{f(x_i)}\rangle_R \right)$ denotes the adversary's database of $q$ random CRPs that are generated from a HPUF $\E_f:\{0,1\}^n \rightarrow (\h^2)^{\otimes m}$. If $\tilde D_q$ denotes the noisy classical response set that is extracted from $|D_q\rangle$ such that $\dist(D_q, \tilde D_q) \leq \varepsilon$ with probability $p_{\extract}$, then
\begin{equation}
    p_{\extract} \leq \sum_{k=(1-\varepsilon)q}^{q}\binom{q}{k} (p_{\guess})^{2mk}(1-(p_{\guess})^{2m})^{q-k},
\end{equation}
where $p_{\guess} \leq p(1+\sqrt{2}p)$, defined in Lemma \ref{lem:guess_prob}.
\end{lemma}

\begin{proof}[Proof Sketch.] 
A $q$-query weak adversary gets a $q$ random outputs from the HPUF $\E_f:\{0,1\}^n \rightarrow (\h^2)^{\otimes m}$ along with $q$ bit random strings $X_i \in_R \{0,1\}^n$. Here each output state is $m$-qubit product state, where each qubit belongs to $\{|0\rangle, |1\rangle, |+\rangle, |-\rangle\}$, depending on the value of the random variable $f(X_i)$. In Lemma \ref{lem:guess_prob}, we show that the probability of guessing a single output bit is $p_{\guess}$. Due to the i.i.d assumption on the different output bits of a single outcome of the CPUF, the probability of guessing all the $2m$ output bits from the state $|\psi_{f(X_i)}\rangle$ is upper bounded by $(p_{\guess})^{2m}$.

Here, we would like to compute the probability of successfully guessing $f(X_i)$'s for at least $(1-\varepsilon)q$ random samples. We denote this probability as $p_{\extract}$. Due to the i.i.d assumption on the outcomes $f(X_i)$'s of the CPUF, the probability of guessing exactly $k$ responses out of $q$ responses is given by $\binom{q}{k} (p_{\guess})^{2m} (1-(p_{\guess})^{2m})^{q-k}$. Therefore, we get the following upper bound on the $p_{\extract}$.

\begin{equation}
    p_{\extract} \leq \sum_{k=(1-\varepsilon)q}^{q}\binom{q}{k} (p_{\guess})^{2mk}(1-(p_{\guess})^{2m})^{q-k}.
\end{equation}
This concludes the proof.
\end{proof}

To provide a better intuition of the expression of $p_{\extract}$ to show the exponential gap, we give in Figure \ref{fig:pextract} the evolution of $p_{\extract}$ for different values of $\varepsilon$. It means that with a bad HPUF with high $\varepsilon$, the $p_{\extract}$ converges to $1-negl(\lambda)$ as $q$, and the number of queries of the QPT weak adversary increases. Otherwise, for a smaller error threshold, corresponding to a better HPUF, it decreases exponentially with $q$. Later, we show in Section \ref{ap:simulation} the $\varepsilon$ of HPUF depends on its underlying CPUFs, and the machine-learning algorithm we use to forge the HPUF.
\begin{figure}[!h]
    \begin{minipage}{\textwidth}
    \centering
      \subcaptionbox{\label{fig:pextract04}}{\includegraphics[width=0.5\linewidth]{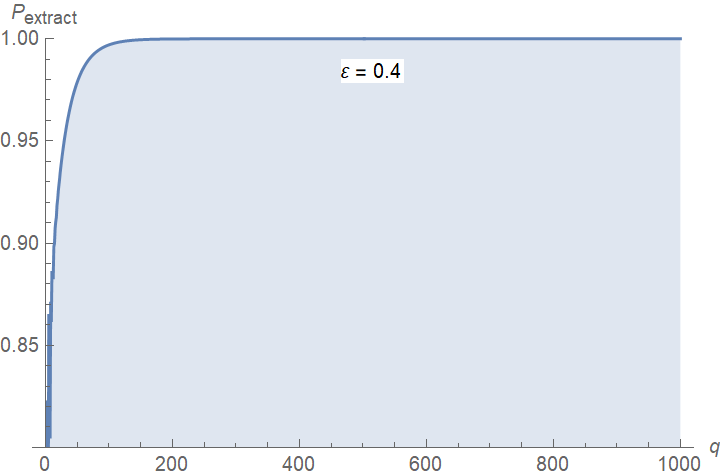}}\hspace*{\fill}
      \subcaptionbox{\label{fig:pextract015}}{\includegraphics[width=0.5\linewidth]{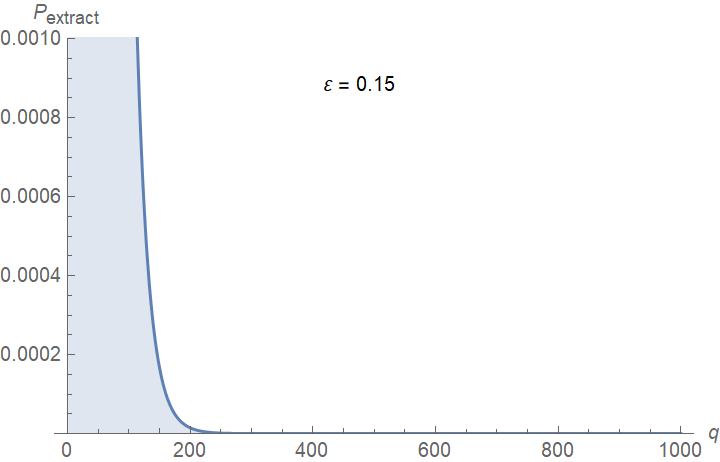}}
      \caption{Evolution of $p_{\extract}$ with different values of $\varepsilon$}\label{fig:pextract}
    \end{minipage}
\end{figure}

In the next theorem, we give an upper bound of the success probability of forging a HPUF by a QPT weak adversary.
\begin{theorem}
\label{thm:hpuf_prob_2}
Let  $f:\{0,1\}^n \rightarrow \{0,1\}^{4m}$ be a classical PUF with $p$-randomness, where $p = (\frac{1}{2}+\delta_r)$ with the following two properties. 
\begin{enumerate}
    \item Let any $q$-query weak adversary win the universal unforgeability game for the CPUF $f$ with probability at most $p_{\forge}^{\classical}(m,p,q) \geq nonnegl(\lambda)$.
    \item There is no QPT adversary that can win the universal unforgeability game for the CPUF using a noisy database $\tilde D_q$ such that $\dist(D_q,\tilde D_q) > \varepsilon$.
\end{enumerate} 
If we construct a HPUF $\E_f$ from such a CPUF $f$, then the $q$-query weak quantum adversary can win the universal  unforgeability game for the HPUF $\E_f$ with probability $p_{\forge}^{quantum}(x^*,p,|Q_q\rangle)$, such that,
\begin{equation}
   p_{\forge}^{\quantum}(x^*,p,|Q_q\rangle) \leq p_{\extract}\times p_{\forge}^{\classical}(m,p,q),
\end{equation}
where $$p_{\extract} \leq \sum_{k=(1-\varepsilon)q}^{q}\binom{q}{k} (p_{\guess})^{2mk}(1-(p_{\guess})^{2m})^{q-k}.$$
\end{theorem}

\begin{proof}
From Lemma \ref{lem:database_dest1}, we get that the optimal adversary's strategy is measure-then-forge. Let $\tilde D_q$ denotes the set of extracted database response. From the $2$nd property, we get that the adversary can forge the HPUF with a non-negligible probability if and only if $\dist(\tilde D_q, D_q) \leq \varepsilon$. Suppose $p_{\extract}$ denotes the optimal success probability of extracting $\tilde D_q$ from $|D_q\rangle$ such that $\dist(D_q, \tilde D_q) \leq \varepsilon$. If $p^{\classical}_{\forge}(\tilde D_q, X^*,p)$ denotes the optimal forging probability using the database $\tilde D_q$, then the total forging probability is given by the following equation.
\begin{equation}
    \label{eq:succ_prob_forge}
    p_{\forge}^{\quantum}(X^*,p,|D_q\rangle) = p_{\extract}\times p_{\forge}^{\classical}(\tilde D_q, X^*,p).
\end{equation}
Note that, the adversary's optimal forging probability with database $D_q$ is always higher than the optimal forging probability with the database $\tilde D_q$, i.e., \begin{equation}
    \label{eq:comp}
    p_{\forge}^{\classical}(m,p,q) \geq p_{\forge}^{\classical}(\tilde D_q, X^*,p).
\end{equation}
Substituting the relation in Equation \eqref{eq:comp} in Equation \eqref{eq:succ_prob_forge} we get the following expression of $p_{\forge}^{\quantum}(X^*,p,|D_q\rangle)$.
\begin{equation}
    \label{eq:succ_prob_forge_mod}
    p^{\quantum}_{\forge}(X^*,p,|D_q\rangle) \leq p_{\extract}\times p_{\forge}^{\classical}(m,p,q).
\end{equation}
From Lemma \ref{lem:pguess} we get that $p_{\extract} \leq \sum_{k=(1-\varepsilon)q}^{q}\binom{q}{k} (p_{\guess})^{2mk}(1-(p_{\guess})^{2m})^{q-k}$. By substituting the expression of $p_{\success}$ in Equation \eqref{eq:succ_prob_forge_mod}, we get the desired upper bound on the $p^{\quantum}_{\forge}(X^*,p,|D_q\rangle)$. This concludes the proof.
\end{proof}


The above result is a general statement for any fixed number of queries and compares the success probability of a weak adversary in breaking the unforgeability of CPUF and HPUF. Given this theorem, we can also easily state the following corollary that ensures the universal unforgeability of an HPUF constructed from a CPUF that does not provide suitable security, yet is not totally broken with overwhelming probability.

\begin{corollary}\label{cor:hpuf}
Let the success probability of any QPT weak-adversary in the universal unforgeability game with a CPUF $f:\{0,1\}^n \rightarrow \{0,1\}^{4m}$ with $p$-randomness, be at most $p^{\text{classic}}_{\forge}$, where $0 \leq p^{\text{classic}}_{\forge} \leq 1 - \text{non-negl}(2m)$. Then, there always exists an error threshold $0 < \varepsilon \leq 1$ for which the success probability of any QPT adversary in the universal unforgeability game for the HPUF $\E_f$, is at most $\epsilon(2m)$, which is a negligible function in the security parameter. Hence such HPUFs are universally unforgeable.
\end{corollary}

This directly follows from Theorem~\ref{thm:hpuf_prob_2} where $p^{\text{classic}}_{\forge} = p_{\forge}^{\classical}(m,p,q)$ for any $q = poly(m)$ is a value between $0$ and $1$, and not negligibly close to $1$. As shown in the proof of Theorem~\ref{thm:hpuf_prob_2} in the Appendix, for a large family of $\varepsilon$ the first part of the probability, namely $p_{\extract}$ becomes negligibly small (in $2m$) and hence the overall probability becomes a negligible function $\epsilon(2m)$.

\subsubsection{Security of the HLPUFs against general adaptive adversaries}

In the last two theorems, we analyse the security of the HPUFs against only weak adversaries. In Theorem \ref{thm:hpuf_auth} we show that if the HPUFs are secure against the weak adversaries then with the lockdown technique we can make the HLPUFs secure against the adaptive adversaries.  

\begin{theorem}
\label{thm:hpuf_auth}
Let $\E_f: \{0,1\}^{n}  \rightarrow (\h^{2})^{\otimes m} \otimes (\h^{2})^{\otimes m}$ be a hybrid PUF that we construct from a classical PUF $f:\{0,1\}^n \rightarrow \{0,1\}^{2m} \times \{0,1\}^{2m}$ and let $\E^L_{f}: \{0,1\}^{n} \times (\h^{2})^{\otimes m}\rightarrow (\h^{2})^{\otimes m}$ denotes the HLPUF that we construct from $\E_f$ using the Construction \ref{cons:hlpuf}.
If $\E_f = \E_{f_1} \otimes \E_{f_2}$ and if each of the mappings $\E_{f_1},\E_{f_2}$ has $(\epsilon, m)$-universal unforgeability against the $q$-query weak adversaries, then the corresponding HLPUF $\E^L_f$ is $(\epsilon, m)$-secure against the $q$-query adaptive adversaries.
\end{theorem}
\begin{proof}[Proof Sketch.]
According to the Construction \ref{cons:hlpuf}, if the adaptive adversary tries to query the HLPUF with any arbitrary challenge $x \in \{0,1\}^n$, then it also needs to send a quantum state $\rho_{f_1(x)}$. The adversary successfully gets $|\psi_{f_1(x)}\rangle$ as a reply if and only if $\ver(\rho_{f_1(x)}, |\psi_{f_1(x)}\rangle\langle \psi_{f_1(x)}) = 1$. Note that the adversary doesn't have any access to the underlying classical PUF $f_1$, therefore it cannot produce such a $\rho_{f_1(x)}$ for an arbitrary $x$. The only possible option is to use some of the previous intercepted queries $x,|\psi_{f_1(x)}\rangle$ that were sent by the server. As the server chooses its queries uniformly at random, the adaptive adversaries need to depend on those random queries to make an adaptive query to the HLPUF. Moreover, for the adaptive queries to the HLPUF, first the adversary needs to forge the mapping $\E_{f_1}$ using the $q$ random challenge-response pairs $\{x_i,|\psi_{f_1(x_i)}\rangle\}_{1\leq i \leq q}$. Here, we assume that the mapping $\E_{f_1}$ is secure against $q$-query weak adversaries, therefore the adaptive adversary cannot forge $\E_{f_1}$. Hence, the $q$-query adaptive adversary can only get the responses from the mapping $\E_{f_2}$ for at most $q$ random queries. According to the assumption, the mapping $\E_{f_2}$ is also secure against $q$-query weak adversaries. Therefore, from $q$ random challenge-response pairs the adaptive adversary couldn't forge $\E_{f_2}$. Hence, the HLPUF remains secure against the $q$-query adaptive adversaries. This concludes the proof sketch.  
\end{proof}

\subsubsection{Security of the HLPUF-based Authentication Protocol:}
In this section, we first give a full formal description of the HLPUF-based authentication protocol, then we define the completeness and security properties of Protocol~\ref{prot:hpuf}. Later, in Theorem \ref{thm:HLPUF_auth} we prove its completeness and security.

\begin{protocol}[!h]
    \caption{Hybrid PUF-based Authentication Protocol with Lockdown Technique}
    \label{prot:hpuf}
    \begin{enumerate}
        \item \textbf{Setup:}
        \begin{enumerate}
            \item The Prover $\P$ equips a Hybrid Locked PUF: $\E_{f}^L$ with HPUF $\E_{f}: \{0,1\}^{n} \rightarrow (\h^{2})^{\otimes 2m}$ constructed upon a classical PUF $f:\X \rightarrow \Y$. Here, the classical PUF $f$ maps an $n$-bit string $x_i \in \{0,1\}^n$ to an $4m$-bit string output $y_i \in \{0,1\}^{4m}$.
            \item The Verifier $\V$ has a classical database $D:= \{(x_i,y_i)\}_{i=1}^{d}$ with all $d$ CRPs of $f$, as well as the necessary quantum devices for preparing and measuring quantum states. 
        \end{enumerate}
        \item \textbf{Authentication:}
        \begin{enumerate}
            \item \label{a}$\V$ randomly chooses a CRP $(x_i,y_i)$ and splits the response equally into two partitions $y_i=f_1(x_i)||f_2(x_i)=y_i^1||y_i^2$ with length $2m$.
            \item \label{b}$\V$ then encodes the first partition of response into $\ket{\psi_{f_1(x_i)}}\bra{\psi_{f_1(x_i)}}:=\bigotimes_{j=1}^m |\psi^{i,j}_{f_1(x_i)}\rangle\langle\psi^{i,j}_{f_1(x_i)}|$ and issues the joint state $(x_i,\ket{\psi_{f_1(x_i)}}\bra{\psi_{f_1(x_i)}})$ to the client. 
            \item \label{c} $\P$ receives the joint state $(x_i,\tilde \rho_1)$ and queries Hybrid Locked PUF $\E_{f}^L$. If the verification algorithm $\texttt{Ver}(\ket{\psi_{f_1(x_i)}}\bra{\psi_{f_1(x_i)}},\tilde \rho_1)\geq 1-\epsilon(\lambda)$ with negligible $\epsilon(\lambda)$, $\P$ obtains $\ket{\psi_{f_2(x_i)}}\bra{\psi_{f_2(x_i)}}:=\bigotimes_{j=1}^m |\psi^{i,j}_{f_2(x_i)}\rangle\langle\psi^{i,j}_{f_2(x_i)}|$ from $\E_{f}^L$ and sends back to $\V$. Otherwise, the authentication aborts.
            \item \label{d}$\V$ receives the quantum state $\tilde{\rho}_2$ and performs the the verification algorithm $\texttt{Ver}(.,.)$ as described in Construction \ref{cons:hlpuf}. If $\texttt{Ver}(\ket{\psi_{f_2(x_i)}}\bra{\psi_{f_2(x_i)}},\tilde \rho_2)\geq 1-\epsilon(\lambda)$ with negligible $\epsilon(\lambda)$, the authentication passes. Otherwise, it aborts.
        \end{enumerate}
    \end{enumerate}
\end{protocol}

\begin{definition}[Completeness of HLPUF-based Authentication Protocol~\ref{prot:hpuf}]
\label{def:comp_hlpuf}
We say the HLPUF-based authentication protocol~\ref{prot:hpuf} satisfies completeness if in the absence of any adversary, an honest client and server generating $|\psi_{f_1(x_i)}\rangle\langle \psi_{f_1(x_i)}|$ and $|\psi_{f_2(x_i)}\rangle\langle \psi_{f_2(x_i)}|$ with a valid HLPUF $\E_f^L$ for any selected challenge $x_i$, can pass the verification algorithms with overwhelming probability:
\begin{align}
    \nonumber
   \text{Pr}[\texttt{Ver}&(|\psi_{f_1(x_i)}\rangle\langle\psi_{f_1(x_i)}|,\tilde \rho_1)\\
   &=\texttt{Ver}(|\psi_{f_2(x_i)}\rangle\langle \psi_{f_2(x_i)}|, \tilde{\rho}_2)= 1] \geq 1 - \epsilon(\lambda)
\end{align}
\end{definition}

Now, we also define the security of our HLPUF-based authentication protocol, in relation with the universal unforgeability game as follows: 

\begin{definition}[Security of the HLPUF-based Authentication Protocol~\ref{prot:hpuf}] 
\label{def:secu_hlpuf}
We say the HLPUF-based authentication protocol~\ref{prot:hpuf} is secure if the success probability of any QPT adaptive adversary $\Aad$ in winning the universal unforgeability game to forge an output of HLPUF $\E_f^L$ according to Construction~\ref{cons:hlpuf}, for any randomly selected challenge of the form $\tilde{c} = (x,|\psi_{f_1(x)}\rangle\langle \psi_{f_1(x)}|)$ is at most negligible in the security parameter:
\begin{equation}
    Pr[1\leftarrow \Gel(\Aad, \lambda)] \leq \epsilon(\lambda)
\end{equation}
\end{definition}


\begin{theorem}
\label{thm:HLPUF_auth}
If the HLPUF $\E^L_f$ is constructed from a hybrid PUF $\E_f$ using the Construction \ref{cons:hlpuf} then the locked PUF-based authentication Protocol \ref{prot:hpuf} satisfies both the completeness and security conditions.
\end{theorem}
\begin{proof}
In Protocol \ref{prot:hpuf} with hybrid PUF $\E_f = \E_{f_1} \otimes \E_{f_2}$, the server chooses the classical input $x_i\in\X$, encodes the quantum state corresponding to $2m$ bits of $f_1(x_i)$ and issues the joint state to the client. If there is no adversary, the client receives the joint state and queries $\E_f^L$ with $x_i$ and $\tilde \rho_1$, where $\tilde \rho_1=\E_{f_1}(x_i)=\ket{\psi_{f_1(x_i)}}\bra{\psi_{f_1(x_i)}}$ for the first $m$ qubits of $\E_{f}(x_i)$. Hence we have:
\begin{equation}
    \label{eq:thm_cs_1}
    Pr\left[\texttt{Ver}(|\psi_{f_1(x_i)}\rangle\langle \psi_{f_1(x_i)}|,\tilde \rho_1)=1\right]=1  
\end{equation}

On the client side, since the verification algorithm of HLPUF $\E_f^L$ always passes with $\texttt{Ver}(|\psi_{f_1(x_i)}\rangle\langle \psi_{f_1(x_i)}|,\tilde \rho_1)=1$, he returns the quantum state $\E_{f_2}(x_i)=|\psi_{f_2(x_i)}\rangle\langle \psi_{f_2(x_i)}|$ corresponding to $2m$ bits of $f_2(x_i)$ to the server. Without the presence of an adversary, the server always receives the state with $\tilde{\rho}_2=|\psi_{f_2(x_i)}\rangle\langle \psi_{f_2(x_i)}|$, and we obtain the equation similarly to Equation (\ref{eq:thm_cs_1}). Therefore, we can say the hybrid locked PUF-based authentication protocol satisfies the completeness condition with
\begin{align}
    \label{eq:thm_cs_2}
    \nonumber
    \text{Pr}[\texttt{Ver}&(|\psi_{f_1(x_i)}\rangle\langle \psi_{f_1(x_i)}|,\tilde \rho_1)\\
    &=\texttt{Ver}(|\psi_{f_2(x_i)}\rangle\langle \psi_{f_2(x_i)}|, \tilde{\rho}_2)= 1] = 1
\end{align}

On the other hand for security, we rely on Theorem \ref{thm:hpuf_auth} that the HLPUF $\E_f^L$ is $(\epsilon,m)$-secure against any $q$-query adaptive adversaries.
In the theorem, we show the fact that the adaptive adversary cannot boost from the weak-learning phase of HPUF $\E_{f_2}$, producing a forgery $\sigma_{2}$ for $\E_f^L$ that passes the verification $\texttt{Ver}(|\psi_{f_2(x_i)}\rangle\langle \psi_{f_2(x_i)}|, \sigma_2)$
. Since $\E_{f_2}$ has the universal unforgeability against a weak adversary by assumption, we have:
\begin{align}
    \nonumber
    Pr[1\leftarrow \Gel(\Aad, m)] &= Pr[1\leftarrow \Geb(\Ana, m)] \\ \label{eq:thm_cs_5}
    &\leq \epsilon(m)
\end{align}
This concludes the proof.
\end{proof}

\section{Challenge Reusability}
\label{ap:reusability}
We have discussed in the main paper about the issue of challenge-reusability in classical PUF-based protocols and discussed how our construction brings forward unique and new solution for this problem. In this section, we dive deeper into this issue and we formally prove why our proposal satisfies the important property of challenge reusability. 

We are thus interested in the eavesdropping attacks by the adversary on the first and second half of the response states that are of the form $\ket{\psi_{f_1(x_i)}}\bra{\psi_{f_1(x_i)}}=\bigotimes_{j=1}^m |\psi^{i,j}_{f_1(x_i)}\rangle\langle\psi^{i,j}_{f_1(x_i)}|$ and $\ket{\psi_{f_2(x_i)}}\bra{\psi_{f_2(x_i)}}=\bigotimes_{j=1}^m |\psi^{i,j}_{f_2(x_i)}\rangle\langle\psi^{i,j}_{f_2(x_i)}|$.
Note that eavesdropping on the states that encode the first part of the response will lead to breaking the locking mechanism while eavesdropping on the second half will lead to an attack on the authentication (\yao{Removed identification}). Without loss of generality, we only consider one of the cases where the adversary wants to eavesdrop on the first (or second) half to break the protocol in the upcoming rounds where the challenge is reused. The arguments will hold equivalently for both cases since the states and verification are symmetric.

Given all these considerations, the challenge reusability problem will reduce to the optimal probability of the eavesdropping attack on $\ket{\psi_{f_1(x_i)}}\bra{\psi_{f_1(x_i)}}=\bigotimes_{j=1}^m |\psi^{i,j}_{f_1(x_i)}\rangle\langle\psi^{i,j}_{f_1(x_i)}|$ which is in fact $m$ qubit states encoded in conjugate basis same as BB84 states. In the most general case, the adversary can perform any arbitrary quantum operation on the state $\bigotimes_{j=1}^m |\psi^{i,j}_{f_1(x_i)}\rangle\langle\psi^{i,j}_{f_1(x_i)}|$ or separately on each qubit state $\ket{\psi^{i,j}_{f_1(x_i)}}$, together with a local ancillary system and sends a partial state of this larger state to the verifier to pass the verification test, and keep the local state to extract the encoded response bits. Let $\rho_{SEC}$ be the joint state of the server, the eavesdropper and the client. Since the states used in the protocol are from Mutually Unbiased Basis (MUB) states i.e. from either $Z = \{\ket{0}, \ket{1}\}$ or $X = \{\ket{+}, \ket{-}\}$, in order to show the optimal attack, we can rely on the entropy uncertainty relations that have been used for the security proof of QKD. The measurements for verification are also performed in the $\{Z,X\}$ bases accordingly. We use the entropy uncertainty relations from~\cite{coles2017entropic} where the security criteria for QKD have been given in terms of the conditional entropy for MUBs measurements. Using these results we show that the entropy of Eve in guessing the correct classical bits for the response is very high if the state sent to the verification algorithm passes the verification with a high probability. Intuitively this is due to the uncertainty that exists related to the commutation relation between $X$ and $Z$ operators in quantum mechanics. Hence we conclude that the success probability of Eve in extracting information from the encoded halves of the response is relatively low. Also, we show that this uncertainty increases linearly with $m$ similar to the number of rounds for QKD. This argument results in the following theorem which we will formally describe and prove in Appendix~\ref{ap:ch-reusability} where we also introduce the uncertainty relations.

\begin{theorem}[informal]\label{th:uncertainty}
In Protocol~\ref{prot:hpuf}, if the client (or server for the second half of the state) verification does not abort for a challenge $x$, then Eve's uncertainty on the respective response of the CPUF, denoted by $H_{min}^{Eve}$ is greater than $m - \epsilon(m)$.
\end{theorem}

Now, we first define the reusability in relation with the unforgeability game and then using Theorem~\ref{th:uncertainty}, we prove the challenge reusability of the HLPUF-based Protocol~\ref{prot:hpuf}.

\begin{definition}[Challenge ($k$-)reusability in the universal unforgeability game]
Let $\Gnre(\lambda, \A, x_{k+1})$ be a special instance of the universal unforgeability game, where a challenge $x$, picked uniformly at random by the challenger, has been previously used $k$ times. We are interested in the events where the same challenge is used in the $(k+1)$-th round, which we denote by $x_{k+1}$. We say the challenge $x$ is \textit{($k$-)re-usable} if the success probability of any QPT adversary in winning $\Gnre(\lambda, \A, x_{k+1})$, i.e, in forging message $x_{k+1}$, is negligible in the security parameter:
\begin{equation}
    Pr_{forge}(\A, x_{k+1}) = Pr[1 \leftarrow \Gnre(\lambda, \A, x_{k+1})] \leq \epsilon(\lambda)
\end{equation}
\end{definition}

\begin{theorem}[Challenge reusability of HLPUF-based Authentication Protocol~\ref{prot:hpuf}] 
\label{thm:chall_reuse_mult}A challenge $x$ can be reused $k$ times during the Protocol~\ref{prot:hpuf} as long as the received respective response $\sigma$ for each round passes the (client's or server's) verification with overwhelming probability. In other words, under the successful verification, the success probability of the adversary in passing the $(k+1)$-th round with the same challenge $x$ is bounded as follows:
\begin{equation}
    Pr_{forge}(\A, x_{k+1}) \leq k 2^{-m} \approx \epsilon(m).
\end{equation}
\end{theorem}

\begin{proof}
To prove this theorem, we use the Theorem~\ref{th:uncertainty} directly. First, we assume that $x$ has been used one time before in a previous round. Given the assumption that the verification is passed with probability $1 - \epsilon(m)$, and this theorem, we conclude that the uncertainty of the adversary in guessing the encoded response of the HLPUF is larger than $m - \epsilon(m)$. In our case, the joint quantum state between the server and the adversary is a classical-quantum state (server has the classical description of $f(x)$, and the adversary has the quantum state $|\psi_{f(x)}\rangle$). For such states, Eve's uncertainty, $H_{min}^{Eve}$ is the same as $-\log P_{guess}^{Eve}$, where $P_{guess}^{Eve}$ is Eve's guessing probability of the classical information encoded in the quantum state \cite{KRS09}. Therefore,
\begin{equation}
\begin{split}
    P_{guess}^{Eve} & = 2^{-H_{min}^{Eve}}\\
    & \leq 2^{-m + \epsilon(m)}.
\end{split}
\end{equation}

This probability is negligible in the security parameter, which means that after performing any arbitrary quantum operations, the adversary's local state includes at most, a negligible amount of information on the response of $x$, each round that the state $x$ is reused. Now, we can use the union bound to show that this success probability only linearly scales with $k$:
\begin{equation}
    P_{guess}^{Eve,k} = P(\bigcup^k_{i=1} E_{guess}^i) \leq \sum^k_{i=1} P(E_{guess}^i) \approx k2^{-m},
\end{equation}
where $E_{guess}^i$ are the events where Eve correctly guesses the response and $P(E_{guess}^i) = (P_{guess}^{Eve})^i$ is the success probability of Eve in guessing in the $i$-th round. Finally, let the success probability of an adversary in the universal unforgeability game for the HLPUF be upper-bounded by $\epsilon_1(m)$ which is a negligible function in the security parameter since we assume that the HLPUF satisfies the universal unforgeability. This is the same as the success probability of the adversary in passing the verification for a new challenge, chosen at random from the database. Now in the $(k+1)$-th round, where the same $x$ is reused, the success probability is at most boosted by the guessing probability over the previous $k$-th rounds, hence we will have:
\begin{equation}
    Pr_{forge}(\A, x_{k+1}) \leq \epsilon_1(m) + k2^{-m} = \epsilon(m)
\end{equation}
As long as $k$ is polynomial in the security parameter, the second term is also a negligible function and since the sum of two negligible probabilities will also be negligible. This concludes the proof.
\end{proof}

\section{Simulation for HPUF/HLPUF}
\label{ap:simulation}
In this section, we simulate the design of HPUF/HLPUF constructions with underlying silicon CPUFs instantiated by \emph{pypuf} \cite{pypuf}. \emph{pypuf} is a python-based emulator that features different existing CPUFs. Furthermore, we simulate the situation where an adversary acquires classical challenges and quantum-encoded responses from HPUF/HLPUF and converts the responses into classical bitstrings by measuring the output quantum state. The adversary then attempts to perform machine learning-based attacks with the obtained CRPs to reproduce a model that predicts accurately enough the behaviour of the underlying CPUF. As a result, we say such an adversary wins the unforgeability game successfully in the end. According to the simulation result, we show the performance of hybrid construction in boosting the security of CPUF, quantify the existing advantage of hybrid construction and discuss potential improvements to obtain greater security. 


XOR Arbiter PUFs \cite{SD07} with $n$-bit challenge to a one-bit response is one of the CPUFs provided by pypuf. Its security is studied widely by Ulrich Rührmair et al. \cite{RSSD10}. In that paper, the performance of different machine learning attacks like \emph{Logistic Regression} (LR), \emph{Support Vector Machines} (SVMs), and \emph{Evolution Strategies} (ES) is evaluated in terms of the prediction accuracy of responses with unseen challenges. It turns out that the LR has the best performance. Moreover, it shows that the LR attacks can handle well with the situation while the training data is erroneous with noise up to $40\%$. In practice, this noise comes from the PUF implementation with the integrated circuit. Meanwhile, quantum encoding of HPUF can be treated as another source of noise to prevent the adversary from modelling CPUFs.


\subsection{BB84 encoding with split attack on the HPUF/HLPUF}
Recall that the HPUFs that we proposed in this paper encodes every two-bit tuple of response $(y_{i,(2j-1)}, y_{i,2j})_{1\leq j \leq 2m}$ into one BB84 state with $y_{i,2j}$ the basis value and $y_{i,(2j-1)}$ the bit value. Here, we assume that each bit of response is generated independently uniformly at random by an XOR Arbiter PUF. We simulate firstly an adaptive adversary on HPUF. he queries with the same classical challenge multiple times until he extracts the classical information from multi-copy of quantum response with high accuracy. The simulation results for modelling underlying CPUF are shown in red of Figure \ref{fig:simulation_k4} and \ref{fig:simulation_k5}.

On the other hand, while we consider HLPUF against an adaptive adversary, the lockdown technique reduces an adversary from adaptive to weak queries on HPUF. With a single copy of each quantum response uniformly at random, we intuitively think that the adversary has a $50\%$ probability of guessing the basis value correctly for each qubit of HPUF. If he guesses the basis value correctly, he can then measure the qubit correctly to obtain the exact $(y_{i,(2j-1)}, y_{i,2j})$. Otherwise, the classical tuple $(y_{i,(2j-1)}^\prime, y_{i,2j}^\prime)$ of each qubit obtained by the adversary is always incorrect. Hence, the success probability of recovering each tuple $\{(y_{i,(2j-1)}, y_{i,2j})\}$ from corresponding qubit $|\psi^{i,j}_\out\rangle \langle \psi^{i,j}_\out|$ by such an adversary is not greater than guessing a tossing coin.

However, there is a specific way to attack HPUFs that we discover throughout the simulation so-called \emph{Split Attack}. To the best of our knowledge, it is the optimal strategy that a weak adversary can perform on HPUF with underlying XORPUFs. We elaborate the attack as follows: Instead of predicting the tuple $(y_{i,(2j-1)}, y_{i,2j})$ simultaneously, the adversary first predicts the bit value $y_{i,(2j-1)}$ of each qubit. For the HPUF with BB84 states encoding, the problem of distinguishing a state from uniformly distributed BB84 states then reduces to the problem of distinguishing two mixed states $\rho_1^{i,j} = \frac{1}{2}|0\rangle\langle 0|+\frac{1}{2}|+\rangle\langle +|$ and $\rho_2^{i,j} = \frac{1}{2}|1\rangle\langle 1|+\frac{1}{2}|-\rangle\langle -|$ with equal probability. 
From Lemma \ref{lem:guess_prob}, we get the optimal success probability as,
\begin{align}
    \label{eq:bit_optimal_dis}
    \nonumber
    &Pr[\mathcal{A}^{i,j}_\text{guess}(x_i,\rho_1^{i,j}, \rho_2^{i,j})=y_{i,(2j-1)}]\\
    \nonumber
    &\leq\frac{1}{2}+\frac{1}{2}(\frac{1}{2}\norm{\rho_1^{i,j} - \rho_2^{i,j}}_1)\\
    \nonumber
    &=\frac{1}{2}+\frac{1}{2\sqrt{2}}\\
    &\approx 0.85.
\end{align}
As it is to say, the adversary $\mathcal{A}$ can perform LR attacks on bit value with a $15\%$ error afflicted CRPs training set. We do the simulation of HPUF with BB84 encoding and an underlying of 4-XOR Arbiter PUF and 5-XOR Arbiter PUF and a challenge size of 64 bits and 128 bits. 
Here, $k=4/5$ of XOR Arbiter PUF is the parameter related to its hardware structure. With higher value of $k$ of XORPUF, it takes more CRPs to model accurately with LR attacks.
The evolution of accuracy in predicting the bit value of each qubit with different underlying XORPUFs are shown in orange of Figure \ref{fig:simulation_k4} and \ref{fig:simulation_k5}.

After the bit value of each qubit can be predicted accurately with a given challenge, the problem of predicting the basis value $y_{i,2j}$ of the following qubits is equivalent to the adversary discriminates either a quantum state $|0\rangle$ from $|+\rangle$ if $y_{i,(2j-1)}=0$ or a quantum state $|1\rangle$ from $|-\rangle$ if $y_{i,(2j-1)}=1$. We denote the success probability of guessing the basis value correctly conditioned on an accurate prediction on bit value $y_{i,(2j-1)}^\prime=y_{i,(2j-1)}$ by $ Pr[\mathcal{A}^{i,j}_\text{guess}(x_i,|\psi^{i,j}_\out\rangle \langle \psi^{i,j}_\out|)=y_{i,2j}| y_{i,(2j-1)}^\prime=y_{i,(2j-1)}]$ from a quantum state $|\psi^{i,j}\rangle \langle \psi^{i,j}|$, we have: 
\begin{align}
    \nonumber
    & Pr[\mathcal{A}^{i,j}_\text{guess}(x_i,|\psi^{i,j}_\out\rangle \langle \psi^{i,j}_\out|)=y_{i,2j}| y_{i,(2j-1)}^\prime=y_{i,(2j-1)}] \\
    &= \frac{1}{2}+\frac{1}{2}\sin45^\circ \approx 0.85.
\end{align}

With the same level of noise introduced by HPUF on guessing the basis value and bit value, the similar performance of LR attack is expected to predict the basis value as long as the prediction accuracy of the bit value is high enough. We have the success probability of guessing both bit and basis values of tuple $(y_{i,(2j-1)}, y_{i,2j})$ as:
\begin{align}
    \nonumber
    &Pr[\mathcal{A}^{i,j}_{guess}(x_i, |\psi^{i,j}_\out\rangle \langle \psi^{i,j}_\out|)
    = (y_{i,(2j-1)},y_{i,2j})] = \\ &Pr[\mathcal{A}^{i,j}_\text{guess}(x_i,|\psi^{i,j}_\out\rangle \langle \psi^{i,j}_\out|)=y_{i,2j}|  y_{i,(2j-1)}^\prime=y_{i,(2j-1)}].
\end{align}
In the end, we get the evolution of accuracy on predicting a tuple $(y_{i,(2j-1)}, y_{i,2j})$ with different CRPs for training as the green curves in Figure \ref{fig:simulation_k4} and \ref{fig:simulation_k5}. The gap between the blue and green curves denotes the reinforcement of security by HPUF construction. \yao{We also simulate in Figure \ref{fig:ksizen} the best-performing training set sizes of CRPs for obtaining accurate enough models from machine learning attacks with different k-XORPUFs in the cases of CPUFs, HPUFs, and HLPUFs constructions. See \cite{hpuf} for details of the simulation.}

\begin{figure}[!h]
    \begin{minipage}{\textwidth}
    \centering
      \subcaptionbox{\label{fig:ksize64}}{\includegraphics[width=0.5\linewidth]{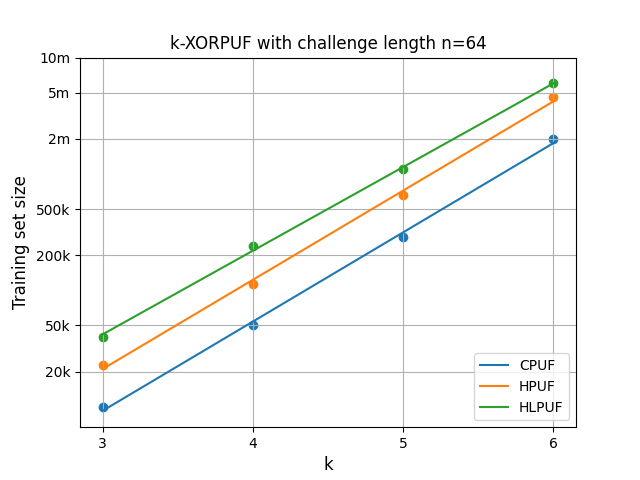}}\hspace*{\fill}
      \subcaptionbox{\label{fig:ksize128}}{\includegraphics[width=0.5\linewidth]{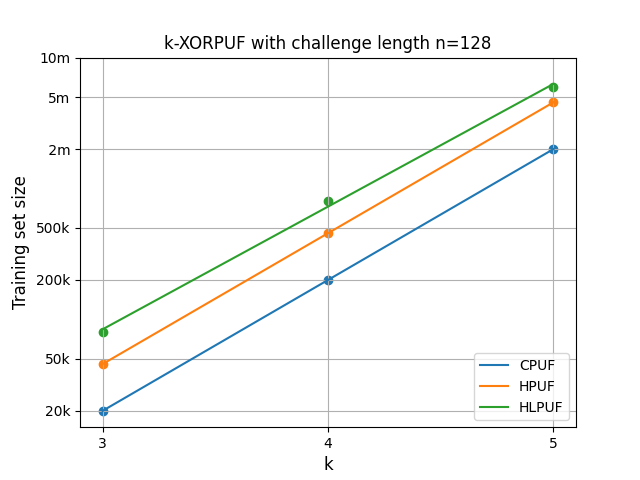}}
      \caption{Attack with best-performing number of CRPs for k-XORPUFs (CPUF, HPUF and HLPUF constructions) with challenge length $n=64/128$ and BB84 encoding}\label{fig:ksizen}
    \end{minipage}
\end{figure}

Corresponding to our proofs in Lemma \ref{lem:pguess} and Theorem \ref{thm:hpuf_prob_2}, our simulation shows an exponential advantage of HPUF compared to the same CPUF with a limited $q$-query in terms of the modelling success probability against an adversary by LR attacks. As to a larger $q$-query, the advantage shown in the simulation limits by the fact that k-XORPUFs is a vulnerable CPUF with a large $\varepsilon$, which allows a modelling attack with a noisy data set. That is to say, the probability $p_{\extract}$ can be high with $\dist(\tilde D^x_q, D^x_q)=0.15$. As long as $Pr(1,\frac{1}{2},q)=1-negl(\lambda)$, the success probability of modelling with hybrid construction converges to $1-negl(\lambda)$ with an increasing $q$.   
Therefore, to decrease the forging probability in practice, there are mainly two directions: Firstly, we choose more robust underlying CPUFs to construct HPUF with lower $\varepsilon$ and $Pr(1,\frac{1}{2},q)=1-negl(\lambda)$ with a greater $q$. Second, we can consider other sophisticated encodings of HPUF, e.g., MUB encoding of quantum states with higher dimensions. In the next section, we show the construction of HPUF with MUB encoding in 8-dimension and the simulation result.

\yao{In our simulations, the construction of H(L)PUF with underlying Arbiter-based PUFs generates a 1-bit response per query, thus although one can observe the exponential gap for a fixed number of queries between CPUF and HPUF, the inverse exponential scaling with $m$ cannot be witnessed. While for a general $m$-qubit response construction this inverse-exponential scaling can be seen from the theoretical results. In Figure~\ref{fig:mpn}, we also attempt to simulate this behaviour for a $m$-qubit response constructed by several Arbiter-based PUFs. The construction is a rather trivial one via parallelism, i.e., we simply duplicate the single structure $m$ times and query them by the same challenge \cite{SD07}. We note that this construction is far from optimal in terms of security, as it does not provide the required independent $m$-qubit outcome required in the theoretical result, and as a result it allows the adversary to perform more effective parallel attacks. However, we can still see that the guessing probability of an eavesdropper decreases inverse exponentially on $m$ until the averaged learning models are all accurate enough (See Figure \ref{fig:mpn} with 4-XORPUFs and different lengths of challenges). Moreover, the quantum encoding can in any case help with the detection of a network adversary trying to perform ML attacks, as such adversaries will perturb the quantum state in the quantum channels due to measurement, enabling the honest parties to detect their existence with high probability, and preventing the adversary from learning $m$-qubit states simultaneously during the protocol, as discussed in Appendix \ref{ap:reusability}.}


\begin{figure}[!h]
    \begin{minipage}{\textwidth}
    \centering
      \subcaptionbox{\label{fig:mp64}}{\includegraphics[width=0.5\linewidth]{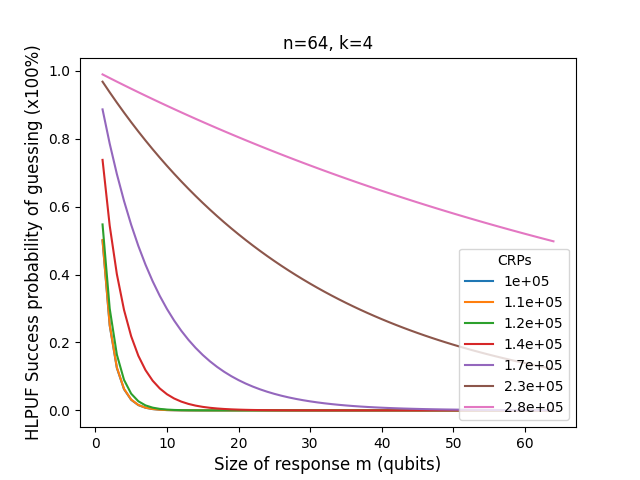}}\hspace*{\fill}
      \subcaptionbox{\label{fig:mp128}}{\includegraphics[width=0.5\linewidth]{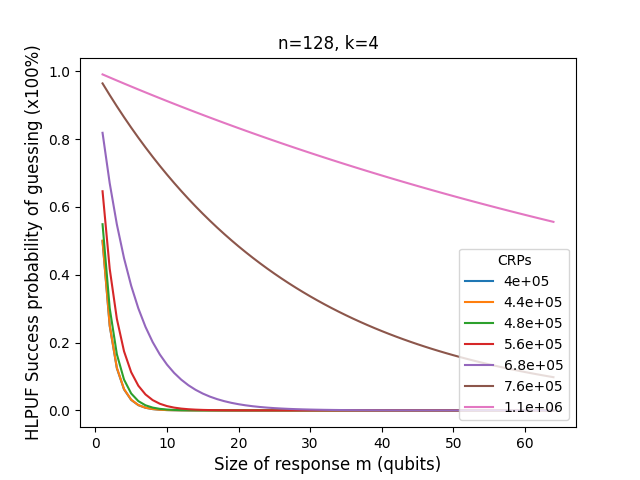}}
      \caption{HLPUF (BB84) Success probability of guessing with 4-XORPUFs and challenge length $n=64/128$}\label{fig:mpn}
    \end{minipage}
\end{figure}

\subsection{MUB in 8-dimension encoding with split attack}
\label{sec:MUB}
In this section, we show that a more sophisticated encoding of quantum state in higher dimensions, i.e., an 8-dimensional quantum state with 9 MUB, leads to more noise introduced to the database that an adversary emulates CPUFs with. We denote the encoding quantum state as:
\begin{equation}
    \ket{x^\theta}, x =x_0x_1x_2~\text{and}~\theta \in \{0,1,...8\}
\end{equation}
, where $\theta$ represents the basis and $x$ represents the state. Here, the adversary attempts to obtain the accurate models of $x_0x_1x_2$ from 3 CPUFs associated with the state value. Similarly to the strategy shown in BB84 encoding, the adversary performs a Split Attack on $x_0x_1x_2$ sequentially. The success probability of guessing bit is equivalent to the probability of distinguishing mixed states out of $\rho_x = \frac{1}{9}\sum_{\theta = 0}^{8} |x^\theta\rangle\langle x^\theta|$. We obtain the optimal $p_0, p_1$ and $p_2$ corresponding to guessing correctly $x_0, x_1$ and $x_2$ as 
\begin{align}
    p_0 \approx 0.62, p_1 \approx 0.69, p_2 \approx 0.77.
\end{align}
More details of the construction of MUBs and the calculation of probabilities are given in Section~\ref{ap:mub8}. We simulate the modelling of XORPUFs under Split Attack in Figure \ref{fig:mub_k5}.

It takes up to $10^6$ CRPs to model the underlying CPUFs accurately. The required number of CRPs to model the underlying $k=5$ CPUFs in 8-dimension encoding is the same as BB84 encoding with less input space with 32 bits challenge size. In the HLPUF authentication protocol, it means a longer usage period with the same hardware. However, the MUB in an 8-dimension encoding setting (or high dimensions) requires multi-qubit gates on both the server and client sides. Hence, there is a trade-off between the complexity of encoding and implementation effort. Furthermore, we should consider the imperfect quantum channels and measurements with the HPUF setting. We leave these as one of our benchmarking works in the future. 

\section{Limitations of Lockdown Technique for Generic Quantum PUFs}
\label{ap:limitation}
In this section, we study for the first time, the possibility of exploiting the lockdown technique for quantum PUFs (QPUFs), and we demonstrate the mathematical model for it. It is also worth mentioning that implementing QPUFs in practice is challenging and subject to current research. Some constructions have been proposed for constructing fully secure unitary QPUFs such as~\cite{KMK21}, but they are usually resourceful quantum constructions. Also, some other classes of QPUFs, namely quantum-readout PUFs~\cite{S12}, have been defined in weaker attack models and under restricted quantum adversaries. Apart from the theoretical aspect of the problem, it is also interesting to see whether the lockdown technique can help to reduce the adversarial power in the quantum case. One of the main problems in the case of QPUFs is that if an adversary manages to query a QPUF with the same input multiple times, then such an adversary can get multiple copies of the same output state. This allows the adversary to use the tools from the quantum state tomography~\cite{DLP2001}, and the quantum emulation algorithm to emulate the input-output behavior~\cite{ML16} of the target QPUF. One possible way to protect it from such sophisticated attacks is to use the lockdown technique. The main goal of such a lockdown technique is to prevent the adversary from querying in an adaptive manner, with arbitrary challenges.

Similarly to the hybrid PUF setting, an important feature of the lockdown technique on QPUFs is the equality test of unknown quantum states for verification. As introduced previously, the verification algorithm can be efficiently implemented by SWAP test~\cite{buhrman_quantum_2001} if two states $\rho_1, \rho_2$ are two pure states. With this constraint in mind, we prove that only very restricted QPUFs can be efficiently constructed as a quantum-locked PUF (QLPUF) with a verification algorithm.

\begin{theorem}
The construction of QLPUF with verification algorithm can be achieved if and only if the input/output mapping of the targeted quantum PUF $\E: \h^{d_{\inp}}: \rightarrow \h^{d_{\out_1}}\otimes \h^{d_{\out_2}}$ is of the form $|\psi_{\inp}\rangle\langle \psi_{\inp}| \mapsto |\psi_{\out}\rangle_{S^1}\langle \psi_{\out}| \otimes |\psi_{\out}\rangle_{S^2}\langle \psi_{\out}|$. Otherwise, such a lockdown technique is incapable of quantum PUFs.
\end{theorem}

\begin{proof}
The proof is twofold. For a quantum PUF $\E: \h^{d_{\inp}}: \rightarrow \h^{d_{\out_1}}\otimes \h^{d_{\out_2}}$ that maps an input state $|\psi_{\inp}^i\rangle_{S_i}\langle \psi_{\inp}^i| \in \h^{d_{\inp}}$ to an output state $|\psi_{\out}^i\rangle_{S_i^1S_i^2}\langle \psi_{\out}^i| \in \h^{d_{\out_1}} \otimes \h^{d_{\out_2}}$ with subsystem $S^1$ and $S^2$. The mapping of the QLPUF $\E_L: \h^{d_{\inp}} \otimes \h^{d_{\out_1}} \rightarrow \h^{d_{\out_2}} \otimes \h^{\perp}$ corresponding to a quantum PUF $\E$ is defined as follows:
\begin{equation}
    \label{eq:QLPUF}
    |\psi_{\inp}^i\rangle_{S_i}\langle \psi_{\inp}^i| \otimes \tilde \rho_{S_i^1} \rightarrow 
    \begin{cases}
    & \rho_{S_i^2}~~~~\text{if } \texttt{Ver}(\rho_{S_i^1}, \tilde \rho_{S_i^1})=1\\
    & \perp ~~~~~\text{otherwise.}
    \end{cases}
\end{equation}
where $\rho_{S_i^1}=\tr_{S_i^2}\left[|\psi_{\out}^i\rangle_{S_i^1S_i^2}\langle \psi_{\out}^i|\right]$ and $ \rho_{S_i^2}=\tr_{S_i^1}[|\psi_{\out}^i\rangle_{S_i^1S_i^2}\langle \psi_{\out}^i|]$. 

According to such construction, the QLPUF takes the input $|\psi_{\inp}^i\rangle_{S_i}\langle \psi_{\inp}^i| \otimes \tilde \rho_{S_i^1}$. Among the two input states, the QLPUF uses $|\psi_{\inp}^i\rangle_{S_i}\langle \psi_{\inp}^i|$ to get an output state $|\psi_{\out}^i\rangle_{S_i^1S_i^2}\langle \psi_{\out}^i|$. The QLPUF outputs a state $\rho_{S_i^2}$ if $\rho_{S_i^1}$ is same as the state $\tilde \rho_{S_i^1}$. Otherwise, it outputs an abort state $\perp$. We refer to Figure \ref{fig:qlpuf_limited} for the circuit of the QLPUF. Note that the QLPUF needs to check internally whether $\rho_{S_i^1} = \tilde \rho_{S_i^1}$ or not. If $\rho_{S_i^1}$ is a pure state then we can use the SWAP test to check the equality of two pure states. The circuit of the SWAP test makes the circuit of the entire QLPUF efficient. 


\begin{figure}[!h]
    \centering
    \begin{tikzpicture}[
    node distance = 3mm, every node/.style = {rectangle, rounded corners, align=center, scale=0.95}]
    \node (puf) [draw, inner xsep=1mm, inner ysep=3mm] {QPUF $\E$};
    \node [coordinate, right=1.5cm of puf] (ADL){};
    \node (first) [below=1cm of ADL, draw, inner xsep=1mm, inner ysep=3mm] {$\texttt{Ver}(\rho_{S_i^1}, \tilde \rho_{S_i^1})$};
    \node (second) [right=4cm of puf,draw, inner xsep=1mm, inner ysep=3mm] {Output};
    \node[name=outer1, dashed, fit=(puf) (first) (second), draw, inner xsep=3.5mm, inner ysep=3mm] {};
    \node (input_1)[left=0.3cm of puf] {$\ket{\psi_{\inp}^i}_{S_i}\bra{\psi_{\inp}^i}$};
    \node (output_1) [right=0.3cm of puf, inner ysep=1mm] {$\ket{\psi_{\out}^i}_{{S_i^1}{S_i^2}}\bra{\psi_{\out}^i}$};
    \node (input_2)[left=2.2cm of first] {$\tilde \rho_{S_i^1}$};
    \node (output_2) [right=0.4cm of second] {$\rho_{S_i^2}$};
    \node (output_3) [below=0.1mm of output_2] {$/\perp$};
    \draw[-latex'] (input_1) -- (puf); 
    \draw[-latex'] (puf) -- (output_1); 
    \draw[-latex'] (output_1) -- (second); 
    \draw[-latex'] (output_1.south -| first.north) -- (first); 
    \draw[-latex'] (input_2) -- (first);
    \draw[-latex'] (second) -- (output_2);
    \draw[-latex'] (first) -| node[pos=0.7,right,font=\footnotesize] {b} (second);
    \end{tikzpicture}
    \caption{Construction of QLPUF $\E_L$ with quantum PUF $\E: \h^{d_{\inp}}: \rightarrow \h^{d_{\out_1}}\otimes \h^{d_{\out_2}}$}
    \label{fig:qlpuf_limited}
\end{figure}

On the other hand, however, in the case when the quantum channel $\E$ of the quantum PUF can have entangling power and hence the subsystems $S^1$ and $S^2$ that represent the different parts of the response, may be entangled. Let's start from the simple situation with a 2-qubit entangled state as $|\psi_{\out}^i\rangle\langle \psi_{\out}^i|$. i.e., for a quantum PUF $\E$ that maps an input state $|\psi_{\inp}^i\rangle\langle \psi_{\inp}^i|$ to an entangled output state $|\psi_{\out}^i\rangle\langle \psi_{\out}^i| := (\alpha\ket{a_1^i}\ket{b_1^i}+\beta\ket{a_2^i}\ket{b_2^i})(\alpha^\ast\bra{a_1^i}\bra{b_1^i}+\beta^\ast\bra{a_2^i}\bra{b_2^i})$ where \(\abs{\alpha}^{2}+\abs{\beta}^{2}=1\), $\ket{a_1}$ and $\ket{a_2}$ are any two vectors in the space of subsystem $S^1$, and $\ket{b_1}$ and $\ket{b_2}$ are any two vectors in the space of subsystem $S^2$. Consider a POVM measurement on the subsystem $S^1$ with $m$ elements $\{E_m\}$ where $\Sigma_m E_m=I$
, the reduced density operator of $S^2$ after tracing out $S^1$ is: 
\begin{align}
    \label{eq:reduced}
    \nonumber
    \rho_{S_i^2}&=\Sigma_m\tr_{S_i^1}[\tr(|\psi_{\out}^i\rangle_{S_i^1S_i^2}\langle \psi_{\out}^i|E_m)]\\
    \nonumber
    &=\Sigma_m\tr_{S_i^1}[\langle \psi_{\out}^i|E_m|\psi_{\out}^i\rangle_{S_i^1S_i^2}]\\
    &=|\alpha|^2\ket{b_1^i}\bra{b_1^i}+|\beta|^2\ket{b_2^i}\bra{b_2^i}
\end{align}

The state of subsystem $S^2$ is clearly a mixed state. However, checking the equality between two mixed states is difficult, and sometimes not possible. For example, we have two different mixed states:
\begin{equation}
    \ket{\psi_1^i}= 
    \begin{cases}
    \ket{b_1^i}~\text{with probability}~|\alpha|^2\\
    \ket{b_2^i}~\text{with probability}~|\beta|^2
    \end{cases}
\end{equation}
and
\begin{equation}
    \ket{\psi_2^i}= 
    \begin{cases}
    \alpha\ket{b_1^i}+\beta\ket{b_2^i}~\text{with probability}~\frac{1}{2}\\
    \alpha\ket{b_1^i}-\beta\ket{b_2^i}~\text{with probability}~\frac{1}{2}
    \end{cases}
\end{equation}
The density operators of both mixed states are represented as Equation (\ref{eq:reduced}). That is to say, these two mixed states are unequal but totally indistinguishable. This can be trivially extended to the n-qubit situation. So the lockdown technique is not implementable with generic quantum PUFs.  
\end{proof}

In the case of quantum PUFs, our study shows that some quantum mechanical properties of quantum PUFs such as entanglement generation, make it challenging to use the straightforward quantum analogue of the classical lockdown technique. However, this is still an interesting observation, because we do not need this sort of condition on encoding the output of classical PUF to construct an HPUF with the lockdown technique.

\section{Detailed Security Analysis}\label{ap:security}

\subsection{Proof of Lemma \ref{lem:guess_prob}}
\label{appn:lem1}
Here, we give a detailed proof for Lemma~\ref{lem:guess_prob}.
\begin{proof}[\textbf{Proof of Lemma~\ref{lem:guess_prob}}]
According to Construction \ref{cons:hpuf_1}, for a given $x_i$, we use the $2j$-th bit $y_{i,2j} \in \{0,1\}$ of the outcome of the CPUF to choose the basis (either $ \{|0\rangle,|1\rangle\}$-basis or $\{|+\rangle,|-\rangle\}$-basis) of the $j$-th qubit output of the HPUF. Further, we use the $y_{i,(2j-1)}\in \{0,1\}$ to choose a state from the chosen basis. Here, if $y_{i,(2j-1)} = 0$ then from an adversarial point of view, the output state is $\rho_0=(\frac{1}{2}+\delta_r)|0\rangle\langle 0| + (\frac{1}{2}-\delta_r)|+\rangle \langle +|$. Similarly,  if $y_{i,(2j-1)} = 1$ then from an adversarial point of view, the output state is $\rho_1=(\frac{1}{2}+\delta_r)|1\rangle\langle 1| + (\frac{1}{2}-\delta_r)|-\rangle \langle -|$. For the adversary, the probability of correctly guessing $y_{i,(2j-1)}$ is the same as distinguishing the two states $\rho_0,\rho_1$.  Here $\Pr[\A^{i,j}_{guess}(x_i, |\psi^{i,j}_{\out}\rangle \langle \psi^{i,j}_{\out}|)= y_{i,(2j-1)}]$ denotes the optimal probability of guessing the bit correctly. From the Helstorm-Holevo bound \cite{helstrom_quantum_1969,H1973} we get,

\begin{align}
\nonumber
    &\Pr[\A^{i,j}_{guess}(x_i, |\psi^{i,j}_{\out}\rangle \langle \psi^{i,j}_{\out}|)= y_{i,(2j-1)}] \\ 
    \nonumber &\leq p[1+\max_{E} Tr[E(\rho_0 - \rho_1)]] \\
    \nonumber &= p[1 + \frac{1}{2}\norm{\rho_0- \rho_1}_1].\\
    \nonumber &= p(1+\sqrt{p^2+(1-p)^2})\\
    &\leq p(1+\sqrt{2}p)
\end{align}

This concludes the proof.

\end{proof}

\subsection{Proof of Theorem \ref{thm:cpuf_vs_hpuf}}
\label{apn:thm1}

We show the contrapositive statement that if you can break HPUF you can also break underlying CPUF. Here we give the proof for $m=1$, and it can easily be generalised for any arbitrary integer $m > 0$.

Suppose for the HPUF, a $q$-query weak-adversary win the unforgeability game with a non-negligible probability $P(m=1,p,q)$. This implies, given a database of $q$ random challenge response from the HPUF, the adversary can produce $|\psi_{f(x^*)}\rangle$ corresponding to a random challenge $x^* \in \{0,1\}^n$ with a non-negligible probability $P(m=1,p,q)$. Note that, for the deterministic adversarial strategy, the adversary can produce multiple copies of the forged state $| \psi_{\tilde f(x^*)}\rangle$ for a random challenge $x^*$. For the random adversaries, we can produce multiple copies of the same forged state $| \psi_{\tilde f(x^*)}\rangle$ just by fixing the internal randomness parameter of the adversarial strategy. Hence, both the random and deterministic adversary can produce multiple copies of the forged state $| \psi_{\tilde f(x^*)}\rangle$ for a random challenge $x^*$. From the multiple (say $K$) such copies of $| \psi_{\tilde f(x^*)}\rangle$, the adversary will extract $\tilde{f}(x^*)$ using the following strategy.

\begin{algorithm}
\caption{Algorithm to Forge CPUF from HPUF}
\begin{algorithmic}
    \Require $K \geq 2$-copies of the forged state $| \psi_{\tilde f(x^*)}\rangle$
    \State Measure the $1$-st copy of the state $| \psi_{\tilde f(x^*)}\rangle$ in $\{|0\rangle, |1\rangle\}$-basis.
    \State Let $z_1 \in \{0,1\}$ be the measurement outcome.
    \For{$i=2$; $i \leq (K-1)$; $i++$}
    \State Measure the $i$-th copy of the state $| \psi_{\tilde f(x^*)}\rangle$ in $\{|0\rangle, |1\rangle\}$-basis.
    \State Let $z_i \in \{0,1\}$ be the measurement outcome.
    \If{$z_i \neq z_{i-1}$}
        \State \textbf{break}   \Comment{Implies $| \psi_{\tilde f(x^*)}\rangle \in \{|+\rangle,|-\rangle\}$.}
    \EndIf
    \EndFor
    \If{$i=K$}
        \State \Return $\tilde f(x^*) = (0,z_i)$
        \Else
        \State Measure the $i+1$-th copy in $\{|+\rangle, |-\rangle\}$-basis.
        \State Let $z_{i+1}$ be the measurement outcome.
        \State \Return $\tilde f(x^*) = (1,z_{i+1})$.
    \EndIf
\end{algorithmic}
\end{algorithm}

If $|\psi_{f(x^*)}\rangle = |\psi_{\tilde f(x^*)}\rangle \in \{|0\rangle, |1\rangle\}$ then in Algorithm 1 all the measurement outcomes $z_i$ (for $1\leq i \leq K$) would be the same, and $\tilde f(x^*) = f(x^*)$. However, if $|\psi_{f(x^*)}\rangle = |\psi_{\tilde f(x^*)}\rangle \in \{|+\rangle, |-\rangle\}$ then $\tilde f(x^*) \neq f(x^*)$ if and only if all the measurement outcomes $z_i$ are equal ($1 \leq i \leq K$). This happens with probability $\frac{1}{2^K}$. Therefore, we get

\begin{equation}
\label{eq:forge_f_psi_x}
    \Pr_{x^*}[\tilde f(x^*) = f(x^*)| |\psi_{f(x^*)}\rangle = |\psi_{\tilde f(x^*)}\rangle] \geq (1-\frac{1}{2^K}).
\end{equation}

If the adversary successfully forges the HPUF with a non-negligible probability $P(m=1,p,q)$ then from Equation \eqref{eq:forge_f_psi_x} we get that the adversary manages the CPUF with probability at least $P(m=1,p,q)(1-\frac{1}{2^K})$, which is also non-negligible. Therefore, if an adversary manages to win the unforgeability game for the HPUF with a non-negligible probability, then using the same forging strategy it can also win the unforgeability game for the corresponding CPUF with a non-negligible probability. This implies, if no QPT weak adversary can win the universal unforgeability game with a non-negligible probability for the CPUF then no QPT adversary can win the universal unforgeability game with a non-negligible probability for the corresponding HPUF. This concludes the proof.

\subsection{Proof of Lemma \ref{lem:database_dest1}}
\label{appn:lem2}
For a successful forgery, the adversary needs to win the universal unforgeability defined in Game \ref{game:uni-unf-abs}. This implies, using the measurement strategy $E(D_q)$ the adversary needs to produce a quantum state $\ket{\psi_{f(x^*)}}$ corresponding to a challenge $x^* \in_R \{0,1\}^n$ that is chosen uniformly at random. Without loss of generality, we can write the measurement strategy as a POVM with two outcomes $E(D_q) = \{E_{\forge}(D_q,x^*),E_{\fail}(D_q,x^*)\}$, where $E_{\forge}(D_q,x^*), E_{\fail}(D_q,x^*)$ denote the measurement operators corresponding to the successful forgery and the failure forgery respectively. Therefore, we can write the successful forging probability $p_{\forge}$ as follows.
\begin{equation}
    p_{\forge} = \tr[E_{\forge}(D_q,x^*)\rho_{D_q}^{x^*}],
\end{equation}
where $\rho_{D_q}^{x^*} := \ket{D_q}\bra{D_q} \otimes \ket{x^*}\bra{x^*}\otimes \ket{0^m}_{out}\bra{0^m}$. Here the $out$ register would contain the forged state. If we write $E_{\forge}(D_q,x^*) = M^{\dagger}_{\forge}(D_q,x^*)M_{\forge}(D_q,x^*)$, then we can rewrite the post-measurement state corresponding to the successful forgery as follows:
\begin{equation}\label{eq:post_state}
\begin{split}
        & \frac{M_{\forge}(D_q,x^*)\ket{D_q} \otimes \ket{x^*} \otimes \ket{0^{m'}}_{out}}{\sqrt{p_{\forge}}} \\
        & = \frac{\ket{\tilde D_q}_R \otimes \ket{x^*} \otimes \ket{\psi_{f(x^*)}}_{out} \otimes \ket{\tilde a}_{out}}{\sqrt{p_{\forge}}},
\end{split}
\end{equation}
where $|\tilde D_q\rangle_R$ denotes the post-measurement database state, and $\ket{\tilde a}_{out}$ is the post-measurement state of the ancillary system which is a $(m'-m)$ dimensional state while as $\ket{\psi_{f(x^*)}}_{out}$ is $m$ dimensional. As $\bigotimes_{i=1}^q|x_i\rangle_C$ is a classical state, in the rest of the proof we don't write them in the expressions. 

Using the \emph{Neimark's theorem} we can replace the POVM measurement strategy $E(D_q)$ with the combination of a unitary acting on an extended system including an ancilla $\ket{anc}_A$, followed by a projective measurement. Let us denote the unitary as $U^{x^*}_{D_q}$ which couples the input state $\ket{D_q} \otimes \ket{0^{m'}}_{out}$ with the ancillary system $|anc\rangle_A$, and let $\{|v\rangle\}$ be the basis on which the projective measurement is applied to the ancilla. We first rewrite the impact of the unitary $U^{x^*}_{D_q}$ on the input state:
\begin{align}\label{eq:application-proof-neimark-unit}\nonumber
    & U^{x^*}_{D_q}\left(\bigotimes_{i=1}^q \ket{\psi_{f(x_i)}}_R \otimes \ket{0}_{out} \otimes \ket{anc}_A\right) \\
    \nonumber
    & = U^{x^*}_{D_q}\left( \ket{\Psi^q_{f}}_R \otimes \ket{0}_{out} \otimes \ket{anc}_A\right) \\ 
    & = \sum_{v} \sqrt{p_v} \ket{\Psi^q_v}_R \otimes \ket{\tilde \psi_v}_{out} \otimes \ket{v}_{A}.
\end{align}
where in the second line we have rewritten everything after applying the unitary in the $\{\ket{v}\}$-basis. Now, the adversary performs a projective measurement on the state \eqref{eq:application-proof-neimark-unit} in this basis. Suppose for the correct forgery, the ancilla is projected into the $\ket{v_{\forge}}_A$ state. Therefore we can rewrite the expression of $p_{\forge}$ as follows:
\begin{equation}
    p_{\forge} = \sum_{v: v = v_{\forge}} p_v |\mbraket{v_{\forge}}{v}|^2.
\end{equation}
Overall, following this strategy, the purification of the adversary's post-measurement state with an optimal POVM measurement can be written as the following:
\begin{equation}\label{eq:application-post-state-neimark}
    \frac{\ket{\tilde D_q}_R \otimes \ket{x^*} \otimes \ket{\psi_{f(x^*)}}_{out} \otimes \ket{v_{\forge}}_A}{\sqrt{p_{\forge}}}, 
\end{equation}
where $\ket{\tilde D_q}$ denotes the post-measurement database state. Note that, due to Neimark's theorem the post-measurement database states in Equation \eqref{eq:post_state}, and \eqref{eq:application-post-state-neimark} are the same, if the same ancillary system has been assumed after the purification and POVM, \emph{i.e.} if $\ket{v_{\forge}}_A = \ket{\tilde a}_{out}$.

Now, let us use the unitary $U^{x^*}_{D_q}$ and the measurement basis $\{\ket{v}\}$ to construct a \emph{measure-then-forge} strategy. As the unitary $U^{x^*}_{D_q}$ only depends on the input $x^*$ and $D_q$, we can rewrite it in the basis that is diagonalised with respect to the states $\{\ket{\Psi^q_v,v}\}_v$.

For the post-measurement state $\ket{v_{\forge}}$, of the ancilla, the adversary applies $U^{x,x^*}_{D_q,\Psi^q_{\forge},v_{\forge}}$ on the $\ket{0}_{out}$ register. Note that, the adversary doesn't have any information about the $\{f(x_i)\}_{1 \leq i \leq q}$ before measuring the ancillary sub-system in the $\{\ket{v}\}$-basis. Hence, the measurement basis $\{\ket{v}\}$ choice only depends on the classical challenges $x_i$'s and $x^*$. Therefore, the adversary can use the same information to find the $\{\ket{v}\}$-basis, and first performs the measurement on the $RA$ register in $\{\ket{\Psi^q_v,v}\}$-basis, and obtains the state $\ket{\Psi^q_{\forge},v_{\forge}}$ with the same probability $p_{\forge}$. After the measurement, the adversary applies the unitary $U^{x^*}_{D_q,\Psi^q_{\forge},v_{\forge}}$ on $\ket{0}_{out}$, and get the forged state $\ket{\psi_{f(x^*)}}$. Therefore, with this strategy, the adversary also wins the unforgeability game with the probability $p_{\forge}$.

Note that, there always exists a unitary $U$ such that $U(\bigotimes_{i=1}^q\ket{\tilde{f}(x_i)}) \otimes \ket{anc} = \ket{\Psi^q_{\forge},v_{\forge}}$, where $\tilde f(x_i)$ denotes the extracted information about $f(x_i)$'s from the encoded database $\ket{D_q}$. Therefore, from any generalised measurement strategy $E(D_q)$ we can construct a strategy for the measure-then-forge protocol that can win the universal unforgeability game with the same probability $p_{\forge}$. This concludes the proof.

\subsection{Proof of Lemma \ref{lem:pguess}}

In this lemma, we give an upper bound on the probability of extracting the CPUF outcomes from the $(1-\varepsilon)q$ out of $q$ responses of the HPUF. Let $\A_h$ be a quantum adversary who plays the unforgeability game against the HPUF. $\A_h$ has access to $q$ queries of the HPUF as $q$ pairs of $\{(X_i, \ket{\psi_{f(X_i)}})\}^q_{i=1}$. Note that, according to the construction \ref{cons:hpuf_1}, $|\psi_{f(X_i)}\rangle\langle \psi_{f(X_i)}|=\bigotimes_{j=1}^{2m} |\psi^{i,j}_{f(X_i)}\rangle\langle\psi^{i,j}_{f(X_i)}|$, where $|\psi^{i,j}_{f(X_i)}\rangle \in \{|0\rangle, |1\rangle, |+\rangle, |-\rangle\}$. As the state in the adversary's possession depends fully on a classical string, we can describe this situation using a classical-quantum state, where the $C$ register contains the classical string $f(X_i)$, and the $S$ register contains the quantum state $|\psi_{f(X_i)}\rangle\langle \psi_{f(X_i)}|$. We assume the $j$-th bit of the string $f(X_i)$ as $Y_{i,j}$. The classical-quantum state for the $j$-th qubit is of the following form.

\begin{equation}
    (\rho_{CS})_{j} = \hspace{-0.15in} \sum_{\substack{Y_{2j-1}, \\ Y_{2j} \in \{0,1\}}} \hspace{-0.15in} \frac{1}{4} |Y_{2j-1,2j}\rangle_{C_{i,j}}\langle Y_{2j-1,2j}| \otimes |\psi^{i,j}_{f(X_i)}\rangle\langle \psi^{i,j}_{f(X_i)}|.
\end{equation}

In Lemma \ref{lem:guess_prob}, we prove that the probability of guessing $Y_j$ is $p_{\guess}$, and it has the following upper bound.

\begin{equation}
\label{eq_appen:pguess}
    p_{\guess}\leq p(1+\sqrt{2}p).
\end{equation}

In Section \ref{sec:assump}, we assume that all the output bits of the CPUF are i.i.d. Therefore the entire classical-quantum state for the $i$-th challenge $X_i$ is $\rho_{CS}$ of the following form. 

\begin{equation}
    \rho_{CS} = \bigotimes_{j=1}^{m} (\rho_{CS})_{j}.
\end{equation}

Therefore, the probability of guessing $f(X_i)$ from the $S$ subsystem is upper bounded by 
\begin{align}
    (p_{\guess})^{2m}.
\end{align}
Let $\rho_{C^qS^q}$ denote the joint state shared between the server and the $q$-query weak adversary. Due to the i.i.d assumption on all the outputs of the underlying classical PUF of the HPUF, $\rho_{C^qS^q}$ has the following form. 

\begin{equation}
    \rho_{C^qS^q} = \left(\bigotimes_{j=1}^{m} (\rho_{CS})_{j}\right)^{\otimes q}.
\end{equation}

Here, we would like to find an upper bound on the probability of successfully guessing $f(X_i)$'s for at least $(1-\varepsilon)q$ responses out of $q$ responses. We denote this guessing probability as $p_{\extract}$. Note that, due to the i.i.d assumption on the different outcomes of the CPUF, the adversary's success probability of guessing exactly $k$ responses out of $q$ responses is upper bounded by $\binom{q}{k}(p_{\guess})^{2mk}(1-(p_{\guess})^{2m})^{q-k}$. Therefore, we can re-write the expression of $p_{\extract}$ as follows,

\begin{equation}
    \label{eq_app:pext_pril}
    p_{\extract} \leq \sum_{k=(1-\varepsilon)q}^{q}\binom{q}{k} (p_{\guess})^{2mk}(1-(p_{\guess})^{2m})^{q-k}.
\end{equation}

This concludes the proof.

\subsection{Proof of Theorem \ref{thm:hpuf_auth}}

At the $i$-th round, the HLPUF $\E^L_f$ receives the queries of the form $(x_i, \tilde \rho_1)$, where the classical string $x_i \in \{0,1\}^n$, and $\tilde \rho_1 \in (\h^{2})^{\otimes m}$. The HLPUF returns $\E_{f_2}(x_i)$ if $\ver(\tilde \rho_{1}, \E_{f_1}(x_i)) =1$, otherwise it returns an abort state $\ket{{\perp}}\bra{\perp}$ corresponding to $\perp$. Hence, to get any non-abort state $\ket{\perp}$ from the HLPUF, the adaptive adversaries $\Aad$ need to produce a query of the form $(x_i, \E_{f_1}(x_i))$. As the adversary doesn't have any direct access to the mapping $\E_{f_1}$, the only way it can get any information about $\E_{f_1}(x_i)$ by intercepting the challenges that are sent by the server to the client. Suppose that the adaptive adversary has access to a set of $q$ queries $X_{[q]} := \{X_i\}_{1\leq i \leq q}$ and the corresponding responses $\Psi_{[q]} := \{\E_{f_1}(x_i)\}_{1\leq i \leq q}$. Here each $X_i$ follows a uniform distribution over the challenge set $\{0,1\}^n$. Hence, for the mapping $\E_{f_1}$ the power of the adaptive adversary reduces to the power of a weak adversary. As $\E_{f_1}$ has the universal unforgeability property against any $q$-query weak adversary, hence we get, for any random challenge $X \not \in X_{[q]}$,

\begin{align}
\label{eq:unforge_e1_1}
    \nonumber
    &\Pr_{X,X_{[q]}}[1 \leftarrow \Gea(\Aad, m,X,X_{[q]})]\\
    &=\Pr_{X,X_{[q]}}[1 \leftarrow \Gea(\Ana, m, X, X_{[q]})]\leq\epsilon(m).
\end{align}

This implies, using the set of challenges $X_{[q]}$ and responses $\Psi_{[q]}$ the adversary cannot produce the response corresponding to a random challenge $X \not \in X_{[q]}$. 
Suppose from the query set $X_{[q]}$ and the responses, the adaptive adversary successfully generates a set $X'_{[q']}$ of $q'$ adaptive queries, and corresponding responses $\Psi_{[q']}$ for the HLPUF $\E^L_f$. Without any loss of generality, we assume that for all of the queries, $X'_i \in X'_{[q']}$ the HLPUF returns a non-abort state.  

We assume that the adaptive adversary wins the universal unforgeability game using the query set $ X_{\text{ad}} = X_{[q]} \cap X'_{[q']}$. This implies,

\begin{equation}
    \label{eq:assump_el}
    \Pr_{X,X^{\E_L}_{[q]_{\text{ad}}}}[1 \leftarrow \Gel(\Aad, m, X, X_{\text{ad}})] \geq  \text{non-negl}(m).
\end{equation}

From the construction of our HLPUF in Construction \ref{cons:hlpuf} we get that winning the universal unforgeability game with the HLPUF $\E^L_{f}$ implies winning the universal unforgeability with $\E_{f_2}$. Hence, we can rewrite Equation \eqref{eq:assump_el} in the following way,

\begin{equation}
    \label{eq:assump_e2}
    \Pr_{X,X_{\text{ad}}}[1 \leftarrow \Geb(\Aad, m, X, X_{\text{ad}})] \geq  \text{non-negl}(m).
\end{equation}

Note that, if the adaptive adversary manages to get non-abort outcomes from the HLPUF corresponding to all $X'_i \in X_{\text{ad}}$ then from the Construction \ref{cons:hlpuf} we get, $1 \leftarrow \Gea(\Aad, m, X'_i, X_{\text{ad}})$. Due to the unforgeability assumption of Equation \eqref{eq:unforge_e1_1} we get,

\begin{align}
\nonumber
    &\Pr_{X,X_{[q]}}[1 \leftarrow \Gea(\Ana, m,X,X_{[q]})] \\ \label{eq:assump_e1_2}
    &= \Pr_{X,X_{\text{ad}}}[1 \leftarrow \Gea(\Aad, m, X, X_{\text{ad}})] \leq   \epsilon(m).
\end{align}

Note that, the main difference between adaptive and weak adversaries lies in the choice of the query set. If we fix the query set $X_{\text{ad}}$, then the both adaptive $\Aad$ and a weak adversary can extract the same amount of information from the responses corresponding to the query set $X_{\text{ad}}$. Therefore, their winning probability of the universal unforgeability game becomes equivalent. This implies, we can rewrite Equation \eqref{eq:assump_e1_2} in the following way,

\begin{align}
   \nonumber
    & \Pr_{X,X_{\text{ad}}}[1 \leftarrow \Gea(\Aad, m, X, X_{\text{ad}})] \\  \label{eq:assump_e1_3}
    &=   \Pr_{X,X_{\text{ad}}}[1 \leftarrow \Gea(\Ana, m, X, X_{\text{ad}})]\leq \epsilon(m).
\end{align}

By combining Equation \eqref{eq:assump_e1_2} and Equation \eqref{eq:assump_e1_3} we get, both the random variables $X_{[q]}$ and $X_{\text{ad}}$ are equivalent. 
From the universal unforgeability property of the PUF $\E_{f_2}$ against any $q$-query weak adversary, we get

\begin{align}
\label{eq:unforge_e2_1}
     \Pr_{X,X_{[q]}}[1 \leftarrow \Geb(\Ana, m, X, X_{[q]})] \leq  \epsilon(m).
\end{align}
As both of the random variables $X_{[q]}$ and $X_{\text{ad}}$ are equivalent, so we get, 

\begin{align}
    \label{eq:final_thm_2}
    \nonumber
    & \Pr_{X,X_{[q]}}[1 \leftarrow \Geb(\Ana, m, X, X_{[q]})] \\
    \nonumber
    & = \Pr_{X,X_{\text{ad}}}[1 \leftarrow \Geb(\Ana, m, X, X_{\text{ad}})]\\
    &= \Pr_{X,X_{\text{ad}}}[1 \leftarrow \Geb(\Aad, m, X, X_{\text{ad}})] \leq \epsilon(m).
\end{align}

The second equality follows from the fact that for a fixed query set $X_{\text{ad}}$ the adaptive adversary $\Aad$ and weak adversary $\Ana$ become equivalent. Note that, only one of Equation \eqref{eq:assump_e2} and Equation \eqref{eq:final_thm_2}  is true. The Equation \eqref{eq:final_thm_2} is true because of the unforgeability of $\E_{f_2}$. Hence, our assumption of Equation \eqref{eq:assump_e2} is wrong. Therefore, Equation \eqref{eq:assump_el} is also not true. Hence, with the proof by contradiction, we get,  
\begin{equation}
    \label{eq:true_el}
    \Pr_{X,X_{\text{ad}}}[1 \leftarrow \Gel(\Aad, m, X, X_{\text{ad}})] \leq  \epsilon(m).
\end{equation}

This concludes the proof.

\subsection{Challenge reusability Proof}\label{ap:ch-reusability}
In this subsection, we give a detailed security analysis and proof for the challenge reusability discussed in Section~\ref{ap:reusability}. First, we introduce the tools and uncertainty relation that we need for the proof mostly from~\cite{coles2017entropic}, then we give the formal statement and proof for Theorem~\ref{th:uncertainty}.

Heisenberg's uncertainty principle is one of the most important fundamental properties of quantum mechanics which is mathematically speaking due to the non-commuting property of some observables like Pauli $X$ and $Z$ measurements. Reformulating these relations in terms of entropic quantities has been very useful in the foundations of quantum information and has also been widely used in the security proofs of different quantum communication protocols such as QKD. The most well-known uncertainty relation for these operators was given by Deutsch~\cite{deutsch1983uncertainty} and later improved~\cite{maassen1988generalized} as follows:

\begin{equation}\label{eq:uncert-xz}
  H(X) + H(Z) \geq \log_2 (\frac{1}{c})  
\end{equation}

where $c$ denotes the maximum overlap between any two eigenvectors of $X$ and $Z$. Usually, a quantum system $A$ is considered where the state is described with the density matrix $\rho_A$ on a finite-dimensional Hilbert space. If the measurement is performed in a $X$ and $Z$ basis (or equivalently any other MUB bases), then the measurements are just projective operators that project the state into the subspace spanned by those bases. In the most general case, the measurements are a set of POVM operators on system $A$ denoted as $\{M^x\}_x$ and $\{N^z\}_z$ where the general Born rule states that the probability of obtaining outcomes $x$ and $z$ to be as follows:

\begin{equation}
  P_X(x) = tr[\rho_A M^x] \quad , \quad  P_Z(z) = tr[\rho_A N^z]
\end{equation}

In this case, the Equation~(\ref{eq:uncert-xz}) still gives the generalised uncertainty relation with the difference that the $c$ is defined as follows:

\begin{equation}
  c = \max_{x,z} c_{zx}, \quad \text{and} \quad c_{xz} = \parallel \sqrt{M^x}\sqrt{N^z} \parallel^2
\end{equation}

where $\parallel\cdot\parallel$ denotes the operator norm (or infinity norm). The above uncertainty relation can be extended to conditional entropy as well in the context of guessing games~\cite{coles2017entropic}. Assume two parties, Alice and Bob, where Bob prepares a state $\rho_A$ and Alice randomly performs the $X$ and $Z$ measurements leading to a bit $K$. Then Bob wants to guess $K$ given the basis choice $R=\{0,1\}$. The conditional Shannon entropy is defined as follows:

\begin{equation}
  H(K|R) := H(KR) - H(R)
\end{equation}
 Thus one can get the same uncertainty relation with the conditional entropy as:
 \begin{equation}
  H(K|R=0) + H(K|R=1) \geq \log_2(\frac{1}{c})
\end{equation}
We also have the quantum equivalent of Shannon entropy for mixed quantum state called von Neumann entropy, which is defined as $H(\rho) = tr[\rho log (\rho)] = -\sum_i \lambda_i \log_2 (\lambda_i)$ where $\lambda_i$ are the eigenvalues of $\rho$. Similar, to the classical case, for a bipartite system $\rho_{AB}$ the conditional von Neumann entropy is defined as follows:

\begin{equation}
  H(A|B) := H(\rho_{AB}) - H(\rho_B)
\end{equation}

Furthermore, this can be generalised to any tripartite quantum system with state $\rho_{ABC}$. An interesting property here is an inequality referred to as \emph{data processing inequality}~\cite{coles2017entropic} which states that the uncertainty of $A$ conditioned on some system $B$ never goes down if $B$ performs a quantum channel on the system. In other words for any tripartite system $\rho_{ABC}$ where system $C$ will perform a quantum operation on the quantum state in order to extract some information, we have the following:

\begin{equation}
  H(A|BC) \leq H(A|B)
\end{equation}

Given the above inequality leads to the general uncertainty relations between any tripartite system including two parties Alice and Bob, and an eavesdropper Eve:

\begin{equation}
  H(K|ER) + H(K|BR) \geq \log_2\left(\frac{1}{c}\right)
\end{equation}

Where $K$ is the measurement output and $R$ is the basis bit. This imposes a fundamental bound on the uncertainty in terms of von Neumann entropy, in other words, the amount of information that an eavesdropper can extract from the joint quantum systems shared between the three parties. These inequalities can also be extended to the case where $n$ bits are encoded in $n$ quantum states where $R^n$ and $K^n$ are bit-strings denoting the basis random choices for the qubits and measurement outputs respectively, and $B^n$ denotes Bob's bit-string. Also, $E$ denotes Eve's system, a general quantum system operating on $n$-qubit messages and any arbitrary local system. We have the following inequality, which is the main result that we will use in the proof of the next theorem:

\begin{equation}
  H(K^n|ER^n) + H(K^n|B^nR^n) \geq n \log_2(\frac{1}{c})
\end{equation}

Now we are ready to give a more formal version of the Theorem~\ref{th:uncertainty} and the proof.

\begin{theorem}\label{th:unc-formal}
In Protocol~\ref{prot:hpuf}, let $x$ be a challenge and $(y_{1},\dots,y_{2m})$ be the response of a classical PUF used inside the HPUF construction, with randomness bias $p = (\frac{1}{2} + \delta_r)^{2m}$ in generating the random classical responses. If the verification algorithm for a state $\tilde{\rho}$ passes with probability $1 - \epsilon(m)$, then Eve's conditional min-entropy $H_{min}^{Eve}$ in terms of von Neumann entropy over the server's (or client's) classical response, satisfies the following inequality:
\begin{equation}
  H_{min}^{Eve} = H_{min}(S^m|ER^m) \geq m - \epsilon(m)  
\end{equation}
\end{theorem}
\begin{proof}
We prove this theorem based on the first half of the state used in Protocol~\ref{prot:hpuf}, i.e., the state $\ket{\psi_{f_1(x_i)}}\bra{\psi_{f_1(x_i)}}=\bigotimes_{j=1}^m |\psi^{i,j}_{f_1(x_i)}\rangle\langle\psi^{i,j}_{f_1(x_i)}|$ that is being sent by the Server (S) and received and measured by the Client (C). Nevertheless, the same proof applies to the second state due to the symmetry of the states and the protocol. 

Let $R^m = (R_1,\dots,R_m)$ be the randomness bitstring showing the choice of the basis encoding of the response, $S^m = (S_1,\dots,S_m)$ be the server's bit encoded in the $R^m$ bases. Note that both $R^m$ and $S^m$ are produced according to the bitstring $(y_{1},\dots,y_{2m})$ which is the first half of the response of CPUF to a given challenge $x$. Also, let $C^m = (C_1,\dots,C_m)$ be the client's correct bit string. We denote the arbitrary joint state of three systems by $\rho_{S^m E C^m}$ where $E$ denotes any arbitrary quantum system held by the eavesdropper. Now, let the Client's measurement outcomes, after the verification be $\tilde{Y}^m = (\tilde{Y_1},\dots,\tilde{Y_m})$ which shows the estimated bits by the Client. Now we can write the tripartite uncertainty principle, in terms of the von Neumann entropy, for MUB measurements and MUB states as follows:
\begin{align}\nonumber
    & H(X_1 X_2 Z_3 X_4 \dots X_{m-1} Z_m | E) + \\ \label{eq:uncert-tripartite}
    &H(Z_1 Z_2 X_3 Z_4 \dots Z_{m-1} X_m | C) \geq \log_2 \left(\frac{1}{c}\right)^m
\end{align}
where $c = \max_{x,z} c_{xz}$ and $c_{xz} = \parallel \sqrt{M^x}\sqrt{N^z} \parallel^2$ for an arbitrary POVM sets $M = \{M^x\}_x$ and $N = \{N^z\}_z$. We note that if the CPUF creates a perfect random bitstring for $R^m$ then states are perfect MUB states and $c = \frac{1}{2}$. Nonetheless, we consider a weaker CPUF with a biased distribution of $p = (\frac{1}{2} + \delta_r)^{2m}$ in creating $0$s and $1$s in the response. Hence, we can translate this imperfectness into a disturbance in the measurement bases. Let $M^0 = \ket{0}\bra{0}$ and $M^1 = \ket{1}\bra{1}$ be the usual measurement in the computational basis but let the $N$ measurements be a slightly shifted version of the measurements in the $X$ basis. Consider the following states:
\begin{equation}
\begin{split}
        & \ket{\psi_N} = \sqrt{\frac{1}{2} + \delta_r}\ket{0} + \sqrt{\frac{1}{2} - \delta_r}\ket{1} \\
        & \ket{\psi^{\perp}_N} = \sqrt{\frac{1}{2} - \delta_r}\ket{0} - \sqrt{\frac{1}{2} + \delta_r}\ket{1} 
\end{split}
\end{equation}
We define the new $N$ projective operators according to the following states as $N^0 = \ket{\psi_N}\bra{\psi_N}$ and $N^1 = \ket{\psi^{\perp}_N}\bra{\psi^{\perp}_N}$. Now we calculate the operator norm for all the pairs of measurements and we have:
\begin{equation}
\begin{split}
        &  \parallel\sqrt{M^0}\sqrt{N^0}\parallel^2 = \frac{1}{2} + \delta_r, \quad  \parallel\sqrt{M^0}\sqrt{N^1}\parallel^2 = \frac{1}{2} - \delta_r\\
        &  \parallel\sqrt{M^1}\sqrt{N^0}\parallel^2 = \frac{1}{2} - \delta_r, \quad  \parallel\sqrt{M^1}\sqrt{N^1}\parallel^2 = \frac{1}{2} + \delta_r
\end{split}
\end{equation}
Thus we conclude that $c = \frac{1}{2} + \delta_r$ and the Equation~\eqref{eq:uncert-tripartite} can be re-written as follows:
\begin{align}
    \nonumber
    & H(X_1 X_2 Z_3 X_4 \dots X_{m-1} Z_m | E) + \\ & H(Z_1 Z_2 X_3 Z_4 \dots Z_{m-1} X_m | C) \geq m  - m\log_2 (1 + 2\delta_r)
\end{align}

Now, as mentioned at the beginning of the section, using the data processing inequality~\cite{coles2017entropic}, we have got the following security criteria that show Eve's uncertainty (in terms of the von Neumann entropy) of the actual response bits $S^m$:
\begin{equation}
    H(S^m | ER^m) + H(S^m | \tilde{Y}^m) \geq m  - m\log_2 (1 + 2\delta_r).
\end{equation}
We can get the same inequality in terms of smooth min and max entropy~\cite{coles2017entropic,tomamichel2011uncertainty}, which is more appropriate for ensuring the security in the finite size, for min and max entropy we equivalently have:

\begin{equation}\label{eq:sec-criteria-uncert}
    H^{\epsilon}_{min}(S^m | ER^m) \geq m  - H^{\epsilon}_{max}(S^m | \tilde{Y}^m) - m\log_2 (1 + 2\delta_r)
\end{equation}
In order to calculate the above bound, we need to find the bound on the $H^{\epsilon}_{max}(S^m | \tilde{Y}^m)$. Here we use another result from~\cite{tomamichel2011uncertainty} where it states that for any bitstring $X$ of $n$ bit and the respective measurement outcome $X'$, which at most a fraction $\zeta$ of them disagree according to the performed statistical test, then the smooth max entropy is bounded as follows:
\begin{equation}
    H^{\epsilon}_{max}(X|X') \leq nh(\zeta)
\end{equation}
where $h(.)$ denotes the classical binary Shannon entropy. Now we can use this result and our assumption of successful verification together. Given the assumption that the verification is passed with a probability $1 - \epsilon(m)$, and the verification algorithm consists of measuring the states in the $Z$ and $X$ bases, we can conclude that the final bits differ in at most a fraction $\zeta = \epsilon(m)$ where $\epsilon(m)$ is a negligible function. As a result, we have:
\begin{equation}\label{eq:max-entropy}
    H^{\epsilon}_{max}(S^m | \tilde{Y}^m) \leq mh(\zeta) \approx m \epsilon(m)
\end{equation}
Putting Equations~\eqref{eq:sec-criteria-uncert} and~\eqref{eq:max-entropy} together, we have:
\begin{equation}
    H^{\epsilon}_{min}(S^m | ER^m) \geq m  - m \epsilon(m) - m\log_2 (1 + 2\delta_r)
\end{equation}
On the right-hand side of the above inequality, the second term is still a negligible function, and the third term depends on the CPUF bias probability distribution. We assume the CPUF satisfies $p$-Randomness, as defined in the Definition~\ref{def:p-randomness}. Thus the $\delta_r$ is a small value, and hence the term $(1 + 2\delta_r)$ is negligibly close to $1$, which means that the third term, is negligibly close to $0$ in the security parameter which is $m$. Finally, we conclude that:
\begin{equation}
    H^{Eve}_{min} = H^{\epsilon}_{min}(S^m | ER^m) \geq m  - \epsilon'(m)
\end{equation}
where $\epsilon'(m)$ is a negligible function, and the proof is complete.
\end{proof}

\section{MUB in 8 dimensions}\label{ap:mub8}
In this section, we consider an HPUF encoded using an 8-dimensional state (3 qubits). We use the construction shown in \cite{TNO09} to compute 9 mutually unbiased bases for the 8-dimensional state $\ket{x^\theta}, x \in \{0,1\}^3, \theta \in \{0,1,2\}^2$, where $\theta$ represents the basis and $x$ represents the state. We denote the set of basis vectors for each basis using the matrices $B^{\theta}, \theta \in \{0,1,...8\}$. The column $B^{\theta}_j$ denotes the $j^{th}$ basis vector for the basis set $\theta$. The MUB set is given as:
\begin{equation}
\begin{split}
B    = &\{\mathbb{I}_8, \mathbf{O \otimes O \otimes O}, \mathbf{U(O \otimes O \otimes I)}, \mathbf{V(O \otimes I \otimes O)}, \\&\mathbf{W(O \otimes I \otimes I)}, \mathbf{W(I \otimes O \otimes O)}, \mathbf{V(I \otimes O \otimes I)}, \\&\mathbf{U(I \otimes I \otimes O)}, \mathbf{I \otimes I \otimes I} \}
\end{split}
\end{equation}
where $
\mathbf{O} = \frac{1}{\sqrt{2}}\begin{bmatrix} 1 & 1 \\ 1 & -1\end{bmatrix}, \mathbf{I} = \frac{1}{\sqrt{2}}\begin{bmatrix} 1 & 1 \\ i & -i\end{bmatrix},\\
\mathbf{U} = \text{diag}\{1,1,1,1,1,-1,-1,1\},\\
\mathbf{V} = \text{diag}\{1,1,1,-1,1,-1,1,1\},\\
\mathbf{W} = \text{diag}\{1,1,1,-1,1,1,-1,1\}\\
$

With predicting the 8-dimension qubit correctly by an adversary, the optimal strategy is to perform \emph{Split Attack} as shown previously on modelling bit by bit of $x=x_0x_1x_2$. In general, we have:
\begin{align}
    p_{dist}(\rho_0,\rho_1) = \max_{E}(\frac{1}{2}+\frac{1}{2}Tr(E(\rho_0-\rho_1))) 
\end{align}
the optimal probability of distinguishing two mixed states with POVM element $E$. For each mixed state $\rho_x = \frac{1}{9}\sum_{\theta = 0}^{8} |x^\theta\rangle\langle x^\theta|$, the optimal success probabilities of guessing $x_0, x_1$ and $x_2$ are given as (See \cite{hpuf} for more details):
\begin{align}
    \nonumber
    p_0 &= p_{guess}(x_0) = p_{dist}(\frac{1}{4}\sum_{i = 0}^3\rho_i,\frac{1}{4}\sum_{i = 4}^7\rho_i)\approx 0.62\\
    \nonumber
    p_1 &= p_{guess}(x_1|x_0)\\
    \nonumber
    &= \frac{1}{2}(p_{guess}(x_1|x_0 = 0) + p_{guess}(x_1|x_0 = 1))\\
    \nonumber
    &\leq p_{dist}(\frac{\rho_0+\rho_1}{2},\frac{\rho_2+\rho_3}{2})+p_{dist}(\frac{\rho_4+\rho_5}{2},\frac{\rho_6+\rho_7}{2})\\
    \nonumber
    &\approx 0.69\\
    \nonumber
    p_2 &= p_{guess}(x_2|x_0,x_1)\\
    \nonumber
    &=\frac{1}{4}\sum_{i,j \in \{0,1\}}p_{guess}(x_2| x_0 = i,x_1 = j)\\
    \nonumber
    &\leq\frac{p_{dist}(\rho_0,\rho_1)+p_{dist}(\rho_2,\rho_3)+p_{dist}(\rho_4,\rho_5)+p_{dist}(\rho_6,\rho_7)}{4}\\
    &\approx 0.77
\end{align}

The result gives us an upper bound on the probabilities, allowing us to fit this attack into our existing simulation framework easily while only giving more power to the adversary, i.e., in an actual scenario, the number of CRPs required to obtain an accurate model would be the same or more than in our simulations.
\end{appendices}

\end{document}